\documentclass[acmsmall]{acmart}
\setcopyright{none}
\usepackage{booktabs}   

\usepackage{amsmath}

\usepackage{mathpartir}
\usepackage{macros}
\usepackage{amssymb}
\usepackage{stmaryrd}
\usepackage{xspace}
\usepackage[caption=false]{subfig}
\usepackage{wrapfig}
\usepackage{graphicx}
\usepackage{tabularx}
\usepackage{array}
\usepackage{color}
\usepackage[T1]{fontenc}
\usepackage{bbold}
\usepackage{scalerel}
\usepackage{longtable}
\usepackage[utf8]{inputenc}

\usepackage{enumitem}
\usepackage{chngcntr}
\usepackage[super]{nth}

\makeatletter
\def\mdseries@tt{m}
\makeatother
 
\usepackage[frozencache=true]{minted}

\newminted{js}{
  linenos,
  firstnumber=last,
  xleftmargin=0.20in,
  numbersep=9pt,
  mathescape,
  fontsize=\footnotesize
}

\newminted{ts}{
  mathescape,
  xleftmargin=.25in,
  xrightmargin=.25in,
  escapeinside=@@,
  fontsize=\footnotesize
}

\newtheorem{assumption}{Assumption}

\usepackage{hyphenat}
\usepackage[title]{appendix}

\begin{document}

\title{Fast and Precise Type Checking for JavaScript}

\author{Avik Chaudhuri}
\affiliation{\institution{Facebook Inc.}\country{USA}}
\email{avik@fb.com}

\author{Panagiotis Vekris}
\affiliation{\institution{University of California, San Diego} \country{USA}}
\email{pvekris@cs.ucsd.edu}

\author{Sam Goldman}
\affiliation{\institution{Facebook Inc.}\country{USA}}
\email{samgoldman@fb.com}

\author{Marshall Roch}
\affiliation{\institution{Facebook Inc.}\country{USA}}
\email{mroch@fb.com}

\author{Gabriel Levi}
\affiliation{\institution{Facebook Inc.}\country{USA}}
\email{gabe@fb.com}


\begin{abstract}
In this paper we present the design and implementation of \flow, a fast and
precise type checker for \js that is used by thousands of developers on millions
of lines of code at \facebook every day.
\flow uses sophisticated type inference to understand common \js idioms
precisely. This helps it find non-trivial bugs in code and provide code
intelligence to editors\- without requiring significant rewriting or annotations
from the developer. We formalize an important fragment of \flow's analysis and
prove its soundness.
Furthermore, \flow uses aggressive parallelization and incrementalization to
deliver near-instantaneous response times. This helps it avoid introducing any
latency in the usual edit-refresh cycle of rapid JavaScript development. We
describe the algorithms and systems infrastructure that we built to scale
\flow's analysis.
\end{abstract}

\ccsdesc[500]{Theory of computation~Type structures}
\ccsdesc[500]{Theory of computation~Program analysis}

\keywords{Type Systems, Type Inference, JavaScript}

\maketitle

\section{Introduction}\label{sec:intro}

\js is one of the most popular languages for writing web and mobile applications
today. The language facilitates fast prototyping of ideas via dynamic
typing. The runtime provides the means for fast iteration on those ideas via
dynamic compilation. This fuels a fast edit-refresh cycle, which promises an
immersive coding experience that is quite appealing to creative developers.

However, evolving and growing a \js codebase is notoriously
challenging. Developers spend a lot of time debugging silly mistakes---like
mistyped property names, out-of-order arguments, references to missing values,
checks that never fail due to implicit conversions, and so on---and worse,
unraveling assumptions and guarantees in code written by others. In many other
languages, this overhead is mitigated by having a layer of types over the code
and building tools for the developer that use type information. For example,
types can be used to identify common bugs and to document interfaces of
libraries. Our aim is to bring such type-based tooling to \js.

\subsection{Goals}

In this paper, we present the design and implementation of \flow, a static type
checker for \js we have built and have been using at \facebook for the past
three years.

The idea of using types to manage code evolution and growth in \js (and related
languages) is not new. In fact, several useful type systems have been built for
\js in recent years. The design and implementation of \flow are driven by the
specific demands of real-world \js development we have observed at \facebook and
the industry at large.
\begin{itemize}
\item The type checker must be able to cover large parts of the codebase
  without requiring too many changes in the code. Developers want precise
  answers to code intelligence queries (the type of an expression, the
  definition reaching a reference, the set of possible completions at a
  point). Relatedly, they want to catch a large number of common bugs with
  few false positives.
\item The type checker must provide very fast responses, even on a very large
  codebase. Developers do not want any noticeable ``compile-time'' latency in
  their normal workflow, because that would defeat the whole purpose of using \js.
\end{itemize}
To meet these demands, we had to make careful choices and solve technical
challenges in \flow that go beyond related existing systems.
\begin{itemize}
\item We precisely model common \js idioms that appear pervasively in a modern
  \js codebase. For example, \flow understands the pattern \inlinejs{x = x || 0}
  that, \eg, initializes an optional parameter \inlinejs{x} in a function body.
  Handling a case like this necessitates support for type refinements:
  the system needs to recognize that the assigned value will not be \inlinejs{null} or \inlinejs{undefined},
  \ie, the type of the initial value of \inlinejs{x} must be \emph{refined} to exclude \emph{falsy} values.\footnote{The values \inlinejs{false},
  \inlinejs{0}, \inlinejs{""}, \inlinejs{null},
\inlinejs{undefined}, and \inlinejs{NaN} are \emph{falsy}. All other values are \emph{truthy}.}
  (More examples are shown below.)

\item At the same time, we \emph{do not} focus on reflection and legacy
  patterns that appear in a relatively small fraction (that is also usually
  stable and well-tested).
  Today, tools like Babel convert modern
  \js to (the more low-level) ES5 executed on browsers. \flow focuses on
  analyzing the source, instead of the target, of such translations (unlike many
  previous efforts that address ES5, or the even more low-level, and therefore harder, ES3).


\item We modularize our constraint-based analysis and implement data structures,
  algorithms, and systems infrastructure for parallel computing, shared-memory
  communication, and incremental updates to scale to millions of lines of code
  while answering most queries in well under a second.
\end{itemize}

\subsection{Overview}\label{subsec:overview}

We now introduce the main ideas behind \flow's design and implementation. (A
full description of \flow is not possible due to space constraints.)

\paragraph{Precise type checking}

One of the main contributors of \flow's precision is path-sensitivity: the way
types interact with runtime tests. The essence of many \js idioms is to put
together ad hoc sets of runtime values and to take them apart with shallow,
structural (in)equality checks. In \flow, the set of runtime values that a
variable may contain is described by its type, and a runtime test on that
variable \emph{refines} the type to a smaller set. This ability turns out to be
quite powerful and general in practice.

In this paper, we formalize refinements in a core subset of \js. The system is
particularly interesting because of \emph{the combination of mutable local
  variables and closures that capture them by reference}. Next, we illustrate
this system via a series of examples (\figref{fig:intro:code}).

\begin{figure}
\footnotesize
\begin{minipage}[l]{.50\textwidth}
\begin{jscode*}{escapeinside=++}
+\label{code:01}+function pipe(x, f) { f(x); }
+\label{code:05}+var hello = (s) => console.log("hello", s);
+\label{code:07}+pipe("world", hello);
+\label{code:09}+pipe("hello", null); // error
\end{jscode*}
\vspace{1em}
\begin{jscode*}{escapeinside=++}
+\label{code:11}+function pipe(x, f) {
+\label{code:12}+  if (f != null) { f(x); } // ok
+\label{code:15}+}
\end{jscode*}
\vspace{1em}
\begin{jscode*}{escapeinside=++}
+\label{code:17}+var nil = { kind: "nil" };
+\label{code:18}+var cons = (head, tail) => {
+\label{code:19}+  return { kind: "cons", head, tail };
+\label{code:20}+}
\end{jscode*}
\end{minipage}
\hfill
\begin{minipage}[r]{.48\textwidth}
\begin{jscode*}{escapeinside=--}
-\label{code:22}-function sum(list) {
-\label{code:23}-  if (list.kind === "cons") {
-\label{code:24}-    return list.head + sum(list.tail); // ok
-\label{code:25}-  }
-\label{code:26}-  return 0;
-\label{code:27}-}
-\label{code:29}-sum(cons(6, cons(7, nil)));
\end{jscode*}
\vspace{1em}
\begin{jscode*}{escapeinside=--}
-\label{code:31}-function merge(x) {
-\label{code:32}-  x = x || nil;
-\label{code:33}-  return x.kind; // ok
-\label{code:34}-}
\end{jscode*}
\vspace{1em}
\begin{jscode*}{escapeinside=--}
-\label{code:36}-function havoc(x) {
-\label{code:37}-  function reset() { x = null; }
-\label{code:38}-  x = x || nil;
-\label{code:39}-  reset();
-\label{code:40}-  return x.kind; // error
-\label{code:41}-}
\end{jscode*}
\end{minipage}
\caption{Modern \js Examples}
\label{fig:intro:code}
\end{figure}


Higher-order functions (lines~\ref{code:01}--\ref{code:09}) are quite common in \js.
Unfortunately, it is also common to use \inlinejs{null} as a default for
everything (\mylineref{code:09}). In particular, this causes the dreaded ``\emph{null is not a
function}'' error to hit often.
Fortunately, \flow finds these errors by
following flows of \inlinejs{null} to calls in the code.

Checking for nullability is the idiomatic way to prevent such errors at
runtime. (In \js, the check \inlinejs{f != null} is equivalent to \inlinejs{f
  !== null && f !== undefined}, which additionally rules out
\inlinejs{undefined}, commonly used to denote missing values.)  Thankfully \flow
understands that this code is safe. It refines the type of \inlinejs{f} to
filter out \inlinejs{null} in \mylineref{code:12}, and thus knows that
\inlinejs{null} cannot reach the call. Many other idiomatic variants also work,
such as \inlinejs{f && f(x)}, where \inlinejs{f} is checked to be truthy
(ruling out \inlinejs{null}, \inlinejs{undefined} and other falsy values)
before calling.

Refinements also power a common technique to encode algebraic data types in \js,
which are used quite widely (to manage actions and dispatchers, data and
queries, \etc in user interface libraries). Records of different shapes have a
common property that specifies the ``constructor,'' and other properties
dependent on the constructor value. These records are then analyzed by ``pattern
matching''---inspecting and branching on the constructor value.

For example, consider the encoding of lists in lines~\ref{code:17}--\ref{code:20}.
A \inlinejs{sum} function (\mylineref{code:22})
checks whether a list is non-empty before accessing properties specific to
non-empty lists.
Following the calls to \inlinejs{sum}, \flow knows that the parameter
\inlinejs{list} in \mylineref{code:22} can contain both kinds of objects---those whose
\inlinejs{kind} property is \inlinejs{"cons"}, and those for which it is
\inlinejs{"nil"}. The latter ones are filtered out by refining the type of
\inlinejs{list} in \mylineref{code:23}, so that the only objects reaching the property
accesses of \inlinejs{head} and \inlinejs{tail} in \mylineref{code:24} are guaranteed to
have those properties. Thus, \flow knows that this code is safe. Without
refinements, on the other hand, the analysis would have over-conservatively
concluded that \inlinejs{nil} can also flow to the property accesses, leading to
spurious type errors.


Refinements are tracked by a flow-sensitive analysis, and interact in
interesting ways with variable assignment. The common idiom in \mylineref{code:32}
of \inlinejs{merge} ensures
that a variable has a non-\inlinejs{null} default.
\flow models the assignment by merging the refined type of \inlinejs{x} with the
type of \inlinejs{nil} and updating the type of \inlinejs{x} with it.


On the other hand, refinements can be \emph{invalidated} by assignments, which
can even happen indirectly via calls (\mylineref{code:39}).
While invalidating refinements is necessary for soundness, they should be
preserved as much as possible to avoid spurious type errors. \flow tracks
variable assignments as \emph{effects} for precise invalidation.

Refinements and their invalidation carry over to higher-order functions. We also
have limited support for refining mutable object properties, but those
refinements are invalidated aggressively (\ie, our analysis is not
heap-sensitive).

Behind the scenes, \flow relies on set-based analysis as a common low-level
``assembly language'' for encoding a wide variety of high-level
analyses. Compared with pure unification, this affords far more precision, but
is much less efficient (quasi-cubic vs. quasi-linear in program size). How do we
scale such an analysis?

\paragraph{Fast type checking}

The key to \flow's speed is modularity: the ability to break the analysis into
file-sized chunks that can be assembled later.
Fortunately, \js is already written using files as modules, so we modularize our
analysis simply by asking that modules have explicitly typed signatures. (We
still infer types for the vast majority of code ``local'' to modules.)
Coincidentally, developers consider this good software engineering practice
anyway.

With modularity, we can aggressively parallelize our analysis. Furthermore, when
files change, we can incrementally re-analyze only those files that depend on
the changed files, and avoid re-analysis when their typed signatures have not
changed. Together, these choices have helped scale the analysis to millions of
lines of code.

Under the hood, \flow relies on a high-throughput low-latency systems
infrastructure that enables distribution of tasks among parallel workers, and
communication of results in parallel via shared memory. Combined with an
architecture where the analysis of a codebase is updated automatically in the
background on file system changes, \flow delivers near-instantaneous feedback as
the developer edits and rebases code, even in a large repository.

\subsection{Contributions}

Overall, this paper (and our work on \flow) shows how, through careful design
and implementation, type checking for \js can be both extremely fast and precise
enough in practice. We make the following contributions.

\begin{enumerate}
\item We identify a lightweight form of type refinement as a crucial feature for
  supporting a variety of common \js idioms in practice. We formalize inference
  to support refinements in a core fragment of \js containing higher-order
  functions, mutable variables, runtime tests, and capture-by-reference. Knowing
  when to invalidate a refinement can be quite tricky in this setting
  (Sections \ref{sec:cons:gen} and \ref{sec:cons:prop}).
  To the best of our knowledge, no prior work has formalized this subset of
  features in a type inference for JavaScript.
  We also discuss our implementation of type inference as a system of set-based
  constraints combined with unification for optimization
  (\secref{sec:inf:impl}).
  We define a runtime semantics (\secref{sec:semantics}) and
  prove our system safe (\secref{sec:type:soundness}) with
  respect to it.
\item We show how our inference system can be extended to check type
  annotations. Union types introduce some interesting complications
  (\secref{sec:annotations}). We then identify a lightweight way to use type
  annotations to modularize our analysis, piggybacking on existing best
  practices in \js development such as breaking a codebase into several small
  modules and documenting types at their boundaries (\secref{sec:modules-deps}).
\item We show how we can exploit modularity and dependency management to make
  \flow responsive at scale. We describe an algorithm for parallelizing
  and incrementally updating the analysis when files change
  (\secref{sec:incremental}). We describe how we extend an existing system
  infrastructure for parallel computation and communication, and implement our
  algorithms on it to achieve high throughput and low latency
  (\secref{sec:parallel}).
\item We perform a thorough experimental evaluation of \flow on a codebase with
  millions of lines of \js. Through key metrics we demonstrate the behavior of
  various stages of type checking at scale,
  and validate our hypotheses on the precision gained by refinement tracking
              and our low annotation footprint.
  (\secref{sec:experiments}).
\end{enumerate}

We conclude with a discussion of related work
(\secref{sec:related}) and of limitations and threats to
validity (\secref{sec:limitations}).

\section{Language \flowcore}\label{sec:core}

We consider a minimal subset of \js that includes functions,
mutable variables, primitive values and
records.
Notably, we leave out data structures like dictionaries and
arrays, as well as object-oriented features like
\inlinejs{this},
methods, classes, and inheritance. These parts of the
language are mostly orthogonal to understanding refinements.
Their type inference, while interesting, is built on the same foundations,
and behave more or less similarly to
previous work---we can safely extend our model to include
them, without significantly complicating our guarantees.
%
Our focus is on formalizing type inference and refinement
strengthening, with the exception of refinements on mutable
fields that are not tracked through the heap.
While compact, this fragment is expressive enough to model the examples of
Section~\ref{sec:intro}---which illustrate how \flow uses predicate refinements
to reduce false positives, while remain sound with respect to variable updates.


\subsection{Syntax}\label{subsec:syntax}

\figref{fig:flow:expr} describes the language of
\emph{expressions} \expr and \emph{statements} \stmt.
Here, \const represents constants, and \evar and \evarb
range over program variables.

\begin{figure}[t]
\begin{langdefsmall}
\langline{\expr}{\production}{
     \evar
\sep \vconst
\sep \assign{\evar}{\expr}
\sep \arrow{\evar}{\stmt}{\expr}
\sep \efuncall{\expr_1}{\expr_2}
}{\textbf{Expressions}}
\\
\langline{}{\sep}{
     \objlit{\fieldsym_1}{\expr_1}{\fieldsym_n}{\expr_n}
\sep \fieldread{\expr}{\fieldsym}
\sep \fieldwrite{\expr_1}{\fieldsym}{\expr_2}
}{}
\\
\langline{}{\sep}{
     \predof{\evar}
\sep \binand{\expr_1}{\expr_2}
\sep \binor{\expr_1}{\expr_2}
\sep \unaryneg{\expr}
}{}
\vsep
\\
\langline{\stmt}{\production}{
     \expr
\sep \varassign{\evar}{\expr}
\sep \ite{\expr}{\stmt_1}{\stmt_2}
\sep \seq{\stmt_1}{\stmt_2}
\sep \skipstmt
}{\textbf{Statements}}
\end{langdefsmall}
\caption{\flowcore Syntax}
\label{fig:flow:expr}
\end{figure}

\paragraph{Expressions}
%
%
%
We elide primitive operations (which may include
arithmetic operations). Constants include, \eg, numbers, strings, and
\inlinejs{undefined}.
The syntax \predof{\evar} draws from a fixed, possibly
infinite set of unary \emph{predicates} \primpred on \evar. These
model dynamic checks, such as
\inlinejs{typeof x === "number"},
\inlinejs{x === undefined},
\inlinejs{x} (testing if an expression is \emph{truthy}),
or model tests like \inlinejs{x.f === "nil"} on
records.  Note that in this system the last check
does not imply a predicate on the value of \inlinejs{x.f},
but rather on \inlinejs{x} itself.
The former would be a heap refinement, which
\flow only supports in a limited fashion, and which is
excluded from the formalism.

General-purpose functions (using the keyword \functionkw)
are complicated in JavaScript: they can be additionally used
as methods and as constructors. To simplify our exposition,
we restrict our attention to \emph{arrow} functions
(essentially lambdas). We assume that a function body
consists of a statement followed by the return of an
expression. Functions that do not explicitly return anything
can be thought of as implicitly returning \inlinejs{undefined}.
(\flow's
treatment of abnormal control flows via \inlinejs{return} is
also interesting, but we omit it here.)
We also include the logical conjunction (\code{\&\&}),
disjunction (\code{||}) and negation (\code{!}) operators, as
they are pervasive in \js and inform our refinement
strategy.

\paragraph{Statements}
We use \inlinejs{var} to introduce variables, and include statements
for conditional execution and sequencing.
We omit \inlinejs{const} because it is
much simpler than \inlinejs{var}, since refinements never need to be invalidated.
We also omit \inlinejs{let}-bound variables. Their main difference with
\inlinejs{var}-bound variables is in scoping rules, so handling them does
not add any insight to our type system overview.
Finally, we omit \inlinejs{while}; although it can be encoded with
\inlinejs{if} and recursion, \flow's treatment of it is more precise.
%

\paragraph{Assumption}
We assume an $\alpha$-renaming pre-pass over the program's AST that
guarantees that each variable definition point
(which is either a \inlinejs{var} statement or an arrow definition)
introduces a unique variable identifier.
This is a fairly straightforward transformation for any
preprocessor that helps avoid unintentional capture of
variables in exported closures.


\subsection{Types, Effects and Constraints}\label{sec:types}

\begin{figure}[t]
%
%
\begin{langdefsmall}
\langline{\tvar, \tvarb, \tvarc, \tvard}{\inclusion}{\tvarset}{\textbf{Type Vars}}
\vsep
\\
\langline{\tlit}{\production}{
     \tbase
\sep \tarrow{\tvar}{\effect}{\type}
\sep \tobj{\fieldsym_1}{\tvar_1}{\fieldsym_n}{\tvar_n}
}{\textbf{Type Literals}}
\vsep
\\
\langline{\type}{\production}{
     \tlit
\sep \tjoin{\type_1}{\type_2}
\sep \tvar
}{\textbf{Types}}
\vsep
\\
\langline{\effvar}{\inclusion}{\effvarset}{\textbf{Effect Vars}}
\vsep
\\
\langline{\efflit}{\production}{\effempty \sep \evar}{\textbf{Effect Literals}}
\vsep
\\
\langline{\effect}{\production}{
     \efflit
\sep \ejoin{\effect_1}{\effect_2}
\sep \effvar
}{\textbf{Effects}}
\vsep
\\
\langline{\use_{\type}}{\production}{
\tvar
\sep \calluse{\type}{\effvar}{\tvar}
\sep \getuse{\fieldsym}{\tvar}
\sep \setuse{\fieldsym}{\type}
\sep \preduse{\pred}{\tvar}
}{\textbf{Type Uses}}
\vsep
\\
\langline{\use_{\effect}}{\production}{
\effvar
\sep \havocuse{\env}
}{\textbf{Effect Uses}}
\vsep
\\
\langline{\pred}{\production}{
\primpred \sep \logneg{\primpred}
}{\textbf{Predicates}}
\vsep
\\
\langline{\cons}{\production}{
\flows{\type}{\use_{\type}}
\sep \flows{\effect}{\use_{\effect}}
}{\textbf{Constraints}}
\end{langdefsmall}
\caption{\flowcore Type, Effect and Constraint Syntax}
\label{fig:flow:constraint}
\label{fig:flow:types}
\end{figure}

The basic ingredients of our constraint system are \emph{types} \type and
\emph{effects} \effect. Their syntax is described in \figref{fig:flow:types}.
%

\paragraph{Types}
Types are ranged over by variables
\tvar, \tvarb, \etc taken from an enumerable set \tvarset.
The building blocks for constructing complex type structures
are \emph{type literals} \tlit.
%
%
These include primitive types \tbase
(\eg, \tnum, \tstring, and \tvoid for \inlinejs{undefined}),
arrow types \tarrow{\tvar}{\effect}{\type} for functions,
and record types \tobj{\fieldsym_1}{\tvar_1}{\fieldsym_n}{\tvar_n}.
%
Arrow types are annotated with an effect \effect which
describes a set of names \evar that may be assigned in the
function's body or transitively in code that is executed
when calling this function.  A more proper introduction of
effects follows.
Types also feature a binary operator \tjoin{}{} denoting the union of types.

\paragraph{Effects}
The effect we are interested in tracking here is variable
updates. Each language term is associated with an effect, as
we will see later in constraint generation. This is
(roughly) the set of variables that are (re)assigned within
this term.
The base constructors of effects are the empty effect
\effempty and variable symbols \evar, corresponding to the
variables that are updated.
Like types, effects also feature a binary operator \tjoin{}{} denoting union of effects.
Finally, effects are ranged over by variables \effvar
taken from an enumerable set \effvarset.

\paragraph{Environments}
An \emph{environment} \env binds variables \evar to entries
\entry{\type}{\tvar}, meaning that its most recent
assignment was of type \type, whereas the type variable
\tvar is used as the collective summary for all its (past,
current, and future) assignments.
Here \type is flow-sensitive---its value may change from one (flow-sensitive)
environment to another---whereas \tvar is invariant.  We also distinguish
between environment \emph{extension}---denoted
\envext{\env}{\evar}{\entry{\type}{\tvar}} (variable \evar is not bound in the
original environment \env), and environment \emph{update}---denoted
\envupd{\env}{\evar}{\entry{\type}{\tvar}} (variable \evar was bound in \env).
%
%

\paragraph{Predicates}
Key to our type refining process is the notion of predicates.
A predicate \pred is a
clause denoting a property of its implied argument.
In our setting, syntactically it can be a \emph{base}
predicate \primpred or its negation.
Base predicates describe properties of constructed or
primitive types.  For the remaining sections we will keep
these predicates abstract, but examples of these predicates
are the ones implied by checks of the form
$\inlinejs{typeof }\star\inlinejs{ === "string"}$,
$\inlinejs{typeof }\star\inlinejs{ === "number"}$,
$\star\inlinejs{.f === "null"}$, \etc, where $\star$ is to be
filled in with a
program variable.

%

\paragraph{Constraints}
A \emph{constraint} \cons is a ``flow'' from a type \type
(resp. effect \effect) to a type \emph{use} $\use_{\type}$
(resp. an effect \emph{use} $\use_{\effect}$).
Flows from types to type uses generalize the notion of subtyping.
However, we chose to enforce some structural restrictions to
the forms that can appear on the right-hand side of
constraints, namely the uses.
Type and effect variables can appear as uses themselves.
We do not allow general types and effects to appear
as uses.
Instead they are wrapped by constructors that contain
information about the operations that caused the generation of such
constraints.
Uses account for data flow through function calls (\callusesym), object
operations (\getusesym, \setusesym), control flow refinement (\predusesym), and
refinement invalidation (\havocusesym).

The use \calluse{\type}{\effvar}{\tvar} corresponds to a function call with
argument type $\type$, resulting in type $\tvar$; the effect variable \effvar
models the effect of the target function. A constraint
\flows{\type'}{\calluse{\type}{\effvar}{\tvar}} looks up the parameter, return, and effect of $\type'$
and propagates \type to the parameter, the return to \tvar, and the effect to \effvar.
The uses for reading and writing to a field, \getuse{\fieldsym}{\tvar} and
\setuse{\fieldsym}{\type}, are straightforward.  A constraint
\flows{\type'}{\getuse{\fieldsym}{\tvar}}
(resp. \flows{\type'}{\setuse{\fieldsym}{\type}}) looks up the field \fieldsym of
$\type'$ and propagates the result to \tvar (resp. \type to the result).
The use \preduse{\pred}{\tvar} is used to
\emph{refine} an incoming type using \emph{predicate} \pred,
resulting in fresh type variable \tvar.
In other words, a constraint
\flows{\type}{\preduse{\pred}{\tvar}} will only allow the
parts of $\type$ that satisfy \pred to flow to \tvar.
Finally, for refinement invalidation we introduce \havocusesym, which takes an environment argument
\env. A constraint \flows{\evar}{\havocuse{\env}} says that the variable \evar
may be updated, so that any refinement involving \evar in \env must be
invalidated. This will be discussed later on in greater detail
(\secref{sec:cons:prop}).
%


\section{Constraint System}
\label{sec:cons:sys}

We present the static semantics of our formal fragment by means of a constraint
generating type inference scheme. Our constraints encode type safety obligations
that arise as values flow to operations throughout the program.


\subsection{Constraint Generation}\label{sec:cons:gen}

The core type inference judgments for expressions and
statements in \flowcore are:
\begin{align*}
\jexpr{\env}{\expr}{\type}{\effect}{\predmap}{\cset}{\env'}
&&
\jstmt{\env}{\stmt}{\effect}{\cset}{\env'}
\end{align*}
The derivation of a judgment relies on a set of constraints
\cset as proof obligations, which appear on the right of the
$\triangleright$ symbol.
We use $\cup$ for the union of two constraint sets and $\bigcup$ for
the union of a number of constraint sets ranged over by the
index in the subscript of the arguments set (\eg, $\bigcup\nolimits_{i=1}^{k}\cset_i$).
For both expressions and statements this judgment is
\emph{flow-sensitive} which is achieved by introducing an
output environment $\env'$, in addition to the input
environment \env.
The set of variable names assigned in \expr or \stmt is modeled by
\effect.
The case of expressions has two additional byproducts: a
\emph{type} \type and a \emph{predicate mapping} \predmap.
The latter includes bindings from names to predicates that
must hold when \expr is truthy, and symbolic operations over
them (explained later):

\begin{langdefinline}
  \langline{\predmap}{\production}{\predmapempty}{empty mapping}
  \\
  \langline{}{\sep}{\singletonpredmap{\evar}{\pred}}{variable binding}
  \\
  \langline{}{\sep}{\predmap_1 \wedge \predmap_2}{conjunction}
  \\
  \langline{}{\sep}{\predmap_1 \vee \predmap_2}{disjunction}
  \\
  \langline{}{\sep}{\neg \predmap}{negation}
  \\
  \langline{}{\sep}{\predmap \backslash \effect}{exclude effect}
\end{langdefinline}

Below we describe constraint generation in more detail,
starting from rules handling variables,
functions, and calls
(\figref{fig:flow:cons:gen:expr}).

\begin{figure}[t]
\judgementHead{Expression Constraint Generation}{\jexpr{\env}{\expr}{\type}{\effect}{\predmap}{\cset}{\env'}}
\begin{mathpar}
\inferruleright{\cgconst\label{rule:cgconst}}{
  }{
  \jexpr{\env}{\const}{\tbase_{\const}}{\effempty}{\predmapempty}{\csetemp}{\env}
}
\and
\inferruleright{\cgvar\label{rule:cgvar}}{
  \idxis{\env}{\evar}{\entry{\type}{\tvar}}
}{
\jexpr{\env}{\evar}{\type}{\effempty}{\singletonpredmap{\evar}{\truthy}}{\csetemp}{\env}
}
\and
\inferruleright{\cgassign\label{rule:cgassign}}{
  \jexpr{\env}{\expr}{\type}{\effect}{\predmap}{\cset}{\env'}
  \\
  \idxis{\env'}{\evar}{\entry{\type_0}{\tvar}}
}{
  \jexpr{\env}{\assign{\evar}{\expr}}{\type}
    {\effectconcat{\effect}{\evar}}
    {\logforget{\predmap}{\evar}}
    {\consconcat{\cset}{\uniflow{\type}{\tvar}}}
    {\envupd{\env'}{\evar}{\entry{\type}{\tvar}}}
}
\and
\inferruleright{\cgfun\label{rule:cgfun}}{
  \is{\many{\evar_i}}{\locals{\stmt}}
  \\
  \fresh{\tvar, \many{\tvar_i}}
  \\
  \is{\env_1}{\envexts{\envext{\erase{\env}}{\evar}{\initarg{\tvar}}}{\evar_i}{\entry{\tvoid}{\tvar_i}}}
  \\
  \jbody{\env_1}{\body{\stmt}{\expr}}{\type}{\effect}{\cset}{\env_2}
}{
  \jexpr{\env}{\arrow{\evar}{\stmt}{\expr}}
  {\tarrow{\tvar}{\logforget{\effect}{\evar,\many{\evar_i}}}{\type}}
    {\effempty}
    {\predmapempty}
    {\cset}{\env}
}
\and
\inferruleright{\cgcall\label{rule:cgcall}}{
  \jexpr{\env}{\expr_1}{\type_1}{\effect_1}{\predmap_1}{\cset_1}{\env_1}
  \\
  \jexpr{\env_1}{\expr_2}{\type_2}{\effect_2}{\predmap_2}{\cset_2}{\env_2}
  \\
  \fresh{\tvar, \effvar}
  \\
  \widen{\env_2}{\env_3}{\cset_3}
  \\
  \is{\effectconcatthree{\effect_1}{\effect_2}{\effvar}}{\effect}
  \\
  \is{\consconcatfour{\cset_1}{\cset_2}{\cset_3}
    {\mkset{\concat{\flows{\effvar}{\havocuse{\env_3}}}
      {\flows{\type_1}{\calluse{\type_2}{\effvar}{\tvar}}}}
    }
  }{\cset}
}{
  \jexpr{\env}{\efuncall{\expr_1}{\expr_2}}{\tvar}{\effect}{\predmapempty}{\cset}{\env_3}
}
\end{mathpar}
\caption{Expression Constraint Generation in \flowcore (Variables and Functions)}
\label{fig:flow:cons:gen:expr}
\end{figure}


\paragraph{Variables}
The rules for reading and assigning a local variable
(\hyperref[rule:cgvar]{\cgvar} and
\hyperref[rule:cgassign]{\cgassign}) involve looking up
and updating the current type for the variable in the
outgoing environment.
This part is what makes this system
\emph{flow-sensitive}.
A flow-insensitive system would use a single environment for
each judgment.
The assigned type would be merged to the same
type used for the variable under update in the first place,
making it less precise.
In addition, reading a variable introduces a truthy
predicate on it. This is useful under specific contexts
such as when the variable is used as the condition part of an if-branch.
Conversely, writing a variable forgets any
refinement coming from expression \expr that concerns \evar.

\paragraph{Arrow Functions}
Rule~\hyperref[rule:cgfun]{\cgfun} handles arrow
functions by approximating the environment at the beginning
with the \emph{flow-insensitive erasure} of the current
environment (since we do not know where this function will
be called).
The meta-function \erasesym computes this new environment by
mapping each \envbinding{\evar}{\entry{\type}{\tvar}} to
\envbinding{\evar}{\entry{\tvar}{\tvar}}
(\figref{fig:flow:aux:func}).
In addition, to capture the hoisting of variables defined
within the scope of the function to the beginning of the
function body, we introduce the meta-function \localssym
that takes as argument a statement \stmt and returns
all variable identifiers $\many{\evar_i}$ declared in \stmt.
Each variable $\evar_i$ is bound to the undefined type \tvoid
(and a fresh general type $\tvar_i$), since its definition is hoisted to
the top of the function body and initialized to
\inlinejs{undefined}.
The inferred arrow type carries the effect \effect of the body of
the function. Note that we are removing the formal parameter \evar
and local variables $\many{\evar_i}$
from the effect as they are only visible within the body of the defined
arrow.


\begin{figure}[t!]
\judgementHeadTwo{Auxiliary Functions}{\eraseis{\env}{\env'}}{\widen{\env}{\env'}{\cset}}
\begin{mathpar}
\inferruleright{\terasee}{
}{
  \eraseis{\envempty}{\envempty}
}
\and
\inferruleright{\terasec}{
  \eraseis{\env}{\env'}
}{
  \eraseis{\envext{\env}{\evar}{\entry{\type}{\tvar}}}
     {\envext{\env'}{\evar}{\entry{\tvar}{\tvar}}}
}
\\
\inferruleright{\twidene\label{rule:twidene}}{
}{
  \widen{\envempty}{\envempty}{\csetemp}
}
\and
\inferruleright{\twidenc\label{rule:twidenc}}{
  \widen{\env}{\env_0}{\cset_0}
  \\
  \fresh{\tvarb}
}{
  \widen{\envext{\env}{\evar}{\entry{\type}{\tvar}}}
     {\envext{\env_0}{\evar}{\entry{\tvarb}{\tvar}}}
     {\consconcat{\cset_0}
      {\mkset{\concat{\flows{\type}{\tvarb}}
             {\flows{\tvarb}{\tvar}}}}}
}
\end{mathpar}
\vspace{0.5em}
\\
\judgementHeadTwo{Auxiliary Environment Operations}{\jenvjoin{\env_1}{\env_2}{\env}}
  {\jenvstrengthen{\env}{\predmap}{\env'}{\cset}}
\begin{mathpar}
\inferruleright{\envjoine\label{rule:envjoine}}{
}{
  \jenvjoin{\envempty}{\envempty}{\envempty}
}
\and
%
\inferruleright{\envjoinc\label{rule:envjoinc}}{
  \jenvjoin{\env_1}{\env_2}{\env}
}{
  \jenvjoin
    {\parens{\envext{\env_1}{\evar}{\entry{\type_1}{\tvar}}}}
    {\parens{\envext{\env_2}{\evar}{\entry{\type_2}{\tvar}}}}
    {\envext{\env}{\evar}{\entry{\parens{\tjoin{\type_1}{\type_2}}}{\tvar}}}
}
\and
\inferruleright{\refemp\label{rule:refemp}}{
}{
  \jenvstrengthen{\env}{\predmapempty}{\env}{\csetemp}
}
\and
\inferruleright{\refsingle\label{rule:refsingle}}{
  \idxis{\env}{\evar}{\entry{\type}{\tvar}}
  \\
  \fresh{\tvarb}
}{
  \jenvstrengthen{\env}
    {\singletonpredmap{\evar}{\pred}}
    {\envupd{\env}{\evar}{\entry{\tvarb}{\tvar}}}
    {\uniflow{\type}{\preduse{\pred}{\tvarb}}}
}
\and
\inferruleright{\refand\label{rule:refand}}{
  \jenvstrengthen{\env}{\predmap_1}{\env_1}{\cset_1}
  \\
  \jenvstrengthen{\env_1}{\predmap_2}{\env_2}{\cset_2}
}{
  \jenvstrengthen{\env}
    {\parens{\predmapand{\predmap_1}{\predmap_2}}}
    {\env_2}{\consconcat{\cset_1}{\cset_2}}
}
\and
\inferruleright{\refor\label{rule:refor}}{
  \jenvstrengthen{\env}{\predmap_1}{\env_1}{\cset_1}
  \\
  \jenvstrengthen{\env}{\predmap_2}{\env_2}{\cset_2}
}{
\jenvstrengthen{\env}{\parens{\predmapor{\predmap_1}{\predmap_2}}}{\envjoin{\env_1}{\env_2}}
  {\consconcat{\cset_1}{\cset_2}} 
}
\and
\inferruleright{\refeffect\label{rule:refeffect}}{
  \jenvstrengthen{\env}{\predmap}{\env_1}{\cset_1}
  \\
  \widen{\env_1}{\env_2}{\cset_2}
  \\
  \is{\env_3}{\setcompr{\envbinding{\evar}{\entry{\tvarb}{\type}}}
    {\concat{\inset{\envbinding{\evar}{\entry{\type}{\tvar}}}{\env}}
      {\inset{\envbinding{\evar}{\entry{\tvarb}{\tvar}}}{\env_2}}}}
}{
  \jenvstrengthen{\env}{\logforget{\predmap}{\effect}}{\env_2}
  {\consconcatthree{\cset_1}{\cset_2}{\uniflow{\effect}{\havocuse{
    \env_3 
  }}}}
}
\end{mathpar}
\caption{Auxiliary Functions for Function Logistics and Environment Operations in \flowcore}
\label{fig:flow:aux:func}
\label{fig:flow:env:aux}
\end{figure}

\paragraph{Calls}
Rule~\hyperref[rule:cgcall]{\cgcall} handles calls. We
approximate the outgoing environment with a
\emph{flow-sensitive widening} of the current environment
(instead of pessimistically erasing everything in scope).
The meta-function \endsym (\figref{fig:flow:aux:func}) computes this new
environment $\env'$ by mapping each \envbinding{\evar}{\entry{\type}{\tvar}} to
\envbinding{\evar}{\entry{\tvarb}{\tvar}} where \tvarb is a fresh type variable
such that \flows{\type}{\flows{\tvarb}{\tvar}}.  For any variable \evar that
gets assigned during the function call, we must fall back to its erasure, \ie, we
must flow \tvar to \tvarb.  For now this is achieved by flowing the effect
\effvar of the call to \havocuse{\env'}. The actual erasure happens later at
constraint propagation~(\secref{sec:cons:prop}), when the type of the receiver
function is known and the incoming effect is no longer abstract.
As we show in \secref{sec:cons:prop}, when a function
type flows to \calluse{\type_2}{\effvar}{\tvar}, the effect
\effvar is instantiated with the actual effect variables
\evar carried over by the incoming function type. These
variables trigger the erasure.

\paragraph{Environment Operations}

Before delving into the remaining typing rules, we introduce
some operations on environments
(\figref{fig:flow:env:aux}).

Ther first one is \emph{environment join} (\join),
a commutative operator that
computes the \emph{least upper bound} of a pair of
environments with the same domain.
Type entries bound to the same symbol in the
input environments need to refer to the same program
variable. This requirement allows us to assume that the general
type of a variable \evar bound in both environments will be
the same.
%
%
%

The next operation we define is \emph{environment
refinement} ($::$).
The semantics of a refinement \predmap is defined by how it
refines environments through the constraint-producing judgment
\jenvstrengthen{\env}{\predmap}{\env'}{\cset}, where an
environment \env is strengthened by the predicates in
\predmap and result in an environment $\env'$, potentially
including fresh variables that are constrained in \cset.
When \predmap is \singletonpredmap{\evar}{\pred}, we update the
relevant binding in the environment \env to a fresh type
\tvarb that is the result of the predicate refinement of the
initial type \type with \pred
(Rule~\hyperref[rule:refsingle]{\refsingle}).
The rules that handle the typical logical operators
(\hyperref[rule:refand]{\refand} and
\hyperref[rule:refor]{\refor}) are straightforward.

Refinements can be invalidated by effects.
In Rule~\hyperref[rule:refeffect]{\refeffect}, we first
refine \env by \predmap, and then apply the effect \effect
through the ``havoc'' mechanism on the resulting environment
$\env_1$.
There is a slight discrepancy in the way this mechanism is
applied in this case compared to function calls, since we only want to revert the
effect of the refinement caused by \predmap, and not fall
back to the most general type.
If ``havoc'' is triggered, then for every variable \evar
bound in $\env_3$, that happens to reach effect \effect, we
only flow type \type (that \evar was bound to in \env before
the refinement) to \tvarb, instead of the most general type
\tvar.
It appears here that we are locally breaking our invariant on the form of
environments, by allowing entries with types \type in the place of the most
general type summary (exponent).  This is a benign violation of our restriction
on environments since the constructed environment $\env_3$ is not used as the
input environment in a type inference judgment, but rather as the argument to
the ``havoc'' use. As we will see later, this context does not produce any
flows towards \type. The use of an environment here is in fact a mere syntactic convenience.

%

Finally, we can have refinements with logical connectives.
The negation of \singletonpredmap{\evar}{\primpred} is simply
\singletonpredmap{\evar}{\neg\primpred}.  Otherwise, we push
negations inward as much as possible, by applying the
laws:
\begin{align*}
  \neg(\predmap_1 \wedge \predmap_2) & \equals \neg \predmap_1 \vee \neg \predmap_2
                                     & \neg(\predmap \backslash \effect) & \equals \neg \predmap \backslash \effect
  \\
  \neg(\predmap_1 \vee \predmap_2) & \equals \neg \predmap_1 \wedge \neg \predmap_2
                                   & \neg(\neg \predmap) & \equals \predmap
\end{align*}

%
%
%

\paragraph{Logical operations}
\begin{figure}[t]
\judgementHead{Expression Constraint Generation}{\jexpr{\env}{\expr}{\type}{\effect}{\predmap}{\cset}{\env'}}
\begin{mathpar}
\inferruleright{\cgand\label{rule:cgand}}{
  \jexpr{\env}{\expr_1}{\type_1}{\effect_1}{\predmap_1}{\cset_1}{\env_1}
  \\
  \jenvstrengthen{\env_1}{\predmap_1}{\env_1'}{\cset_{2}}
  \\
  \jexpr{\env_1'}{\expr_2}{\type_2}{\effect_2}{\predmap_2}{\cset_3}{\env_2}
  \\
  \fresh{\tvar_1}
  \\
  \is{\logand{\parens{\logforget{\predmap_1}{\effect_2}}}{\predmap_2}}{\predmap}
  \\
  \jenvstrengthen{\env_1}{\logneg{\predmap_1}}{\env_1''}{\cset_4}
  \\
  \jenvjoin{\env_1''}{\env_2}{\env'} 
}{
\jexpr{\env}{\binand{\expr_1}{\expr_2}}{\tjoin{\tvar_1}{\type_2}}
{\effectconcat{\effect_1}{\effect_2}}{\predmap}
{\consconcat{\medcupl{i=1}{4}\cset_i}{\mkset{\flows{\type_1}{\preduse{\falsy}{\tvar_1}}}}}
  {\env'}
}
\and
\inferruleright{\cgor\label{rule:cgor}}{
  \jexpr{\env}{\expr_1}{\type_1}{\effect_1}{\predmap_1}{\cset_1}{\env_1}
  \\
  \jenvstrengthen{\env_1}{\logneg{\predmap_1}}{\env_1'}{\cset_2}
  \\
  \jexpr{\env_1'}{\expr_2}{\type_2}{\effect_2}{\predmap_2}{\cset_3}{\env_2}
  \\
  \fresh{\tvar_1}
  \\
  \is{\logor{\parens{\logforget{\predmap_1}{\effect_2}}}{\predmap_2}}{\predmap}
  \\
  \jenvstrengthen{\env_1}{\predmap_1}{\env_1''}{\cset_4}
  \\
  \jenvjoin{\env_1''}{\env_2}{\env'} 
}{
  \jexpr{\env}{\binor{\expr_1}{\expr_2}}
    {\tjoin{\tvar_1}{\type_2}}{\effectconcat{\effect_1}{\effect_2}}{\predmap}
    {\consconcat{\medcupl{i=1}{4}\cset_i}{\mkset{\flows{\type_1}{\preduse{\truthy}{\tvar_1}}}}}{\env'}
}
\and
\inferruleright{\cgnot\label{rule:cgnot}}{
  \jexpr{\env}{\expr}{\type}{\effect}{\predmap}{\cset}{\env'}
}{
  \jexpr{\env}{\unaryneg{\expr}}{\tbool}{\effect}{\logneg{\predmap}}{\cset}{\env'}
}
\and
\inferruleright{\cgpred\label{rule:cgpred}}{
}{
  \jexpr{\env}{\predof{\evar}}{\tbool}{\effempty}{\singletonpredmap{\evar}{\primpred}}{\csetemp}{\env}
}
\end{mathpar}
\caption{Expression Constraint Generation in \flowcore (Logical Operations)}
\label{fig:flow:cons:gen:expr:log}
\end{figure}
The rules of \figref{fig:flow:cons:gen:expr:log}
are interesting for their effect on predicate refinement.
%
%
In Rule~\hyperref[rule:cgand]{\cgand}, $\expr_2$ is
analyzed under the refinement $\predmap_1$ (since otherwise
it would not be evaluated).
The type inferred for the entire expression contains
components from both $\expr_1$ and $\expr_2$. From the
former it contains type $\tvar_1$ that is a version of
$\type_1$ refined by the \falsy predicate, since it
corresponds to the case where $\expr_1$ is actually falsy.
From the latter it includes the type $\type_2$ as is.
For the output environment we follow a similar strategy. The
component that corresponds to $\expr_1$'s output environment
will be refined with $\neg\predmap_1$, since otherwise we
would be using the environment corresponding to $\expr_2$.
With respect to the output predicate mapping,
parts of $\predmap_1$ that apply
on names written in $\expr_2$ are forgotten when taking the
conjunction with $\predmap_2$.
%
%
Rule~\hyperref[rule:cgor]{\cgor} is the dual of the
above rule, and works similarly.
Finally, rules~\hyperref[rule:cgnot]{\cgnot} and
\hyperref[rule:cgpred]{\cgpred} are straightforward.
The former just negates the refinement and the latter
introduces a refinement from a runtime test \primpred.

\paragraph{Records}
\begin{figure}[t]
\judgementHead{Expression Constraint Generation}{\jexpr{\env}{\expr}{\type}{\effect}{\predmap}{\cset}{\env'}}
\begin{mathpar}
%
%
\inferruleright{\cgrecord\label{rule:cgrecord}}{
  \isequiv{\env}{\env_0}
  \\
  \forallindot{i}{[1,n]}{\jexpr{\env_{i-1}}{\expr_i}{\type_i}{\effect_i}{\predmap_i}{\cset_i}{\env_i}}
  \\
  \fresh{\tvar_i}
}{
\jexpr{\env}
  {\objlitmany{\fieldsym_i}{\expr_i}}
  {\tobjmany{\fieldsym_i}{\tvar_i}}
  {\medsqcupl{i=1}{n}\effect_i}
  {\predmapempty}{\consconcat{\medcupl{i=1}{n}\cset_i}{\medcupl{i=1}{n}\uniflow{\type_i}{\tvar_i}}}{\env_n}
}
\and
\inferruleright{\cgfldrd\label{rule:cgfldrd}}{
  \jexpr{\env}{\expr}{\type}{\effect}{\predmap}{\cset}{\env'}
  \\
  \fresh{\tvar}
}{
  \jexpr{\env}{\fieldread{\expr}{\fieldsym}}{\tvar}{\effect}{\predmapempty}
  {\consconcat{\cset}{\uniflow{\type}{\getuse{\fieldsym}{\tvar}}}}
  {\env'}
}
\and
\inferruleright{\cgfldwr\label{rule:cgfldwr}}{
    \jexpr{\env}{\expr_1}{\type_1}{\effect_1}{\predmap_1}{\cset_1}{\env_1}
    \\
    \jexpr{\env_1}{\expr_2}{\type_2}{\effect_2}{\predmap_2}{\cset_2}{\env_2}
}{
  \jexpr{\env}{\fieldwrite{\expr_1}{\fieldsym}{\expr_2}}{\type_2}
    {\effectconcat{\effect_1}{\effect_2}}
    {\predmapempty}
    {\consconcatthree{\cset_1}{\cset_2}{\uniflow{\type_1}{\setuse{\fieldsym}{\type_2}}}}
    {\env_2}
}
\end{mathpar}
\caption{Expression Constraint Generation in \flowcore (Records)}
\label{fig:flow:cons:gen:expr:rec}
\end{figure}
The rules of \figref{fig:flow:cons:gen:expr:rec}
for record type inference are mostly routine.
During record creation the initializer types flow to the newly constructed record
literal type. Subsequent assignments of type \type
to a field \fieldsym widen the type of \fieldsym
by introducing flows to the use \setuse{\fieldsym}{\type}.

In practice, \flow follows a slightly stricter approach. It ``fixes'' the type
of an object at initialization and checks that all subsequent writes adhere to
this type. This essentially amounts to checking for type annotations
which is out of scope in this section of type inference.

\paragraph{Statements}
\begin{figure}[t]
\judgementHead{Statement Constraint Generation}{\jstmt{\env}{\stmt}{\effect}{\cset}{\env'}}
\begin{mathpar}
\inferruleright{\cgexp\label{rule:cgexp}}{
  \jexpr{\env}{\expr}{\type}{\effect}{\predmap}{\cset}{\env'}
}{
  \jstmt{\env}{\expr}{\effect}{\cset}{\env'}
}
\and
\inferruleright{\cgvardecl\label{rule:cgvardecl}}{
  \jexpr{\env}{\assign{\evar}{\expr}}{\type}{\effect}{\predmap}{\cset}{\env'}
}{
  \jstmt{\env}{\varassign{\evar}{\expr}}{\effect}{\cset}{\env'}
}
\and
\inferruleright{\cgif\label{rule:cgif}}{
  \jexpr{\env}{\expr}{\type}{\effect}{\predmap}{\cset_1}{\env'}
  \\
  \jenvstrengthen{\env'}{\predmap}{\env_1}{\cset_2}
  \\
  \jstmt{\env_1}{\stmt_1}{\effect_1}{\cset_3}{\env_1'}
  \\
  \jenvstrengthen{\env'}{\negpredmap{\predmap}}{\env_2}{\cset_4}
  \\
  \jstmt{\env_2}{\stmt_2}{\effect_2}{\cset_5}{\env_2'}
  \\
  \jenvjoin{\env_1'}{\env_2'}{\env''} 
}{
  \jstmt{\env}{\ite{\expr}{\stmt_1}{\stmt_2}}{\effectconcatthree{\effect}{\effect_1}{\effect_2}}
  {\medcupl{i=1}{5}\cset_i}{\env''}
}
\and
\inferruleright{\cgseq\label{rule:cgseq}}{
  \jstmt{\env}{\stmt_1}{\effect_1}{\cset_1}{\env_1}
  \\
  \jstmt{\env_1}{\stmt_2}{\effect_2}{\cset_2}{\env_2}
}{
  \jstmt{\env}{\seq{\stmt_1}{\stmt_2}}{\effectconcat{\effect_1}{\effect_2}}{\consconcat{\cset_1}{\cset_2}}{\env_2}
}
\end{mathpar}
\caption{Statement Constraint Generation in \flowcore}
\label{fig:flow:cons:gen:stmt}
\end{figure}
The main difference compared to the respective expression rule is the
omission of the assigned type and the refinement predicate.
Rule~\hyperref[rule:cgvardecl]{\cgvardecl} reuses the rule
for assignment that we saw earlier, since due to
variable hoisting, \evar is already in scope.
Rule~\hyperref[rule:cgif]{\cgif} handles conditional statements. This rule uses
the refinement \predmap for the conditional expression \expr to refine the
environments that are used to check each branch, with the appropriate sign in
each case. The output environment is the join of the environments at the end of
each branch.

\subsection*{Example}
%
We now examine how the rules of
\figsref{fig:flow:cons:gen:expr}{fig:flow:cons:gen:stmt}
handle
the code in \linesref{code:17}{code:41}
in \figref{fig:intro:code}.
In the following we keep the
produced type bindings on the left and constraint sets on the right.
Whenever, a general type (exponent) is not made explicit, this means
that it's not important for that particular binding.
Also, to avoid clutter, we do not define a new environment for each
program point, but rather introduce different versions for variables
that get updated or refined.

By applying Rule~\hyperref[rule:cgrecord]{\cgrecord} on
\mylineref{code:17}:
\begin{align}
  & \envbinding{\code{nil}}{\tobjsingleton{\texttt{kind}}{\tvar_1}}
  &&\issupseteq{\cset}{\uniflow{\texttt{"nil"}}{\tvar_1}}
\label{ex:cg:1}
\intertext{
Here, \code{"nil"} is the string literal type
denoting the exact string \inlinejs{"nil"}.
For the function \inlinejs{cons}
(\linesref{code:18}{code:20}) we get
}
&\envbinding{\code{cons}}{\tlambda{\parenspair{\tvar_2}{\tvar_3}}{O}}
&&\issupseteq{\cset}{\setofthree{\flows{\code{"cons"}}{\tvar_4}}{\flows{\tvar_2}{\tvar_5}}{\flows{\tvar_3}{\tvar_6}}}
\label{ex:cg:2}
\intertext{%
where
\isdef{O}{\tobjthree{\code{kind}}{\tvar_4}{\code{head}}{\tvar_5}{\code{tail}}{\tvar_6}}.
We also define
\isdef{\type_{\inlinejs{cons}}}{\tlambda{\parenspair{\tvar_2}{\tvar_3}}{O}}.
The function's effect is empty, so omitted here.
Moving on to function \code{sum}, before checking its body
we introduce bindings for the (recursive) function itself and its parameter:
}
& \envbinding{\code{sum}}{\tlambda{\tvar_7}{\type_r}},\;\envbinding{\code{list}}{\tvar_7} &&
\label{cg:def:sum}
\intertext{%
We define
\isdef{\type_{\inlinejs{sum}}}{\tlambda{\tvar_7}{\type_r}}.
Checking the conditional in \mylineref{code:23}, \code{list}
gets a more precise type, and is referred to as $\code{line}_1$ inside the then-branch:
}
& \envbinding{\code{list}_1}{\tvarb_7}
&& \issupseteq{\cset}{\uniflow{\tvar_7}{\preduse{\primpred_c}{\tvarb_7}}}
\label{ex:cg:4}
\intertext{%
Here, \isdef{\primpred_c}{\star\inlinejs{.kind === "cons"}} is the predicate of exact equality
of the field \inlinejs{kind} with the string \inlinejs{"cons"}.
The uses of \inlinejs{list} in \mylineref{code:24}
produce the following constraints (here we focus on the interesting uses
\ie, the two field accesses and the call):
}
&
&&\issupseteq{\cset}{\mkset{
  \begin{tabular}{l}
    \flows{\tvarb_7}{\getuse{\texttt{head}}{\tvarc_1}}, \\
    \flows{\tvarb_7}{\getuse{\texttt{tail}}{\tvarc_2}}, \\
    \flows{\type_{\code{sum}}}{\calluse{\tvarc_2}{}{\tvard_1}}
  \end{tabular}
  }}
  \label{ex:cg:5}
\intertext{
We omit the constraints pertinent to the return statements, since they are not
crucial in this example.
The compound calls in \mylineref{code:29} further produce the constraints (starting from
deeper nesting levels):
}
& &&\issupseteq{\cset}{
\uniflow{\type_{\texttt{cons}}}{\calluse{\parenspair{\tnum}{\tobjsingleton{\texttt{kind}}{\tvar_1}}}{}{\tvard_2}}
}
\label{cg:call:cons:a}
\\
& &&\issupseteq{\cset}{
\uniflow{\type_{\texttt{cons}}}{\calluse{\parenspair{\tnum}{\tvard_2}}{}{\tvard_3}}
}
\label{cg:call:cons:b}
\\
& &&\issupseteq{\cset}{
\uniflow{\type_{\texttt{sum}}}{\calluse{\tvard_3}{}{\tvard_4}}
\label{cg:call:sum}
}
\intertext{%
In function \code{merge}, let $\code{x}_1$ correspond to the initial value for \code{x}
and $\code{x}_2$ to the value after the update in \mylineref{code:32}.
Below, the first three constraints correspond to the use of the \code{||} operator and the last one
to the field access in line~\ref{code:33}:
}
& \envbinding{\code{x}_1}{\entry{\tvar_8}{\tvar_8}},\;
  \envbinding{\code{x}_2}{\entry{\tvar_{11}}{\tvar_8}}
&&
  \issupseteq{\cset}{\mkset{\begin{tabular}{l}
      \flows{\tvar_8}{\preduse{\truthy}{\tvarb_8}}, \\
      \flows{\tjoin{\tvarb_8}{\type_{\texttt{nil}}}}{\tvar_{11}}, \\
      \flows{\tvar_{11}}{\getuse{\texttt{kind}}{\tvar_{10}}}
    \end{tabular}
    \label{ex:cg:10}
  }}
\intertext{%
Finally, function \code{havoc} in lines~\ref{code:36}--\ref{code:41} is similar to \code{merge}
(so we won't repeat the common parts), but additionally, defines a function
\code{reset}, that assigns \inlinejs{null} to \inlinejs{x}.
Crucially, the type of \inlinejs{x} inside
\inlinejs{reset} has been erased to $\tvar_8$:
}
& \envbinding{\code{reset}}{\tarrow{\parens{}}{\code{x}}{\tvoid}}
&&
\issupseteq{\cset}{\uniflow{\tnull}{\tvar_8}}
\label{ex:cg:11}
\intertext{%
The call to \code{reset} in line~\ref{code:39}
needs to handle the function's effect, so a fresh variable \effvar is generated:
}
&
&&
\issupseteq{\cset}{\mkset{\begin{tabular}{l}
  \flows{\tarrow{\parens{}}{\code{x}}{\tvoid}}{\calluse{\parens{}}{\effvar}{\tvoid}}, \\
  \flows{\effvar}{\havocuse{\envupd{\env}{\code{x}}{\entry{\tvar_{11}}{\tvar_8}}}}
  \end{tabular}
  \label{ex:cg:12}
  }}
\end{align}

For the moment, we have merely constructed a flow network, but haven't reached
any critical conclusions. In the next section, we'll see how we can use these
facts to discover inconsistencies, and what guarantees we get if we do not find any.



\subsection{Propagation}\label{sec:cons:prop}

Thinking of our system as a dataflow analysis framework, constraint generation
amounts to setting up a flow network.
The next step is to allow the system to stabilize under a set of appropriate flow functions.
This latter
part is called \emph{constraint propagation} and corresponds to exploring
\emph{all} potential data-flow paths and finding inconsistencies in them.
Decomposing complex constraints intro simpler ones is done by the
rules shown in \figref{fig:flow:cons:prop}.
We say that a constraint set \cset is in \emph{closed form}, if it is closed with
respect to these rules.
%
%
%
%
In practice, we keep our constraint sets in closed form at all times
during constraint generation; that is, for every new
constraint that gets generated,
we apply all eligible propagation
rules until we reach a fixpoint.

\begin{figure}[t]
\begingroup
\begin{align}
  \issubseteq{\mkset{\concat{\flows{\type}{\tvar}}{\flows{\tvar}{\tuse}}}}{\cset} &
  \implies
  \inset{\flows{\type}{\tuse}}{\cset}
  \label{rule:cptrans}
  \tag{\cptrans}
  \\
  \issubseteq{\setoftwo{\flows{\effect}{\effvar}}{\flows{\effvar}{\euse}}}{\cset} &
  \implies
  \inset{\flows{\effect}{\euse}}{\cset}
  \label{rule:cptranseffect}
  \tag{\cptranseffect}
  \\
  \inset{\flows{\type_1\join\type_2}{\tuse}}{\cset} &
  \implies
  \issubseteq{\setoftwo{\flows{\type_1}{\tuse}}{\flows{\type_2}{\tuse}}}{\cset}
  \label{rule:cpjoin}
  \tag{\cpjoin}
  \\
  \inset{\flows{\effect_1\join\effect_2}{\euse}}{\cset} &
  \implies
  \issubseteq{\setoftwo{\flows{\effect_1}{\euse}}{\flows{\effect_2}{\euse}}}{\cset}
  \label{rule:cpjoineffect}
  \tag{\cpjoineffect}
  \\
  \inset{\flows{\tarrow{\tvar}{\effect}{\type}}{\calluse{\type'}{\effvar}{\tvarb}}}{\cset} &
  \implies
  \issubseteq{\setofthree{\flows{\type'}{\tvar}}{\flows{\type}{\tvarb}}{\flows{\effect}{\effvar}}}{\cset}
  \label{rule:cpcall}
  \tag{\cpcall}
  \\
  \inset{\flows{\evar}{\havocuse{\envext{\env}{\evar}{\entry{\tvarb}{\tvar}}}}}{\cset} &
  \implies
  \inset{\flows{\tvar}{\tvarb}}{\cset}
  \label{rule:cphavoc}
  \tag{\cphavoc}
  \\
  \inset{\flows{\tlit}{\preduse{\pred}{\tvar}}}{\cset}
  \wedge
  \checkpred{\tlit}{\pred} &
  \implies
  \inset{\flows{\tlit}{\tvar}}{\cset}
  \label{rule:cppredbase}
  \tag{\cppredbase}
  \\
  \issubseteq{\setoftwo
    {\flows{\type}{\tvar}}
    {\flows{\tctxidx{\type'}{\tvar}^{+}}{\preduse{\pred}{\tvarb}}}
  }{\cset} &
  \implies
  \inset{\flows{\tctxidx{\type'}{\type}^{+}}{\preduse{\pred}{\tvarb}}}{\cset}
  \label{rule:cppredtrans}
  \tag{\cppredtrans}
  \\
  \inset{\flows{\tobjone{\fieldsym}{\tvar}}{\getuse{\fieldsym}{\tvarb}}}{\cset} &
  \implies
  \inset{\flows{\tvar}{\tvarb}}{\cset}
  \label{rule:cpget}
  \tag{\cpget}
  \\
  \inset{\flows{\tobjone{\fieldsym}{\tvar}}{\setuse{\fieldsym}{\type}}}{\cset} &
  \implies
  \inset{\flows{\type}{\tvar}}{\cset}
  \label{rule:cpset}
  \tag{\cpset}
\end{align}
\endgroup
\caption{Constraint Propagation in \flowcore}
\label{fig:flow:cons:prop}
\end{figure}

If we consider the elements of \cset as subtyping constraints, then these
rules amount to subtyping rules.
Rules~\hyperref[rule:cptrans]{\cptrans}
and~\hyperref[rule:cptranseffect]{\cptranseffect} express
transitivity for types and effect, respectively.
\hyperref[rule:cpjoin]{\cpjoin}
and~\hyperref[rule:cpjoineffect]{\cpjoineffect} decompose as usual flows from joins
of elements.

Rule~\hyperref[rule:cpcall]{\cpcall} decomposes the flow of an arrow type to
a calling context. Note that the incoming arrow type has a type variable \tvar
as the parameter type, since this is the form in which it is produced by
\hyperref[rule:cgfun]{\cgfun}. Also by the \hyperref[rule:cgcall]{\cgcall}
the effect and return type portion of the calling use are also variables.
Handling this flow propagates three new flows:
(i)~the argument's type $\type'$  flows to the parameter type variable $\tvar$,
(ii)~the return type $\type$ flows to the call-site's type $\tvarb$,
and
(iii)~the function's effect \effect flows to the call's effect variable $\effvar$.
This last byproduct often triggers the ``havoc'' mechanism,
which carries out the task of applying a function's effect on the
variables that are updated by it.

Rule~\hyperref[rule:cphavoc]{\cphavoc} recovers soundness by
restoring the conservative types for variables that are updated
through function calls (\hyperref[rule:cgcall]{\cgcall})
or are reassigned during conditional type refinement
(\hyperref[rule:refeffect]{\refeffect}).
Lets assume the havoc
operation was introduced due to the former rule.
(The latter case
works similarly.)
When \hyperref[rule:cphavoc]{\cphavoc} is triggered,
our analysis has determined that variable \evar gets
updated in the called function, and so entries bound to
\evar in the environment after the function call should be
conservatively approximated.
Of course, this rule is only meaningful if \evar is bound in that
environment. Otherwise this effect can be ignored.
Effectively, this corresponds to erasing the type of the
binding \binding{\evar}{\entry{\tvarb}{\tvar}}, by generating
a flow from the flow-insensitive type \tvar to \tvarb.
This process may happen far away from the actual
call-site, which exemplifies the global character of the
type inference.
An observant reader might notice that the environment
argument of \havocusesym has entries of the form
\entry{\tvarb}{\tvar}.
We can guarantee that this is the only possible form, by
construction of the respective flows in rules
\hyperref[rule:cgcall]{\cgcall} and
\hyperref[rule:refeffect]{\refeffect}.
In both cases this
happens after a widening operation, which guarantees that
the base of the environment entry is a type variable.

Rule~\hyperref[rule:cppredbase]{\cppredbase}
handles predicate refinement.
The intuition here is that \tlit should flow to \tvar, if it succeeds in
the check implied by \pred, \ie, if \checkpred{\tlit}{\pred} is true.
We have kept the representation of base predicates abstract, and so we will
do with the definition of \checksym.
In general, \checksym should be
able to decide if \tlit satisfies \pred by inspecting its top-level
constructor (for checks like \inlinejs{typeof}$\;\star$
\inlinejs{ === "string"})
or one of its fields.

%
%
Rule~\hyperref[rule:cppredtrans]{\cppredtrans} is a technical one.
It allows parts of types under refinement to be concretized.
In \tctxidx{\type'}{\tvar},
the form \tctxidx{\type'}{}
is a \emph{type context}, \ie, a type with a ``hole'' that is filled in with \tvar,
for example \tobjsingleton{\fieldsym}{\angles{}}.
While rule~\hyperref[rule:cptrans]{\cptrans} will fail to instantiate \tvar,
\hyperref[rule:cppredtrans]{\cppredtrans} allows type variables
appearing under a type
constructor (\eg, the object constructor) to be instantiated.
However, not all substitutions are allowed, but only
the ones where \tvar is in a positive position with respect
to type \emph{polarity}~\cite{Pottier1998, Dolan2017}%
\footnote{\secref{sec:app:polarities}
includes a
formal definition of polarity and type contexts.}.
The reason we require type variable \tvar to appear in a positive
position is to abide by our restriction that type joins cannot appear
at the right-hand side of constraints. If we allowed the replacement
of \tvar from \type in any part of $\type'$, this could potentially
break this invariant in a later propagation.
We will also see the importance of this rule in the upcoming
example.


Finally, Rules~\hyperref[rule:cpget]{\cpget}
and~\hyperref[rule:cpset]{\cpset} handling record field
access and update
are standard.



\subsection{Consistency}
\label{sec:consistency}

The goal of running constraint generation and propagation is
to eventually discover inconsistencies in the saturated
constraint set. These effectively correspond to potential
bugs in the use of the various operators, for example they
could correspond to the case of a non-function value
reaching the receiver position of a call. Below we present a
formal description of consistency.

\begin{definition}[Consistency]
A closed constraint set is \emph{consistent} if it does not contain any constraints
in one of the forms:
\begin{itemize}
  \item \flows{\tlit}{\calluse{\type}{\effvar}{\tvar}} where \tlit is \emph{not} an
    arrow type (or an arrow-like type, \eg, the type of a constructor object).
  \item
    \flows{\tlit}{\setuse{\fieldsym}{\type}} or
    \flows{\tlit}{\getuse{\fieldsym}{\tvar}}
    where \tlit is \emph{not} an
    record type literal (or an object-like type) containing \fieldsym.
\end{itemize}
\end{definition}

If our analysis finds an inconsistency, then this leads to an error report.
Otherwise, if no inconsistency can be found then the input program enjoys
the safety guarantees of Theorem~\ref{thm:flow:type:safety}.


%

\subsection*{Example}

We continue where we left off in the example of
\secref{sec:cons:gen},
by applying the rules from Figure~\ref{fig:flow:cons:prop} on \cset, in order to discover
inconsistencies or prove the absence thereof.

\paragraph{Use of predicates}
We start by applying \hyperref[rule:cpcall]{\cpcall} on
the calls of~\eqref{ex:cg:5},~\eqref{cg:call:cons:a},~\eqref{cg:call:cons:b},%
~\eqref{cg:call:sum}, and the respective function definitions:
\begin{align}
  & \issupseteq{\cset}{\setoftwo{\flows{\tvarc_2}{\tvar_7}}{\flows{\type_r}{\tvard_1}}}
  \label{ex:cp:13}
  \\
  & \issupseteq{\cset}{\setofthree{\flows{\tnum}{\tvar_2}}
    {\flows{\tobjsingleton{\texttt{kind}}{\tvar_1}}{\tvar_3}}{\flows{O}{\tvard_2}}}
  \label{ex:cp:14}
  \\
  & \issupseteq{\cset}{\setofthree{\flows{\tnum}{\tvar_2}}{\flows{\tvard_2}{\tvar_3}}{\flows{O}{\tvard_3}}}
  \label{ex:cp:15}
  \\
  & \issupseteq{\cset}{\setoftwo{\flows{\tvard_3}{\tvar_7}}{\flows{\type_r}{\tvard_4}}}
  \label{ex:cp:16}
\end{align}

Now lets focus on the interesting case of handling the getters of~\eqref{ex:cg:5}.
By transitivity (\hyperref[rule:cptrans]{\cptrans}) using~\eqref{ex:cp:15},~\eqref{ex:cp:16}
and~\eqref{ex:cg:4}, the record type $O$ flows to the predicate use:
\begin{align}
  &\issupseteq{\cset}{\uniflow
    {\tobjthree{\texttt{kind}}{\tvar_4}{\texttt{head}}{\tvar_5}{\texttt{tail}}{\tvar_6}}
    {\preduse{\primpred_c}{\tvarb_7}}
  }
  \label{ex:cp:17}
\intertext{\normalsize
We use Rule~\hyperref[rule:cppredtrans]{\cppredtrans} on~\eqref{ex:cg:2}
and~\eqref{ex:cp:17} to obtain:
}
&\issupseteq{\cset}{\uniflow
  {\tobjthree{\texttt{kind}}{\texttt{"cons"}}{\texttt{head}}{\tvar_5}{\texttt{tail}}{\tvar_6}}
    {\preduse{\primpred_c}{\tvarb_7}}
  }
  \label{ex:cp:18}
\intertext{
This is now a successful test since the string literal type \texttt{"cons"} of field
\inlinejs{kind} satisfies $\primpred_c$ and so:
}
& \issupseteq{\cset}{\uniflow
  {\tobjthree{\texttt{kind}}{\texttt{"cons"}}{\texttt{head}}{\tvar_5}{\texttt{tail}}{\tvar_6}}
    {\tvarb_7}
  }
  \label{ex:cp:19}
\intertext{%
\flow has thus discovered a path in which a ``cons'' object reaches the field accesses
of \mylineref{code:24}.
However, this latest constraint has enabled new flows
that could
cause inconsistencies, \eg, the
recursive calls to \inlinejs{sum} on the \inlinejs{tail} of \inlinejs{list}.
By~\eqref{ex:cp:19} and~\eqref{ex:cg:5}, and applying
\hyperref[rule:cptrans]{\cptrans} and
\hyperref[rule:cpget]{\cpget}:
}
&\issupseteq{\cset}{\uniflow{\tvar_6}{\tvarc_2}}
  \label{ex:cp:20}
\intertext{\normalsize
Indeed, by combining~\eqref{ex:cp:14},~\eqref{ex:cg:2},~\eqref{ex:cp:20},~\eqref{ex:cp:13}
and~\eqref{ex:cg:4}
with \hyperref[rule:cptrans]{\cptrans} and the result with~\eqref{ex:cg:1}
with \hyperref[rule:cppredtrans]{\cppredtrans}:
}
& \issupseteq{\cset}{\uniflow{\tobjsingleton{\texttt{kind}}{\texttt{"nil"}}}
    {\preduse{\primpred_c}{\tvarb_7}}
  }
  \label{ex:cp:21}
\intertext{\normalsize
This test, however, will fail, as it would at runtime, and so the ``nil'' object will not reach the
getter for \inlinejs{head} or \inlinejs{tail} through $\tvar_6$.
Without the predicate refinement filtering out ``nil'' objects,
we would have introduced a false positive.
\paragraph{Refinements and Mutation} Last, we illustrate how \flow handles
functions \inlinejs{merge} and \inlinejs{havoc}. We start by processing~\eqref{ex:cg:12}
with \hyperref[rule:cpcall]{\cpcall} and then \hyperref[rule:cphavoc]{\cphavoc},
which yields:
}
&\issupseteq{\cset}{\uniflow{\tvar_8}{\tvar_{11}}}
  \label{ex:cp:22}
\intertext{\normalsize
This allows the \inlinejs{null} from the \inlinejs{reset} function to find its way
to $\tvar_{11}$ from~\eqref{ex:cg:11} and from there to the ``get'' operation through~\eqref{ex:cg:10}:
}
& \issupseteq{\cset}{\uniflow{\tnull}{\getuse{\texttt{kind}}{\tvar_{10}}}}
  \label{ex:cp:23}
\end{align}
This latter constraint signals a consistency violation, keeping \flow sound
with respect to variable updates that invalidate prior refinements.


\subsection{Implementation of Type Inference}
\label{sec:inf:impl}

A set of flow constraints \cset can be thought of
as a constraint graph, where variables, literals
and uses are the nodes and the constraints among them are the edges.
In this section, we briefly discuss how we represent constraint graphs
and compute their closure efficiently.
Let us refer to type and effect variables as ``unknowns.'' Following
\citet{Pottier2001}, the constraint graph maps each unknown to a set of
lower bounds and a set of upper bounds, each of which contains the unknown
itself. The transitive propagation rules are specialized to exploit this
structure to efficiently keep the constraint graph in closed form.

However, equality constraints are quite inefficient in this system: they are
represented as a pair of subset constraints, which causes a cubic blowup in the
transitive propagation rules. On the other hand, equality constraints are quite
useful and common in \flow. They arise due to invariant typing of object
properties, array elements, and type arguments of polymorphic classes. They
directly model equations expressed by type aliasing. Finally, even though we
formalize \hyperref[rule:cphavoc]{\cphavoc} with a constraint of the form
$\tvara \leq \tvarb$, we can replace it without loss of generality with $\tvara =
\tvarb$.

To address the inefficiency, we generalize the constraint graph by considering
each unknown to be in an equivalence class containing other unknowns it is
unified with, and mapping each equivalence class to either ``unresolved'' bounds
(like~\citet{Pottier2001}), or to a ``resolved'' type or effect (as in unification). The
transitive propagation rules generalize in a straightforward way. Overall, this
simple optimization leads to $O(n)$ reduction in space and time complexity.

\section{Runtime Semantics}
\label{sec:semantics}

Before we describe our safety result
(\secref{sec:type:soundness}) we present the runtime
semantics for the formal fragment of \secref{sec:core},
%
which is heavily based on that
used by~\citet{Rastogi2015} that cover a subset of
\js, emphasizing on features of interest, while
abstracting away non-crucial ones.

\begin{figure}[t!]
\begin{langdefsmall}
\langline{\expr}{\production}{\dots \sep \loc}{\textbf{Runtime Expressions}}
\vsep
\\
\langline{\val}{\production}{\const \sep \loc}{\textbf{Values}}
\vsep
\\
\langline{\hval}{\production}{
  \val \sep \storearrowshort{\store}{\evar}{\exprbody}
       \sep \objlit{\fieldsym_1}{\val_1}{\fieldsym_n}{\val_n}
}{\textbf{Heap Values}}
\vsep
\\
\langline{\ectx}{\production}{
  \hole \sep
  \assign{\evar}{\ectx} \sep
  \efuncall{\ectx}{\expr} \sep
  \efuncall{\loc}{\ectx} \sep
  \binand{\ectx}{\expr} \sep
  \binor{\ectx}{\expr} \sep
}{\textbf{Evaluation Contexts}}
\\
\langline{}{\sep}{\objlitthreespread{\fieldsym_1}{\val_1}{\fieldsym_k}{\ectx_k}{\fieldsym_n}{\expr_n}}{}
\\
\langline{}{\sep}{
  \unaryneg{\ectx} \sep
  \fieldread{\ectx}{\fieldsym} \sep
  \fieldwrite{\ectx}{\fieldsym}{\expr} \sep
  \fieldwrite{\val}{\fieldsym}{\ectx} \sep
  \varassign{\evar}{\ectx}
}{}
\\
\langline{}{\sep}{
  \ite{\ectx}{\stmt_1}{\stmt_2} \sep
  \return{\ectx} \sep
  \seq{\ectx}{\stmt}
}{}
\end{langdefsmall}
\vspace{0.3em}
\begin{minipage}{0.45\textwidth}
\begin{langdefsmall}
\langline{\heap}{\production}{\cdot \sep \heapext{\heap}{\loc}{\hval}}{\textbf{Heaps}}
\vsep
\\
\langline{\stack}{\production}{\cdot \sep \stackcons{\stack}{\store}{\ectx}}{\textbf{Stacks}}
\vsep
\\
\langline{\store}{\production}{\cdot \sep \storeext{\store}{\evar}{\loc}}{\textbf{Stores}}
\end{langdefsmall}
\end{minipage}
\begin{minipage}{0.45\textwidth}
\begin{langdefsmall}
\langline{\rtstate}{\production}{\rttriplet{\heap}{\stack}{\store}}{\textbf{States}}
\vsep
\\
\langline{\rtconf}{\production}{
  \mkconf{\rtstate}{\expr} \sep
  \mkconf{\rtstate}{\stmt} \sep
  \mkconf{\rtstate}{\exprbody}
}{\textbf{Configurations}}
\end{langdefsmall}
\end{minipage}
\begin{minipage}{0.05\textwidth}
  \quad
\end{minipage}
\caption{Runtime Definitions in \flowcore}
\label{fig:flow:runtime}
\end{figure}

\subsection{Definitions}

\figref{fig:flow:runtime} contains the definitions for
runtime configurations in \flowcore.

\paragraph{Runtime Values}
To account for heap-allocated values, we introduce
\emph{locations} \loc that index runtime heaps.
Together with constants they synthesize runtime
\emph{values}, which are normal form as far as execution is
concerned.
Locations are also added to the set of expressions in our runtime
language along with all other expression forms introduced earlier.

\paragraph{Runtime State}
There are three constituent parts that compose a \emph{runtime
state} \rtstate.
The first part is the \emph{heap} \heap, which includes
bindings from locations to \emph{heap values} \hval, which
in turn are either values, closures, or heap objects.
A \emph{closure} is a pair containing a store \store that
binds all external variables available at the point of
definition of the arrow function (capture by reference), and
the function's code, which is a statement succeeded by a
returned expression.
The second part of the runtime state is the \emph{stack}
\stack, that contains a list of stack frames.  Each stack
frame includes a store containing the variables bound in the
stack frame at the time execution left that frame, and an
evaluation context \ectx that holds the context that
execution would jump into when returning to that stack
frame. Evaluation contexts are defined in the usual way
having the same structure as expressions or statements but
with a hole \hole at the position of the term that is about
to be evaluated next.
Finally, the runtime state includes a \emph{store} \store,
that comprises bindings of variable names to locations to
allow closures to capture values by reference.

\paragraph{Runtime Configurations}
We write our \emph{runtime configurations} \rtstate (\ie,
programs under execution) as pairs that contain a runtime
state \rtstate, and a language term, which can either be an
expression \expr, a statement \stmt, or a function body
\body{\stmt}{\expr}.
We conflate the notions of expressions and function bodies
into a common notion using the symbol \exprbody, for
compactness in stating our results.

\subsection{Reduction Rules}

\begin{figure}[t!]
\judgementHeadTwo{Expression and Statement Reduction Rules}
  {\stepsconf{\rtstate}{\exprbody}{\rtstate'}{\exprbody'}}
  {\stepsconf{\rtstate}{\stmt}{\rtstate'}{\stmt'}}
\begin{mathpar}
\inferruleright{\rtectx\label{rule:rtectx}} {
  \stepsconf{\rttriplet{\heap}{\stack}{\store}}{\expr}
        {\rttriplet{\heap'}{\stack}{\store'}}{\expr'}
}{
  \stepsconf{\rttriplet{\heap}{\stack}{\store}}{\ectxidx{\ectx}{\expr}}
        {\rttriplet{\heap'}{\stack}{\store'}}{\ectxidx{\ectx}{\expr'}}
}
\and
\inferruleright{\rtevar\label{rule:rtevar}}{
}{
  \stepsconf{\rtstate}{\evar}{\rtstate}{\idx{\confproj{\rtstate}{\heap}}{\idx{\confproj{\rtstate}{\store}}{\evar}}}
}
\and
\inferruleright{\rtasgn\label{rule:rtasgn}}{
  \is{\heap'}{
    \upd{\heap}{\idx{\store}{\evar}}{\val}
  }
}{
  \stepsconf
    {\rtstate}
    {\assign{\evar}{\val}}
    {\confasgn{\rtstate}{\heap'}}
    {\val}
}
\and
\inferruleright{\rtarrow\label{rule:rtarrow}}{
  \fresh{\loc}
  \\
  \is{\heap'}{\heapext{\heap}{\loc}{\storearrowshort{\confproj{\rtstate}{\store}}{\evar}{\exprbody}}}
}{
  \stepsconf{\rtstate}{\arrowshort{\evar}{\exprbody}}{\confasgn{\rtstate}{\heap'}}{\loc}
}
\and
\inferruleright{\rtcall\label{rule:rtcall}}{
  \is{\idx{\heap}{\loc}}{\storearrowshort{\store_0}{\evar}{\exprbody}}
  \\
  \fresh{\loc', \many{\loc_i}}
  \\
  \is{\many{\evar_i}}{\locals{\exprbody}}
  \\
  \is{\heap'}{\heapcons{\heapext{\heap}{\loc'}{\val}}{\many{\heapbinding{\loc_i}{\texttt{undefined}}}}}
  \\
  \is{\stack'}{\stackcons{\stack}{\store}{\ectx}}
  \\
  \is{\store'}{\storecons{\storeext{\store_0}{\evar}{\loc'}}{\many{\storebinding{\evar_i}{\loc_i}}}}
}{
\stepsconf{\rttriplet{\heap}{\stack}{\store}}{\ectxidx{\ectx}{\efuncall{\loc}{\val}}}
      {\rttriplet{\heap'}{\stack'}{\store'}}{\exprbody}
}
\and
\inferruleright{\rtpredvar\label{rule:rtpredvar}}{
  \isequiv{\rtstate}{\rttriplet{\heap}{\stack}{\store}}
  \\
  \is{\hval}{\idx{\heap}{\idx{\store}{\evar}}}
}{
  \stepsconf{\rtstate}{\predof{\evar}}{\rtstate}{\idx{\delta_{\primpred}}{\hval}}
}
\and
\inferruleright{\rtandtru\label{rule:rtandtru}}{
  \istruthy{\val}
}{
  \stepsconf{\rtstate}{\binand{\val}{\expr}}{\rtstate}{\expr}
}
\and
\inferruleright{\rtandfls\label{rule:rtandfls}}{
  \isfalsy{\val}
}{
  \stepsconf{\rtstate}{\binand{\val}{\expr}}{\rtstate}{\val}
}
\and
\inferruleright{\rtortru\label{rule:rtortru}}{
  \istruthy{\val}
}{
  \stepsconf{\rtstate}{\binor{\val}{\expr}}{\rtstate}{\val}
}
\and
\inferruleright{\rtorfls\label{rule:rtorfls}}{
  \isfalsy{\val}
}{
  \stepsconf{\rtstate}{\binor{\val}{\expr}}{\rtstate}{\expr}
}
\and
\inferruleright{\rtneg\label{rule:rtneg}}{
  \is{\val'}{\neg\tobool{\val}}
}{
  \stepsconf{\rtstate}{\unaryneg{\val}}{\rtstate}{\val'}
}
\and
\inferruleright{\rtrecord\label{rule:rtrecord}}{
  \fresh{\loc}
  \\
  \is{\heap'}{\heapext{\heap}{\loc}{\objlit{\fieldsym_1}{\val_1}{\fieldsym_n}{\val_n}}}
}{
  \stepsconf{\rtstate}{\objlit{\fieldsym_1}{\val_1}{\fieldsym_n}{\val_n}}{\confasgn{\rtstate}{\heap'}}{\loc}
}
\and
\inferruleright{\rtfldrd\label{rule:rtfldrd}}{
  \is{\idx{\confproj{\rtstate}{\heap}}{\loc}}{\tbraces{
      \many{\objbinding{\fieldsym_i}{\val_i}}\code{,}\,
      \objbinding{\fieldsym}{\val}\code{,}\,
      \many{\objbinding{\fieldsym_j}{\expr_j}}
  }}
}{
  \stepsconf{\rtstate}{\fieldread{\loc}{\fieldsym}}{\rtstate}{\val}
}
\and
\inferruleright{\rtfldwr\label{rule:rtfldwr}}{
  \is{\heap'}{
    \heapupd{\confproj{\rtstate}{\heap}}
            {\loc}
            {\heapupd{\idx{\confproj{\rtstate}{\heap}}{\loc}}
                     {\fieldsym}
                     {\val}
            }
  }
}{
  \stepsconf{\rtstate}{\fieldwrite{\loc}{\fieldsym}{\val}}{\confasgn{\rtstate}{\heap'}}{\val}
}
\\
%
%
\inferruleright{\rtlet\label{rule:rtlet}}{
  \is{\heap'}{\heapupd{\heap}{\loc}{\val}}
}{
  \stepsconf
    {\rtstate}
    {\varassign{\evar}{\val}}
    {\confasgn{\rtstate}{\heap'}}
    {\skipstmt}
}
\and
\inferruleright{\rtif\label{rule:rtif}}{
  \isequiv{\stmt}{\istruthy{\val}\;?\;\stmt_1\;:\;\stmt_2}
}{
  \stepsconf{\rtstate}{\ite{\val}{\stmt_1}{\stmt_2}}{\rtstate}{\stmt}
}
%
\and
\inferruleright{\rtret\label{rule:rtret}}{
  \is{\confproj{\rtstate}{\stack}}{\stackcons{\stack'}{\store}{\ectx}}
  \\
  \is{\rtstate'}{\triplet{\confproj{\rtstate}{\heap}}{\stack'}{\store}}
}{
  \stepsconf{\rtstate}{\return{\val}}{\rtstate'}{\ectxidx{\ectx}{\val}}
}
\and
\inferruleright{\rtskip\label{rule:rtskip}}{
}{
  \stepsconf{\rtstate}{\seq{\skipstmt}{\stmt}}{\rtstate}{\stmt}
}
\end{mathpar}
\caption{Operational Semantics of \flowcore}
\label{fig:flow:opsem}
\end{figure}

Figure~\ref{fig:flow:opsem}
contain a small-step operational semantics for
programs in \flowcore.
The rules can have the forms: \stepsconf{\rtstate}{\expr}{\rtstate'}{\expr'}
and  \stepsconf{\rtstate}{\stmt}{\rtstate'}{\stmt'}.

Next we describe some of the most interesting rules.
Rule~\hyperref[rule:rtevar]{\rtevar} shows the indirection
in dereferencing variables. First the store \store is looked
up and then the resulting location is used to access the
heap \heap. Similarly variable assignments have to go through the
same process in Rule~\hyperref[rule:rtasgn]{\rtasgn}.
Here, symbol $\triangleleft$ denotes the update of state \rtstate
with the new heap $\heap'$.

When evaluating arrows, the current store
\confproj{\rtstate}{\store} is saved as part of the created
closure, along with the code \exprbody of the function
(Rule~\hyperref[rule:rtarrow]{\rtarrow}).
This store is restored when the function is
called~(Rule~\hyperref[rule:rtcall]{\rtcall}).
The new store $\store'$ that will be used in the new stack
frame also includes a binding for the function parameter
\evar and bindings from all variables \many{\evar_i} defined in
the body \exprbody, since their definition is hoisted to the
top of the function scope.
We use metavariable \localssym to extract these variables.
All new variables are bound to fresh locations $\loc_i$.
Locals have not been initialized yet, so their
locations are bound to \inlinejs{undefined} in the
heap $\heap'$. Finally, a new stack frame $\store.\ectx$ is pushed
on the existing stack \stack as we enter the new function context.
After returning from this function, execution will return to \ectx
(Rule~\hyperref[rule:rtret]{\rtret}).
The rest of the expression reduction rules are routine.

\section{Metatheory}\label{sec:type:soundness}

In order to prove type safety for our type system we first introduce
a declarative type system that closely matches the intuition of the
type inference system described in \secref{sec:cons:sys}.
Based on the declarative system we then formulate a
type safety argument for the above language fragment
via a progress and a preservation theorem~\cite{Wright1994},
that connect type checking with the
runtime semantics of \secref{sec:semantics}.
Essentially, we establish the fact that if a program has
been checked with the above algorithm and has been found
consistent, then its execution will not lead to uncaught
type errors (\eg, ``undefined is not a function'').
Introducing this intermediate step in our metatheory is not
mandatory, but it vastly reduces the complexity of reasoning
about type safety, compared to the inference
version of \secref{sec:cons:sys}.

\subsection{Declarative Type System}
\label{sec:declarative}

This system assigns concrete types, \ie, types stripped off
of type variables, to language terms of
\flowcore.
%
%
Environments \env also map program variables to
concrete type entries (where both flow-sensitive and flow-insensitive
types are concrete).
The same holds for effects \effect.
The typing judgments for expressions and statements are:
\tcexprsingle{\env}{\expr}{\type}{\effect}{\predmap}{\env'}
and
\tcstmtsingle{\env}{\stmt}{\effect}{\env'}.
%
The respective rules for these judgments follow the main intuitions of
the inference system and
are therefore deferred to \secref{sec:app:declarative},
along with other attendant definitions.

A \emph{substitution} \gsubst maps type variables of the
inference system to concrete types of the declarative system, and can
be extended to types, effects and environments in a point-wise manner.
Constraints \cons in the inference system correspond to
subtyping relations in the concrete system for both types and effects.
We can use the same substitution \gsubst to convert a constraint \cons
to one or multiple subtyping constraints over concrete types or effects.
Since our type language has been kept simple overall, the subtyping
rules for concrete types are routine and so a discussion is deferred to
the appendix.
We say that a substitution \gsubst \emph{satisfies} a constraint set \cset if all
subtyping constraints generated by mapping \gsubst over \cset
are valid. In this case we write \satisfies{\gsubst}{\cset}.

We argue about the soundness of our type inference system with
respect to the declarative system with the following lemma.

\begin{lemma}[Soundness of Type Inference]
\label{lemma:main:flow:type:inf:soundness}
If
\jexpr{\env}{\expr}{\type}{\effect}{\predmap}{\cset}{\env'}
and there exists substitution \gsubst \st
\satisfies{\gsubst}{\cset},
then
\tcexprsingle{\appsubst{\gsubst}{\env}}
    {\expr}{\appsubst{\gsubst}{\type}}
    {\appsubst{\gsubst}{\effect}}{\predmap}
    {\appsubst{\gsubst}{\env'}}.
\end{lemma}

\subsection{Type Safety}

Before we state our type safety result for the declarative
system, we extend the type checking judgment to
runtime configurations:
\tcrtconfexpr{\genv}{\heapty}{\rtstate}{\expr}{\type}.
Here \genv is a flow-insensitive environment mapping variables to
their most general concrete type throughout the program.
The judgment is to be read as: under a
\emph{heap typing} \heapty, mapping
heap locations to types and a flow-insensitive environment \genv,
a configuration
\mkconf{\rtstate}{\expr} is a assigned a type \type.
We can now state our type safety result.

\begin{theorem}[Type Safety]\label{thm:flow:type:safety}
For a configuration \mkconf{\rtstate}{\expr} and heap typing \heapty,
if \tcrtconfexpr{\genv}{\heapty}{\rtstate}{\expr}{\type},
then:
\begin{itemize}[noitemsep,nolistsep,leftmargin=*]
  \item \textbf{(Preservation)} If \stepsconf{\rtstate}{\expr}{\rtstate'}{\expr'}, then
		there exists $\heapty'$, such that
    \tcrtconfexpr{\genv}{\heapty'}{\rtstate'}{\expr'}{\type'}.
  \item \textbf{(Progress)} Either \expr is a value, or there exists
    a configuration \mkconf{\rtstate'}{\expr'} such that
    \stepsconf{\rtstate}{\expr}{\rtstate'}{\expr'}.
\end{itemize}
\end{theorem}

Supporting lemmas and proofs for the above results  can be found in the appendix.

\section{Type Annotations}\label{sec:annotations}

So far, we have described a system for \emph{type inference} that
ensures that values are used in ways that are consistent with their
definitions. But can we check that the inferred types of values are consistent
with types we specify?

In this section, we introduce a system for \emph{type checking}. Type
annotations \type follow a similar grammar as types except that there are no
type variables, types can appear anywhere type variables could appear, and there
are no effects. We consider a type annotation to be just another kind of type
use, that expects some type of values. In other words, like everything else we
can formulate type checking with flow constraints.

Technically, we need some new propagation rules for flow constraints involving
type annotations.

When we see a constraint of the form
\flows{\lb}{\tobj{\fieldsym_1}{\type_1}{\fieldsym_n}{\type_n}}, we propagate it
with new constraints
\flows{\lb}{\getuse{\fieldsym_1}{\tvar_1}},
\flows{\tvar_1}{\type_1},
\flows{\lb}{\setuse{\fieldsym_1}{\type_1}},
\dots,
\flows{\lb}{\getuse{\fieldsym_n}{\tvar_n}},
\flows{\tvar_n}{\type_n},
\flows{\lb}{\setuse{\fieldsym_n}{\type_n}}. (As mentioned \secref{sec:cons:gen},
when \lb is a record type, these flow constraints can be replaced by
unification constraints as an optimization.)

The remaining propagation rules rely on some new definitions.

\paragraph{Escaping effects}

Since a function type annotation has nothing to do with any particular function
expression, we cannot calculate its effect in the usual manner.
Instead, we assume that there is an effect variable $\effectescape$ that
captures ``escaping'' effects, and that all function type annotations have this
effect. When we see a constraint of the form
\flows{\lb}{\tarrow{\type_1}{}{\type_2}}, we propagate it with the new
constraints
\flows{\lb}{\calluse{\type_1}{\effectescape}{\tvar_2}} and \flows{\tvar_2}{\type_2}.

\paragraph{Conditional flow constraints}

Checking that an inferred type is consistent with a union type annotation is
tricky. Intuitively, this amounts to checking that the inferred type is
consistent with either case of the union type annotation. But since the inferred
type may contain type variables, it may not be obvious which case to pick.
Consider the code in Figure~\ref{fig:type:annotations} (left),
where the parameter \inlinejs{f} on \mylineref{src:onstring} has a
type that is the union of two function types.
The call on \mylineref{src:onstring} is safe, since both function types take
\inlinejs{string}. However, it is unclear which choice of function type to use
for \inlinejs{id} on \mylineref{src:onstring:call}. Picking \inlinejs{IDString} seems fine
``locally,'' but turns out to be the wrong choice since \inlinejs{null} is
passed on the next line and picking \inlinejs{IDNullableString} instead would
type check.
\begin{figure}
\begin{minipage}[l]{.50\textwidth}
  {\footnotesize
\begin{jscode*}{escapeinside=++, fontsize=\footnotesize}
+\kwtype+ IDString = (string) => string;
+\kwtype+ IDNullableString = (?string) => ?string;
+\kwtype+ Ambiguous = IDString | IDNullableString;

+\label{src:onstring}+function onString(f: Ambiguous) { f(""); }
var id = (x) => x;
+\label{src:onstring:call}+onString(id);
id(null);
\end{jscode*}
}
\end{minipage}
\hfill
\begin{minipage}[r]{.47\textwidth}
  {\small
\begin{jscode*}{escapeinside=++, fontsize=\footnotesize}
+\kwtype+ Correlated
  = { type: "string", val: string }
  | { type: "number", val: number };

function stringIsString(x: Correlated) {
  if (x.type === "string")
+\label{src:displaystring}+    displayString(x.val);
}
+\label{src:stringisstring}+stringIsString({ type: "string", val: 0 });
\end{jscode*}
}
\end{minipage}
\caption{Type Annotations}
\label{fig:type:annotations}
\end{figure}

Alternatively, we could propagate the choice further, effectively introducing
disjunction in the logic of flow constraints. While appealing from a theoretical
perspective, in practice this approach is complicated to implement and difficult
to scale. It also doesn't mesh well with refinements. Consider the
code in Figure~\ref{fig:type:annotations} (right).
%
%
%
Here, we propagate the choice of type for the object argument on \mylineref{src:stringisstring}
into the fields of the object types in the union \inlinejs{Correlated}; this
causes the call to type check, since the \inlinejs{type} and \inlinejs{val}
fields of the object separately typecheck against the corresponding
unions. However, this is unsound, as the call on \mylineref{src:displaystring} shows (it passes a
number at run time where a string is expected).

Instead, our approach is to try picking a case without ambiguity (\ie,
without considering type variables), or demand further type annotations to
disambiguate.
Specifically, let any constraint of the form $\flows{\tvar}{\star}$ or
$\flows{\star}{\tvar}$ be a \emph{condition}: its validity is conditional on
what type $\tvar$ is inferred to be. Suppose that we have a restricted form of
constraint propagation without the transitivity rule
\hyperref[rule:cptrans]{\cptrans}, so that conditions are not propagated. The
propagation rule for choice, described below, uses this restricted form of
constraint propagation to generate a set of constraints. If the rule signals an
ambiguity, then the developer must provide annotations for the type variables
involved in any generated conditions. Otherwise, the generated conditions are
propagated further using the unrestricted rules.

A constraint of the form \flows{\lb}{\tjoin{\type_1}{\type_2}} is propagated as
follows:
\begin{itemize}
\item Either \flows{\lb}{\type_1} generates an inconsistent set of constraints; then we continue with \flows{\lb}{\type_2}.
\item Or \flows{\lb}{\type_1} generates a consistent set of constraints $\cons_1$.
  \begin{itemize}
  \item Either \flows{\lb}{\type_2} generates an inconsistent set of constraints; then we continue with $\cons_1$.
  \item Or \flows{\lb}{\type_2} generates a consistent set of constraints $\cons_2$.
    \begin{itemize}
    \item Either $\cons_1 \subseteq \cons_2$; then we continue with $\cons_1$.
    \item Or $\cons_1 \not\subseteq \cons_2$; then we signal an ambiguity.
    \end{itemize}
  \end{itemize}
\end{itemize}

\section{Modules and Dependencies}\label{sec:modules-deps}

Until now, we have presented \flow's analysis on ``whole'' programs. However,
\js codebases can be quite large (\eg, at \facebook we have millions of lines of
\js code), and a whole-program set-based analysis is simply not fast enough at
scale.

In this section, we show how we modularize \flow's analysis. Modularization is
important for performance, in terms of both time and space. It is also a natural
fit for modern \js, where code is typically split across a (large) number of
(small) files; every file is mapped to a module, and possibly imports other
modules; definitions are local by default, unless they are exported; and
accessing global definitions other than builtins is generally discouraged.

Broadly, we follow the standard approach of analyzing each file separately, once
all files it depends on have been analyzed. This strategy allows us to
incrementally analyze the program as files change (\secref{sec:incremental}),
and parallelize the analysis across files (\secref{sec:parallel}).

The key idea is to demand a ``signature'' for every module. We ensure that types
inferred for the expressions a module exports do not contain type
variables---wherever they do, we demand type annotations. For example, the
parameter of an exported function expression must specify its type. Otherwise,
we risk having the type of the parameter depend on calls in other files that
import this module, which breaks modularization. (Alternatively, we could try
generalizing the parameter's type based on how it is used inside the function
expression, but in our experience it leads to unwieldy types.)

Requiring annotations for module interfaces is much better than requiring
per-function annotations: a typical module exports one object or function, while
having a bunch of module-internal code. The minimum annotation burden is only a
small fraction of lines of code. (Of course, annotations are permissible even
where they are not required.)

The type annotation syntax is designed to be nearly as expressive as the
internal type language, but it is not true that \emph{all} well-typed code can
be cleanly refactored into modules to preserve typing. This is by design:
modules are abstractions, and so some rewriting might be needed to prevent
leaking abstractions. For example, effects do not appear in types at all, and
local effects that leak after refactoring would need rewriting.

Independently, having a signature for every module turns out to be a desirable
choice for software engineering. It is considered good practice for
documentation (files that import the module can simply look up its signature
instead of its implementation), as well as error localization (blames for errors
do not cross module boundaries).

\paragraph{Modules, exports, and imports}

For the purposes of this paper, let us assume that every file maps to a module
by the same name. (This is sufficient to model the popular CommonJS module
system, which is the default module system in \flow. However, we also support
module systems where module names are independent of file names and not
necessarily in 1-1 correspondence.)

A file can import definitions that another file exports. However, the reference
$m$ to the exporting file inside the importing file is \emph{relative} to the
importing file.

\paragraph{Module loading and dependency tracking}

We assume a module loading judgment $\mathbb{FS} \Vdash \file_i::m \leadsto
\file_e^?$ that, given a file name $\file_i$ and a reference $m$, either
computes the name $\file_e$ of the file being referenced by $m$ in $\file_i$ or
errors, while recording any files looked up by the derivation in $\mathbb{FS}$
(which may or may not exist). Whenever $\mathbb{FS} \Vdash \file_i::m \leadsto
\file_e$, we assume $\file_e$ exists and have $\file_e$ in $\mathbb{FS}$.

Given a file system state that satisfies the constraints $\mathbb{FS}$, we say
that $\file_i$ \emph{depends} on $\file_e$ whenever $\mathbb{FS} \Vdash
\file_i::m \leadsto \file_e$.

\paragraph{Compilation and linking}

Files are ``compiled'' and ``linked'' in dependency order. In practice, there may be
cycles in the dependency graph, so this process is run on the directed acyclic
graph of strongly connected components, where each strongly connected component
is considered to have all the imports and exports of the files in it.  For each
file, compilation and linking generates exported types and \emph{signature}
constraints of the form $\flows{\tlit}{\tvar}$, where any type variables in
\tlit are in positive positions. Such signature constraints fully describe the
types of exports of the file, which can be ``substituted'' for the types of
corresponding imports in dependent files.

Suppose that compiling a file \file, with module references typed as fresh type
variables $\tvar_1, \dots, \tvar_n$, generates constraints \cons and an exported
type \type. Furthermore, suppose that the module references resolve to files
$\file_1, \dots, \file_n$ that have their signature constraints $\cons^\star_1,
\dots, \cons^\star_n$ and exported types $\type_1, \dots, \type_n$.

Then, we link \file by adding $\cons^\star_1, \dots, \cons^\star_n$ to \cons, propagating the
constraints $\type_1 \leq \tvar_1, \dots, \type_n \leq \tvar_n$, and transforming
\cons to $\cons^\star$ by the process $\sfsym{signature}(\type)$, defined
recursively as follows:
\begin{enumerate}

\item $\sfsym{signature}(\tvar)$ throws away upper bounds of \tvar, and calls
  $\sfsym{signature}(\tlit)$ for each \flows{\tlit}{\tvar}.

\item $\sfsym{signature}(\tlit)$ demands a type annotation for each type
  variable in a negative position in \tlit, and calls $\sfsym{signature}(\tvar)$
  for each \tvar in a positive position in \tlit.

\end{enumerate}

Intuitively, this process walks over constraints, with the exported types as
roots, while doing two things. One, any constraints that would be unreachable
when the exported types flow to type uses in dependent files are pruned
away. Two, the developer must provide annotations wherever constraints could
propagate back from dependent files.

Formally, a key property of signature constraints $\cons^\star$ and exported
types \type is that it is impossible for a constraint of the form
\flows{\type}{\use} to lead to, via propagation, a constraint of the form
\flows{\tlit}{\tvar} where \tvar is in $\cons^\star$. In other words, signature
constraints and exported types can be considered ``closed'' when linking
dependent files. This means that they do not need to be recomputed for
correctness---in fact, they can be memoized and reused---which is crucial for
performance. This also means that whenever ``diamonds'' occur in the dependency
graph, \ie, \file depends on another file via multiple paths, the signature
constraints of \file are the same no matter which order the paths are explored.

\section{Incrementalization}\label{sec:incremental}

In this section, we show how \flow exploits modularity and dependency management
to incrementally analyze files as they change.

\paragraph{Architecture}

\textsc{Flow}'s architecture consists of a server, a client, and a file system watcher.
The server initially analyzes the entire codebase, following the procedure in
the previous section, and stores a bunch of information in memory. The
information not only includes the status (type errors), but also the results of
separately compiling and linking every file, and the dependencies between
files. Once the server is initialized, it runs in the background.

The client queries the server for information, relaying commands issued via the
command-line or various IDEs. Typically, the client is interested in the
status. But the client could also ask for the type at a particular position, the
definitions reaching a particular reference, \etc, in which case the server
computes that information almost instantaneously from the information already
stored in memory.

Finally, the file system watcher informs the server of changes to the file
system: which files have been added, modified, or deleted. Based on this
information, the server re-analyzes a (hopefully small) fraction of the code
base, and updates the information stored in memory. (When the client queries the
server again, the response is based on this updated information.)

\paragraph{Incremental analysis}

Files that are added or modified need to be re-analyzed. In addition, a subset
of unmodified ``dependent'' files needs to be re-analyzed. This set can be partitioned into
\emph{direct} dependents and \emph{indirect} dependents.

Direct dependents are computed as follows. For any file \file that is added,
modified, or deleted, whenever $\file$ is in $\mathbb{FS}$ for any unmodified
$\file_i$ where $\mathbb{FS} \Vdash \file_i::m \leadsto \file_e^?$, we consider
$\file_i$ is a direct dependent.
The \emph{depends}
relation is modified by re-resolving module references in direct
dependents. Indirect dependents are unmodified files that recursively depend on
direct dependents.

\section{Parallelization}\label{sec:parallel}

We use a map/reduce algorithm augmented with shared memory communication for
parallelizing various stages of type checking, building on and extending \hack's
model~\cite{hack:spec}.

\subsection{Workers}

Assume we want to distribute a task on a number of files, computing a
result of type \resultmt. We describe the task with a function
\kw{job}\oftype{\filesmt \rightarrow \intermt}; a function
\kw{merge}\oftype{\resultmt \times \intermt \rightarrow \resultmt}; a value
\kw{neutral}\oftype{\resultmt}; and a function \kw{next}\oftype{\unitmt
  \rightarrow \filesmt}.

We have a master process and as many worker processes as the number of available
processors. The master initially has \resultmv as \kw{neutral}, and considers
all workers free. Then, it repeatedly does the following:
\begin{itemize}
\item If there is a free worker, call \applymv{\kw{next}}{} to obtain a list of
  files \filesmv. If \filesmv is empty and all workers are free, exit with
  result \resultmv. If \filesmv is not empty, send it to the free worker.
\item If a worker has sent back \intermv, consider the worker free, run
  \applymv{\kw{merge}}{\pair{\resultmv}{\intermv}} and update \resultmv with it.
\end{itemize}
Correspondingly, every worker repeatedly does the following. If the master has
sent \filesmv, fork a process to run \applymv{\kw{job}}{\filesmv}, wait for a
value \intermv, and send back \intermv to the master.

Usually, the \kw{next} function is simple. It just remembers an index into the
original list of files. When the index is out of bounds, it returns an empty
list; otherwise it returns a sublist from that index of some fixed ``bucket''
size, and advances the index. This models processing a static list of files in
no particular order.

What if the processed list is computed dynamically, in a particular order? For
that, we make the following changes. We maintain a ``worklist'' of files, and
have \kw{next} create a bucket from the worklist instead. Let the intermediate
result type $\intermt'$ be $\filesmt \times \intermt$. Then
\applymv{\kw{job}'}{\filesmv} returns
\parens{\pair{\filesmv}{\applymv{\kw{job}}{\filesmv}}}. Finally,
\applymv{\kw{merge}'}{\pair{\resultmv}{\parens{\pair{\filesmv}{\intermv}}}}
updates the worklist with \filesmv, before returning
\applymv{\kw{merge}}{\pair{\resultmv}{\intermv}}.

\subsection{Shared Heap}

As described above, the master and the workers communicate by serializing and
deserializing data like files and intermediate results. In practice, the results
that need to be computed are often maps from files to large values, and thus
communicating intermediate results from the workers back to the master becomes a
performance bottleneck. Furthermore, manipulating large results in the master
spikes its memory usage and causes frequent garbage collection pauses, which
affects the server's responsiveness.

Thus, we use a different mechanism for sharing results: a large hashtable mapped
to RAM, accessible to both the master and the workers, that is logically divided
into various maps. The hashtable provides fast (lock-free) concurrent access for
reads, as well as for writes as long as they \emph{add} entries for
\emph{disjoint} keys. Only the master can remove entries. These conditions turn
out to be easily satisfied in our setting: at any stage of type checking,
different workers always operate on different files, and old entries are only
ever cleared to make way for new entries when processing file system changes in
the master.

Entries are compressed on writes. (We use LZ4~\cite{LZ4} because it is extremely
fast, while providing sufficient compression.) In practice, this means we can
tolerate redundancy in the entries, trading off space for time by precomputing
information. Moreover, entries are cached on reads.

With the shared heap, the types \resultmt and \intermt can be quite small
(typically, metadata for bookkeeping). Moreover, the processes forked by workers
to run jobs are short-lived: their memory is reclaimed by killing them on every
completion.

\subsection{Parallelizing Parsing}

Files are parsed in parallel, in no particular order, using a static \kw{next}. Every \kw{job} writes the abstract syntax trees for corresponding files to shared memory.

\subsection{Parallelizing Compilation and Linking}

Next, the files are compiled and linked in parallel, following dependency order, using a dynamic \kw{next}. We maintain a count of dependencies for each file. The worklist initially contains files that do not depend on any other files. As files are processed, they decrement the dependency counts of other files that depend on them, possibly causing those files to be added to the worklist (because they no longer depend on any files that have not already been processed).

Every \kw{job} reads the abstract syntax trees for the corresponding files, and the signature flow constraints for the files they depend on, does compilation and linking, and finally writes the signature flow constraints for the corresponding files to shared memory.

Note that a file does not need to be rechecked if the signature flow constraints
for none of its dependencies change. This is a major optimization: it means that
even though a file may have a lot of recursive dependents, only a small fraction
of them may actually need to be rechecked. Indeed, in practice it
results in order-of-magnitude differences in recheck times. The implementation
is slightly tricky because new signature flow constraints often contain fresh
type variables that make them trivially different than old signature flow
constraints. We compute and compare hashes modulo such trivial differences to
detect when signature flow constraints have changed.

Compiling and linking files in dependency order limits some parallelism in theory, but in practice, the alternative approach of processing every file independently ends up doing far more work. Overall we save processing time by an order of magnitude.

\section{Experiments}\label{sec:experiments}

We ran experiments on the main internal repository at \facebook, in which (at
the time of writing this paper) around 13M LOC of \js are covered by \flow,
spanning around 122K files. We chose this repository because it contains a wide
variety of \js projects, implementing client code and frameworks for web
applications, that depend on each other but are owned by different teams and do
not necessarily conform to a uniform coding style. (A smaller repository
contains client code and frameworks for mobile applications, which is also
covered by \flow, but which we do not consider here---the general conclusions
about \flow's behavior, however, remain the same.)

\paragraph{Distribution of annotations}

\begin{figure*}[t!]
    \centering
    \begin{minipage}{0.48\textwidth}
        \centering
        \includegraphics[width=\textwidth]{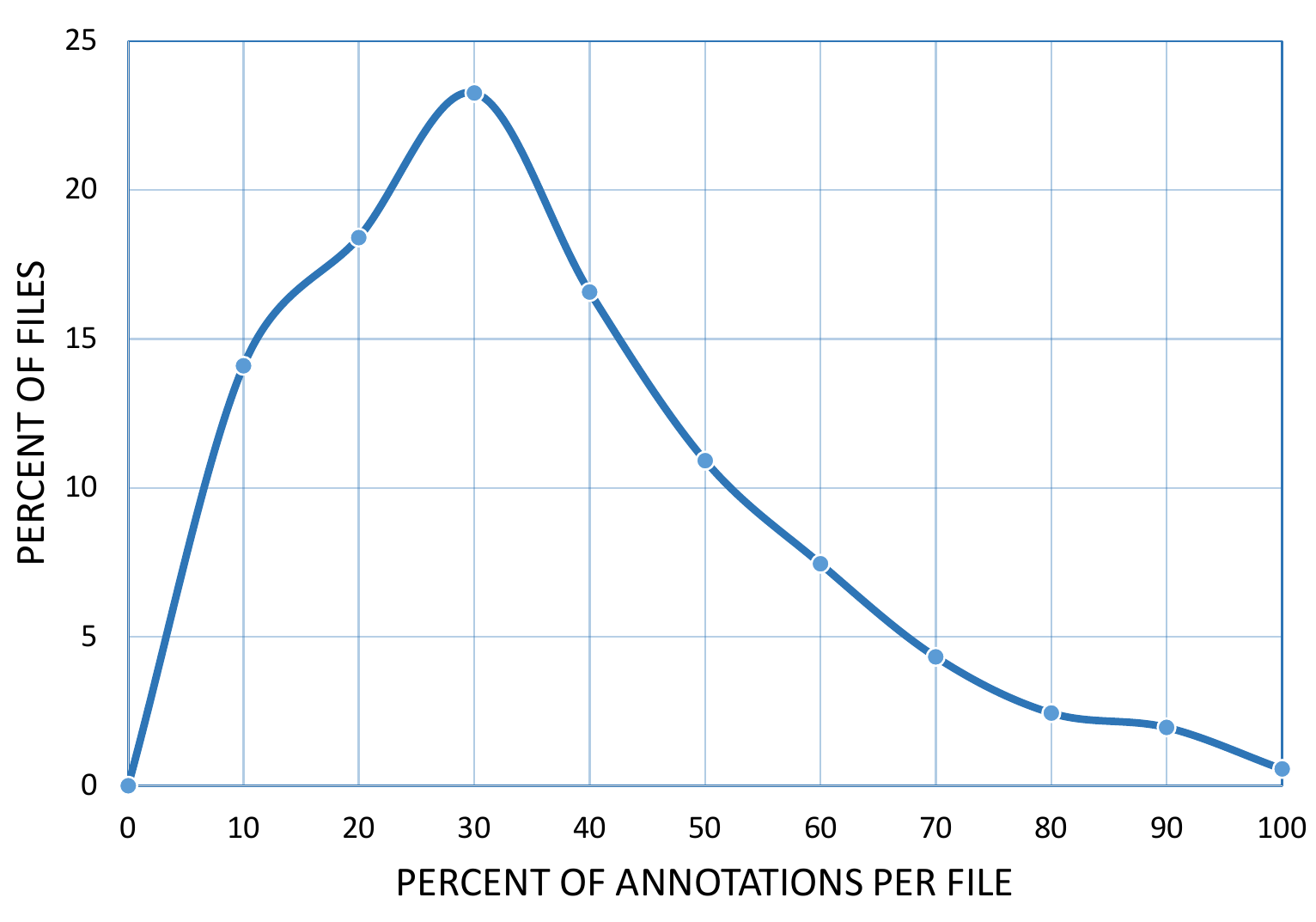}
        \caption{Distribution of annotations}
        \label{fig:distribution-of-annotations}
    \end{minipage}%
    \hfill
    \begin{minipage}{0.46\textwidth}
        \centering
        \includegraphics[width=\textwidth]{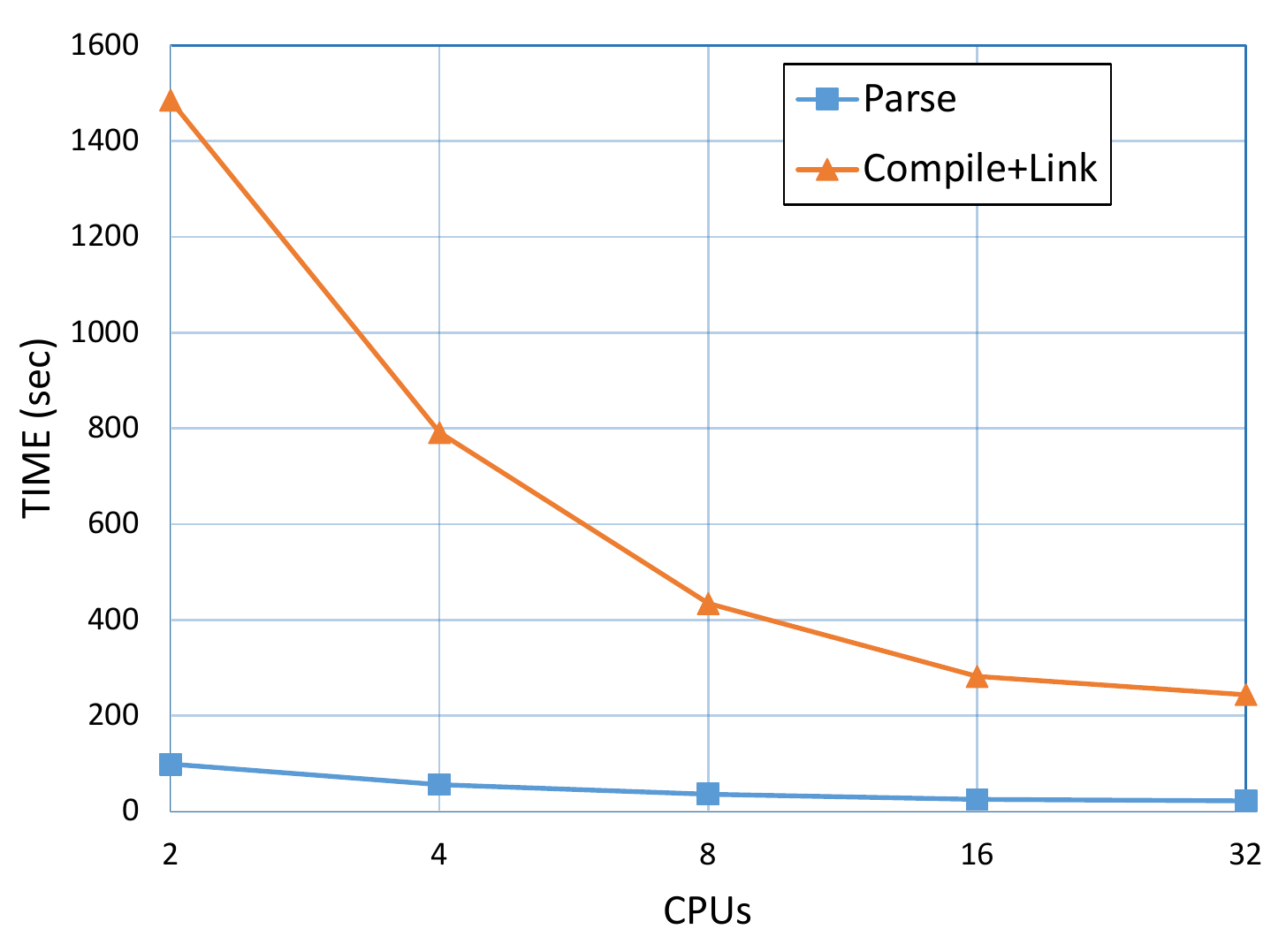}
        \caption{Effects of parallelization on performance}
        \label{fig:procs-times}
    \end{minipage}
\end{figure*}


For each file, we counted the number of annotations as a fraction of the total
number of locations that could potentially be
annotated. \figref{fig:distribution-of-annotations} shows the distribution of
these relative numbers across the repository. The median is 29\% annotations. In
aggregate, there are around 686K annotations vs. 1502K other locations that
could potentially be annotated but are not. Note that these numbers do not
account for type inference of every expression---only declaration sites are
considered---and thus conservatively overapproximate the annotation burden.

\paragraph{Usage of refinements}

As a quick test of the impact of supporting refinements in \flow, turning it off
led to more than 145K spurious errors across this repository.

\paragraph{Effects of parallelization on performance}


Total initialization time was around 225s with 32 processors (Xeon, 2.2GHz),
which means around 3 ms per file on average. We also used around 4GB of shared
memory.

\figref{fig:procs-times} shows how performance varies with the number of
processors (we only measured performance for powers-of-2 processors and
extrapolated). Parsing times become half from 2 to 4 processors, but then the
improvements slow down (since parsing is already quite fast, and communication
starts dominating---an example of Amdahl's law). Compilation and linking times
hit the limit much later, continuing to improve almost linearly until 16
processors. (This makes sense, since compiling and linking is much slower than
parsing.)

\paragraph{Effects of code size on performance}

The time taken to compile and link a file grows approximately linearly with code
size. While set-based analysis is well-known to be worst-case cubic in code
size, the lines of code per file is small (average 106), and the imports are
summarized by signature constraints, which are much smaller than the
corresponding lines of code of dependencies.

\paragraph{Effects of dependencies on performance}

The time taken to process a file is not correlated with the number of recursive
dependencies---the spread of times remains relatively flat as dependencies
grow. This is not very surprising, since signature constraints of imports are
supposed to ``compress'' the information contained in recursive dependencies. On
the other hand, the time taken to process a file grows approximately linearly
with the number of imports.

The size of the signature constraint graph, which form the results of compiling and linking, is not correlated with code size---for most files, the sizes are between 50-100KB.

\paragraph{Effects of incrementalization on performance.}

The performance of incremental type checking is tied to the number of files that
are rechecked when a file is modified, \ie, the number of files that recursively
depend on that file.

The distribution of the number of recursive dependents is highly positively
skewed---the median is less than 10, and the \nth{90} percentile is less than
100. Considering the time taken to recheck to be roughly proportional to the
number of files to link, this means that in 90\% of cases, recheck time is less
than 200ms plus some constant. Of course this is only a rough calculation: in
practice, while this closely approximates the common scenario of editing single
files, it doesn't account for occasional rebases that can cause larger numbers
of files to be rechecked.

\section{Related Work}\label{sec:related}


There has been a lot of work on type systems for \js and related languages, as
well as set-based inference techniques.
We focus here only on the most closely related work.


\subsection{Mainstream Type Systems for Dynamic Languages}
\citet{ts:design:goals} is a widely used typed superset of \js. Like \flow, it aims to improve
developer productivity by providing tooltips through IDEs. Unlike \flow, it
focuses only on finding ``likely errors'' without caring about soundness
\cite{Bierman2014}.
Type inference in \ts is mostly local and in some cases contextual; it doesn't
perform global type inference like \flow, so in general more annotations are needed.
Whenever type annotations are
missing, they are considered to be \kw{any} (instead of being implicitly
inferred). Thus, many type errors are missed. Consider for example the program:
\begin{jscode*}{
  linenos=false,
  xleftmargin=.25in,
  xrightmargin=.25in,
}
function square(n) { return n * n; }
square("oops");
\end{jscode*}
\flow will signal an error for trying to use a string as an argument to a
multiplication. \ts, on the other hand, will infer \kw{any}
as the type for \inlinejs{n} and accept this program as valid.
To get a similar behavior from \ts we would need to add a
type annotation to the parameter \inlinejs{n} of \inlinejs{square}.

Furthermore, even with fully
annotated programs, \ts misses type errors because of unsound typing rules. For
example, ``bivariant'' subtyping means that functions and instances of
polymorphic classes can be passed to contexts that do not preserve their typing
invariants, as can be seen in the following erroneous example:
\begin{tscode}
var assertString = (x: string) => assert(typeof x === "string");
var app = (f: (x: string | number) => void, x: number) => f(x);
app(assertString, 1);
\end{tscode}
A checker that implements sound contra-variant argument subtyping, like \flow,
would signal an error at the call to \inlinejs{app}, since
\inlinejs{assertString} is incompatible in its argument with the expected type
for parameter \inlinejs{f} of \inlinejs{app}.
In practice, this means that \ts developers have to code defensively
with dynamic checks, even when types are included.
\safets~\cite{Rastogi2015} ``fixes'' soundness problems in \ts with stricter
typing rules and runtime enforcement mechanisms to restore gradual
typing.

\citet{dart:spec} is another language that shares the
same philosophy. Unsoundness is a deliberate choice in \ts and \dart, motivated
by the desire to balance convenience with bug-finding. But we have enough
anecdotal evidence from developers at \facebook that focusing on soundness is
not only useful but also desirable, and does not necessarily imply
inconvenience.
Similar to \safets, recent work recovers soundness in \dart
~\cite{Heinze2016}.

\citet{closure:spec} is another widely used type system for \js
that focuses on transforming code for size reduction. As far as we can tell, it
is sound modulo similar assumptions as \flow, but lacks type
inference. \racket \citep{Tobin-Hochstadt:2008}, and \citet{hack:spec} (for
PHP) are also quite close in spirit: their optional typing is at the level of
modules and they use occurrence typing to perform similar kinds of
refinements. They differ in that they lack type inference and,
compared to \flow, their treatment of mutable variables is far more
simplistic---there is no distinction between mutability on the stack and on the
heap. On the other hand, \flow heavily borrows from \hack's design and
implementation for scaling to millions of lines of code.

\subsection{Research Static Analysis for \js}


\paragraph{Early work}

Several static typing systems have successfully been ported to the dynamic
setting of \js. Early work by \citet{Thiemann2005} and \citet{Anderson2005}
focus on restricted subsets of the language.


\paragraph{Type Refinement}

In the area of type refinement,
\citet{Guha2011} develop \emph{flow typing}, which, like \flow, supports type
narrowing as a consequence of control flow.
Unlike \flow, their analysis is strictly intra-procedural, does not perform type
inference and does not track non-local effects (\eg, variable updates).
Building on this work, \citet{Lerner2013} present a
framework for building type systems for \js, engineered modularly to encourage
experimentation, but which also suffers from limited type inference compared
to \flow.

\paragraph{Static Objects}

\citet{Choi2015} propose a static type system for ahead-of-time compilation of
\js that guarantees fixed object layout. Its type inference is based on very
similar foundations as \flow. \sjs focuses mainly on taming legacy
object-oriented features (constructor functions, open methods, and prototype
inheritance). While \flow does support these features, its model and guarantees
are different---it models these features with \emph{extensible} objects, and
guarantees type consistency of gets and sets of properties (but necessarily
their existence). Many of these concepts are replaced by classes in ``modern''
(ES6+) \js, where \flow can provide stronger guarantees with advanced type
system features like bounded polymorphism, \kw{this} types, and read-only
properties.
\citet{Chandra2016} build on this work by adding support for abstract objects,
first-class methods, and recursive objects, and prove their extensions sound.
Their type system supports additional features such as polymorphic arrays,
operator overloading, and intersection types in manually-written interface
descriptors for library code, that they found important for building
GUI applications.
Their formalization focuses on their object model.
Unlike \flow, they do not discuss type refinement based on conditional
checks, and their formalization is flow-insensitive, so less precise in
that respect compared to \flow.

\paragraph{Abstract Interpretation}

This is a more heavyweight approach in statically analyzing \js.
Approaches here include
TAJS~\citep{Jensen2009, Jensen2010, Jensen2011, Andreasen2014},
JSAI~\citep{Kashyap2014}
and SAFE~\citep{Lee2012, Park2015}. These approaches vary in precision,
user customizability and flow-, context- and path-sensitivity,
but being whole-program analyses, they are out of scope at our
scale, while being much more precise and not needing annotations.





\paragraph{Program Logics}

Recent advances in SMT solver technology has spurred the interest in using program
logics to track the dynamic behavior of \js programs.
DJS~\citep{Chugh2012} combines nested refinements with alias types~\citep{Smith2000},
a restricted separation logic,
to account for aliasing and flow-sensitive heap updates to
obtain a static type system for a large portion of \js.
DJS comes with very limited type inference and hence requires complex annotations
at function and loop boundaries.
To reduce this annotation burden, \citet{Vekris2016}
offer refinement type inference for \ts, based on Liquid Type
inference~\citep{Rondon2008}. Their type system is more expressive than
\flow's as it allows logical predicates (taken from a number of decidable
logics) to be attached on a base type system (a subset of \ts's type system).
However, this precision comes at the cost of a higher annotation burden and a
penalty on scalability.
Also, while offering global refinement type inference,
unlike \flow it does not infer the base (underlying) types of programs
and requires explicit immutability annotations.

\subsection{Inference and Subtyping in Dynamic Languages}

\paragraph{Constrained Types} 

Set constraints have been used for the purpose of type inference by
\citet{Aiken1992} and \citet{Aiken1994}, who adopt the set-theoretic model to infer types in
a simple functional language.
\citet{Trifonov1996} and \citet{Pottier1998} infer polymorphic recursively
constrained types, but retain a simpler interpretation of type terms. In their
work, ground types are regular terms, and subtyping is defined explicitly on
terms. This enables various simplifications to their constraint sets, like
garbage collection~\citep{Eifrig1995, Pottier1998, Pottier2001}.
\citet{Flanagan1999} use a simpler type representation and, based on simplification
algorithms that exploit the observable equivalence of constraint sets, perform
\emph{componential} set-based analysis.

\flow builds directly on work by \citet{Pottier2001}, but does not infer
polymorphic types. Instead, it exposes features less frequently
addressed in
the context of set-constraint based analyses, such as variable updates and type
refinement based on conditional checks.
In addition, \flow's analysis is not context-sensitive, due in part to
anecdotal concerns about performance in \drjs~\cite{Vardoulakis:Thesis}.  In
practice, polymorphic type annotations recover context-sensitivity where
needed.

\paragraph{Constraint Graph Simplification}

The research directions above 
already include several
simplification techniques~\citep{Fahndrich1996, Flanagan1997}.
To further improve performance of inclusion constraint analyses,
\citet{Fahndrich1998} propose a
technique for eliminating cycles in constraint graphs
that is based on a non-standard graph representation called \emph{inductive form},
and only traverses part of the paths during the search for cycles.
To address the problem of redundant paths in a constraint graph, \citet{Su2000}
propose \emph{projection merging}, a technique intended to be used in conjunction with
the above.
In contrast, we directly implement
unification constraints using union-find over a base representation of inclusion
constraints.

\paragraph{Semantic and Algebraic Subtyping}

Advances in semantic and algebraic foundations have spurred
renewed interest in this rich area.
Semantic subtyping has been proposed in the context of
functional languages for XML based programming~\citep{Frisch2008},
ML-like languages~\citep{Castagna2016},
and more recently for imperative object-oriented languages, where
fields can be mutable~\citep{Ancona2016}, and in a gradual
typing setting~\citep{Castagna2017}.
Even though polymorphic type inference with subtyping is known to be
undecidable~\cite{Su2002}, \citet{Dolan2017} infer compact principal types by
keeping a strict separation between the types used to describe inputs and those
used to describe outputs (polarities).
In comparison, \flow is less ambitious with
union and intersection types.

\section{Limitations and Threats to Validity} \label{sec:limitations}

We conclude this paper by discussing limitations and
threats to validity.

\flow's analysis is cubic in the worst case. Although pathological examples are
not entirely uncommon, we have so far been able to mitigate them with
low-hanging optimizations.

Its analysis is context-insensitive, and also not well-suited for libraries
with reflection. Many libraries provide annotations without checked
implementations, so we can typecheck the vast majority of code that uses these
libraries. Better techniques for checking libraries (e.g., TAJS) can complement
\flow.

Like many other type systems for dynamically typed languages, \flow has the
\kw{any} type, with which type checking can be completely bypassed. Unlike
gradual type systems, though, there is no runtime enforcement of types when they
interact with \kw{any}. For sound gradual typing, the subtyping rules can be
augmented to mark all type constructors as either trusted or untrusted.

Even without \kw{any}, some aspects of \js force us into choosing unsoundness
where it is objectively justified. We can lay down the conditions for soundness,
but not enforce them. For example, arrays in \js can have ``holes'': it is
possible to add an element out of bounds, in which case any intermediate
positions are filled with \vundef. Likewise, records in \js can also be accessed
as dictionaries, so it is possible to read and write a named property by passing
a computed string. Short of complicated numeric and string analysis, soundness
would demand that we lose type information on array dereferences and dictionary
reads, but this is too restrictive in practice. Instead we hope that developers
who care about soundness will not create arrays with holes (e.g., by always
using \kw{Array.push} to add elements), or will check for \vundef on
dereferences when needed; and the properties that are named and those that are
accessed via computed strings are disjoint.

\begin{acks}                            
Thanks to Basil Hosmer, Jeff Morrison, Nat Mote, Satish Chandra, Caleb Meredith, and James Kyle for their
  contributions to Flow's design and implementation, to Julien Verlaguet, Dwayne Reeves, and Yoann Padioleau
  for their work on infrastructure that Flow is built on, and to anonymous
  reviewers for their valuable feedback on previous drafts of this paper.
\end{acks}

\bibliography{main}


\clearpage

\appendix
\appendixpage

\linespread{1.3}

\setlist[itemize]{listparindent=0pt, topsep=3pt, itemsep=3pt}
\setlist[enumerate]{topsep=3pt, itemsep=3pt}

\counterwithin{equation}{theorem}
\counterwithin{enumii}{enumi}
\counterwithin{enumiii}{enumii}

These appendices contain material that was omitted from the main
paper. All definitions that are included in the main paper also hold here.
In addition we introduce ground types (the model
upon which our types are based), substitutions and subtyping, as well
as a brief note on polarities that were mentioned in the main paper.
We also include typing for our language's runtime, and proofs of
soundness of the inference of our system with respect to the declarative system
introduced in the main paper and the type safety of the declarative system.

\section{Types}
\label{sec:app:flow:types}

In this section we include a discussion on ground types
which is the model that the types described in
\secref{sec:core} are based on.
We then provide some
more context on the notion of polarity that was alluded to
during the discussion about constraint propagation
(\secref{sec:cons:prop}).
Finally, we define notions related to type subsumption as
they are going to be useful for the statement of lemmas and
theorems moving forward.

\subsection{Ground Types}\label{subsec:app:flow:ground}

At the basis of the type language described in
\secref{sec:types} is the notion of ground types.
The formulation of our ground type language follows the one
presented by~\citet{Pottier1998}. Here we will focus on the
changes we made to adapt that formulation to our
system's needs.
Ground types in our system are \emph{regular trees}. The
formal definition is similar to
\citet[Definition~1.1]{Pottier1998} but our ground signature
$\Sigma_g$ contains the terminals \tbase and $\rightarrow$
for types and the terminals \effempty and the set of
program variables \evarset for effects. Also $\rightarrow$ has arity 3
to also account for the function's effect, whose position is co-variant.


\paragraph{Ground Substitutions}
We connect the notion of types that were introduced in
\secref{sec:types} with ground types using
the notion of ground substitutions \gsubst.

\begin{definition}[Ground Substitution]
  A \emph{ground substitution} \gsubst (we will also refer
  to it as \emph{solution}) is a total mapping from type variables to ground types.
\end{definition}

A ground substitution can be applied to types by
recursively applying the substitution to the parts of the
type, replacing type variables with their ground type
mapping.

As regular trees, ground types can be infinite structures,
whereas the types we introduced in the main part are finite,
but crucially include type variables. This means that a
finite, yet recursively defined, type may correspond
(through a substitution) to an infinite ground type.

\subsection{Ground Subtyping}
Because of their infinite nature defining a subtyping
relation on ground types requires some special treatment.
Here, we define an ordering on ground types by quantifying over
paths in the regular trees that represent types. The
symbol $\leq_k$ denotes subtyping up to level $k$. The
definition is similar to
\citet[Definition~1.5]{Pottier1998}. 
For the case of effects, $\leq_0$ is reflexive and $\bot$
is the minimum element.
The definition of the subtyping relation $\leq$ over ground
types follows \citet[Definition~1.4]{Pottier1998} and
as in \citet[Proposition~1.3]{Pottier1998},
$\issubtype{\type}{\type'}$ is equivalent to:
$$
\foralldot{k\geq0}{\type\leq_{k}\type'}
$$
Equipping our ground alphabet with $\bot$ and $\top$ for
types, and $\top$ for effects (we have omitted them from our
formulation to avoid clutter), and using the subtyping
relation, our ground types can form a
lattice. The proof follows
\citet[Proposition~1.3]{Pottier1998}.

\paragraph{Effects}
The subtyping relation is extended to ground effects as well (we also 
refer to them as concrete effects).
Concrete effects can be interpreted as sets of variables and so effect
subtyping corresponds to the subset relation.

\paragraph{Environments}

In the following we assume that applying a substitution \gsubst to an
environment $\env$ of the constraint generation system returns a pair
containing two environments:
\begin{itemize}
  \item a concrete \emph{flow-sensitive environment} \cenv 
    binding variables to ground types in a flow-sensitive manner
    (\ie types that correspond to the base of the entries in \env), and 
  \item a \emph{general environment} \genv binding variables to ground types
    corresponding to the general type of each entry in \env.
\end{itemize}
We write this as:
$$
\is{\appsubst{\gsubst}{\env}}{\dividerpair{\cenv}{\genv}}
$$
We also use indexes as subscripts to retrieve the first or second
part of the above pair:
\begin{align*}
\is{\appsubst{\gsubst}{\env}_1}{\cenv}
&&
\is{\appsubst{\gsubst}{\env}_2}{\genv}
\end{align*}

The subtyping relation is extended to \cenv and \genv in a point-wise manner.

Note that this is different that the notation we used in the
main paper where we kept type entries of the form
\entry{\type}{\type'}. We made this change with the hope
that it removes unnecessary clutter from our formalization.

\subsection{Constraint Satisfaction}

The following definitions relate ground substitutions with constraint sets.

\begin{definition}[Constraint Satisfaction]
\label{def:flow:cons:sat}
We say that a ground substitution \gsubst satisfies a constraint \cons, 
and we write \satisfies{\gsubst}{\cons},
if the corresponding subtyping relation(s)
in the right hand side of the definitions below hold(s):
\[
\centering
\begin{tabular}{>{$}l<{$} >{$}c<{$} >{$}l<{$}}
  \satisfies{\gsubst}{\flows{\type}{\tvar}}                      & \defeq & \issubtype{\appsubst{\gsubst}{\type}}{\appsubst{\gsubst}{\tvar}} \\
  \satisfies{\gsubst}{\flows{\type}{\calluse{\type_1}{\effect}{\type_2}}}
                                                                 & \defeq & \issubtype{\appsubst{\gsubst}{\type}}{\appsubst{\gsubst}{\tarrow{\type_1}{\effect}{\type_2}}} \\
  \satisfies{\gsubst}{\flows{\type}{\preduse{\pred}{\type'}}}    & \defeq & \issubtype{\appsubst{\gsubst}{\refine{\type}{\pred}}}{\appsubst{\gsubst}{\type'}} \\
  \satisfies{\gsubst}{\flows{\type}{\getuse{\fieldsym}{\type'}}} & \defeq & \issubtype{\appsubst{\gsubst}{\type}}{\appsubst{\gsubst}{\tobjsingleton{\fieldsym}{\type'}}} \\
  \satisfies{\gsubst}{\flows{\type}{\setuse{\fieldsym}{\type'}}} & \defeq & \issubtype{\appsubst{\gsubst}{\type}}{\appsubst{\gsubst}{\tobjsingleton{\fieldsym}{\type'}}} \\
  \satisfies{\gsubst}{\flows{\effect}{\effvar}}                  & \defeq & \issubtype{\appsubst{\gsubst}{\effect}}{\appsubst{\gsubst}{\effvar}} \\
  \satisfies{\gsubst}{\flows{\effect}{\havocuse{\env}}}          & \defeq & \forallindot{\evar}{\appsubst{\gsubst}{\effect}}{\issubtype{\idx{\genv}{\evar}}{\idx{\cenv}{\evar}}}  \\
                                                                 &        & \text{where } \is{\appsubst{\gsubst}{\env}}{\dividerpair{\cenv}{\genv}}
\end{tabular}
\]
\end{definition}

\begin{definition}[Constraint Set Satisfaction under Substitution]
We say that a ground substitution \gsubst satisfies a constraint set \cset, 
and we write \satisfies{\gsubst}{\cset}, if for all \cons of \cset it holds that 
\satisfies{\gsubst}{\cons}.
\end{definition}

The following proposition connects constraint set consistency that was
discussed in \secref{sec:consistency} with constraint satisfiability
under ground substitution defined above.

\begin{proposition}[Constraint Set Satisfaction]
A (saturated) constraint set \cset is \emph{satisfiable},
iff there exists ground substitution \gsubst \st \satisfies{\gsubst}{\cset}.
\end{proposition}

\subsection{Polarities} \label{sec:app:polarities}

%
%

In \secref{sec:cons:prop}, we introduced
Rule~\hyperref[rule:cppredtrans]{\cppredtrans} that
contained the notion of a ``positive type hole''. To define
this formally we first introduce \emph{polar} types, which
can be positive or negative.
A  \emph{positive} type \ptype is a type used to to describe
outputs, whereas a \emph{negative} type \ntype describes
inputs. Similar definitions hold for effects (\peffect and
\neffect). 
Formally:
\begin{align*}
  \ptype & \production \tbase \sep \tarrow{\ntype_1}{\peffect}{\ptype_2} \sep
                          \tobj{\fieldsym_1}{\ptype_1}{\fieldsym_n}{\ptype_n} \sep 
                          \tvar \sep
                        \tjoin{\ptype_1}{\ptype_2}
  \\
  \ntype & \production \tbase \sep \tarrow{\ptype_1}{\neffect}{\ntype_2} \sep
                          \tobj{\fieldsym_1}{\ntype_1}{\fieldsym_n}{\ntype_n} \sep 
                        \tvar 
  \\
  \peffect & \production \effempty \sep \evar
  \sep \effvar \sep \ejoin{\peffect_1}{\peffect_2}
  \\
  \neffect & \production \effempty \sep \evar
  \sep \effvar 
\end{align*}


With this in mind we now define a \emph{type context} \tctx as a
type that contains a hole \hole in one of its leafs. Type
contexts also come in two flavors:
\begin{align*}
  \ptctx & \production 
    \tbase \sep \tarrow{\ntctx_1}{\peffect}{\ptctx_2} \sep
    \tobj{\fieldsym_1}{\ptctx_1}{\fieldsym_n}{\ptctx_n} \sep 
    \tvar \sep \tjoin{\ptctx_1}{\ptctx_2} \sep \tctxidx{}{}
  \\
  \ntctx & \production 
    \tbase \sep \tarrow{\ptctx_1}{\neffect}{\ntctx_2} \sep
    \tobj{\fieldsym_1}{\ntctx_1}{\fieldsym_n}{\ntctx_n} \sep 
    \tvar
\end{align*}

The critical part in the above definition is that negative contexts
\ntctx do not contain joins at their top-levels.

\section{Declarative Type System}
\label{sec:app:declarative}

\begin{figure}[t!]
\judgementHead{Expression Typing}
{\tcexpr{\cenv}{\genv}{\expr}{\type}{\effect}{\predmap}{\cenv'}}
\begin{mathpar}
\inferruleright{\tcconst\label{rule:tcconst}}{
}{
  \tcexpr{\cenv}{\genv}{\const}{\tbase_{\const}}{\effempty}{\predmapempty}{\cenv}
}
\and
\inferruleright{\tcvar\label{rule:tcvar}}{
  \idxis{\cenv}{\evar}{\type}
}{
  \tcexpr{\cenv}{\genv}{\evar}{\type}{\effempty}{\singletonpredmap{\evar}{\truthy}}{\cenv}
}
\and
\inferruleright{\tcassign\label{rule:tcassign}}{
  \tcexpr{\cenv}{\genv}{\expr}{\type}{\effect}{\predmap}{\cenv'}
}{
  \tcexpr{\cenv}{\genv}{\assign{\evar}{\expr}}{\type}
    {\effectconcat{\effect}{\evar}}
    {\logforget{\predmap}{\evar}}
    {\envupd{\cenv'}{\evar}{\type}}
}
\and
\inferruleright{\tcfun\label{rule:tcfun}}{
  \tcbody
    {\envcons{\envext{\eraseto{\genv}{\cenv}}{\evar}{\type}}{\locals{\stmt}}}{\genv}
    {\body{\stmt}{\expr}}{\type'}{\effect}{\cenv'}
}{
  \tcexpr{\cenv}{\genv}{\arrow{\evar}{\stmt}{\expr}}
    {\tarrow{\type}{\logforget{\effect}{\evar}}{\type'}}
    {\effempty}
    {\predmapempty}
    {\cenv}
}
\and
\inferruleright{\tccall\label{rule:tccall}}{
  \tcexpr{\cenv}{\genv}{\expr_1}{\type_1}{\effect_1}{\predmap_1}{\cenv_1}
  \\
  \tcexpr{\cenv_1}{\genv}{\expr_2}{\type_2}{\effect_2}{\predmap_2}{\cenv_2}
  \\
  \issubtype{\type_1}{\tarrow{\type_2}{\effect}{\type}}
  \\
  \is{\cenv'}{\erasevarto{\cenv_2}{\genv}{\effect}}
}{
  \tcexpr{\cenv}{\genv}{\efuncall{\expr_1}{\expr_2}}{\type}
    {\effectconcatthree{\effect_1}{\effect_2}{\effect}}
    {\predmapempty}{\cenv'}
}
\end{mathpar}
\caption{Expression Typing in \flowcore (Variables and Functions)}
\label{fig:flow:expr:var:fun:typing}
\end{figure}

In Figures~\ref{fig:flow:expr:var:fun:typing},~\ref{fig:flow:expr:rec:typing} 
and~\ref{fig:flow:stmt:typing} we define a declarative type
system that assigns types to expressions and statements of
\flowcore. The typing judgments for expressions and
statements are:
\begin{align*}
\tcexpr{\cenv}{\genv}{\expr}{\type}{\effect}{\predmap}{\cenv'}
&& 
\tcstmt{\cenv}{\genv}{\stmt}{\effect}{\cenv'}
\end{align*}
Here types \type are identical in structure to the types
introduced for the inference system but are concrete, \ie
there contain no type variables.
As mentioned earlier, environments \cenv bind variables \evar to types \type
(instead of type entries containing both a precise and a
general type).
The most general type for each variable is included in
environment \genv -- a flow-insensitive structure that
gathers the most general type (globally) for each variable
across the entire program. Thanks to $\alpha$-renaming each
defined variable to a unique name, there is no ambiguity
among variable identifiers.

Effects \effect are also concrete in this declarative
system. This means that they can now be directly interpreted
as sets of variables (since no effect variables are
presents).

We use the shorthand $\eraseto{\genv}{\cenv}$ to denote the
erasure of an environment \cenv with the types of \genv. This
operation effectively creates a new environment binding all
variables in \cenv to their bound types in \genv.
We also introduce the variant
$\erasevarto{\cenv}{\genv}{\effect}$,
where \effect is a concrete effect, to denote the environment
$\cenv\brackets{\mapping{\evar}{\idx{\genv}{\evar}}\mid\inset{\evar}{\effect}}$.

Environment join ($\sqcup$) and environment refinement ($\coloncolon$)
have similar definitions as in their constraint generation counterparts of
\figref{fig:flow:env:aux}, and so are omitted here.

\begin{figure}[t!]
\judgementHead{Expression Typing}
{\tcexpr{\cenv}{\genv}{\expr}{\type}{\effect}{\predmap}{\cenv'}}
\begin{mathpar}
\inferruleright{\tcand\label{rule:tcand}}{
  \tcexpr{\cenv}{\genv}{\expr_1}{\type_1}{\effect_1}{\predmap_1}{\cenv_1}
  \\
  \tcexpr{\refine{\cenv_1}{\predmap_1}}{\genv}{\expr_2}{\type_2}{\effect_2}{\predmap_2}{\cenv_2}
  \\
  \is{\type}{\tjoin{\refine{\type_1}{\falsy}}{\type_2}}
  \\
  \is{\effect}{\effectconcat{\effect_1}{\effect_2}}
  \\
  \is{\predmap}{\logand{\parens{\logforget{\predmap_1}{\effect_2}}}{\predmap_2}}
  \\
  \is{\cenv'}{\envjoin{\parens{\refine{\cenv_1}{\logneg{\predmap_1}}}}{\cenv_2}}
}{
  \tcexpr{\cenv}{\genv}{\binand{\expr_1}{\expr_2}}{\type}{\effect}{\predmap}{\cenv'}
}
\and
\inferruleright{\tcor\label{rule:tcor}}{
  \tcexpr{\cenv}{\genv}{\expr_1}{\type_1}{\effect_1}{\predmap_1}{\cenv_1}
  \\
  \tcexpr{\refine{\cenv_1}{\logneg{\predmap_1}}}{\genv}{\expr_2}{\type_2}{\effect_2}{\predmap_2}{\cenv_2}
  \\
  \is{\type}{\tjoin{\refine{\type_1}{\truthy}}{\type_2}}
  \\
  \is{\effect}{\effectconcat{\effect_1}{\effect_2}}
  \\
  \is{\predmap}{\logor{\parens{\logforget{\predmap_1}{\effect_2}}}{\predmap_2}}
  \\
  \is{\cenv'}{\envjoin{\parens{\refine{\cenv_1}{\predmap_1}}}{\cenv_2}}
}{
\tcexpr{\cenv}{\genv}{\binor{\expr_1}{\expr_2}}{\type}{\effect}{\predmap}{\cenv'}
}
\and
\inferruleright{\tcnot\label{rule:tcnot}}{
  \tcexpr{\cenv}{\genv}{\expr}{\type}{\effect}{\predmap}{\cenv'}
}{
\tcexpr{\cenv}{\genv}{\unaryneg{\expr}}{\tboolean}{\effect}{\logneg{\predmap}}{\cenv'}
}
\and
\inferruleright{\tcpred\label{rule:tcpred}}{
}{
\tcexpr{\cenv}{\genv}{\predof{\evar}}{\tboolean}{\effempty}{\singletonpredmap{\evar}{\primpred}}{\cenv}
}
\and
\inferruleright{\tcrecord\label{rule:tcrecord}}{
  \isequiv{\cenv}{\cenv_0}
  \\
  \forallindot{i}{[1,n]}{
    \tcexpr{\cenv_{i-1}}{\genv}{\expr_{i}}{\type_{i}}{\effect_{i}}{\predmap_{i}}{\cenv_{i}}
  }
  \\
  \forallindot{i}{[1,n]}{
    \issubtype{\type_{i}}{\type_{i}'}
  }
}{
\tcexpr{\cenv}{\genv}{\objlit{\fieldsym_1}{\expr_1}{\fieldsym_{n}}{\expr_{n}}}
  {\tobj{\fieldsym_1}{\type_1'}{\fieldsym_{n}}{\type_{n}'}}
  {\medsqcupl{i=1}{n}\effect_{i}}{\predmapempty}{\cenv_{n}}
}
\and
\inferruleright{\tcfldrd\label{rule:tcfldrd}}{
  \tcexpr{\cenv}{\genv}{\expr}{\type}{\effect}{\predmap}{\cenv'}
  \\
  \issubtype{\type}{\tobjsingleton{\fieldsym}{\type'}}
}{
  \tcexpr{\cenv}{\genv}{\fieldread{\expr}{\fieldsym}}{\type'}{\effect}{\predmapempty}{\cenv'}
}
\and
\inferruleright{\tcfldwr\label{rule:tcfldwr}}{
  \tcexpr{\cenv}{\genv}{\expr_1}{\type_1}{\effect_1}{\predmap_1}{\cenv_1}
  \\
  \issubtype{\type_1}{\tobjsingleton{\fieldsym}{\type_{\fieldsym}}}
  \\
  \tcexpr{\cenv_1}{\genv}{\expr_2}{\type_2}{\effect_2}{\predmap_2}{\cenv_2} 
  \\
  \issubtype{\type_2}{\type_{\fieldsym}}
}{
  \tcexpr{\cenv}{\genv}{\fieldwrite{\expr_1}{\fieldsym}{\expr_2}}{\type_2}
    {\effectconcat{\effect_1}{\effect_2}}
    {\predmap_2}
    {\cenv_2}
}
\end{mathpar}
\caption{Expression Typing in \flowcore (Logical Operators and Records)}
\label{fig:flow:expr:rec:typing}
\end{figure}

\begin{figure}[t!]
\judgementHead{Statement Typing}
{\tcstmt{\cenv}{\genv}{\stmt}{\effect}{\cenv'}}
\begin{mathpar}
\inferruleright{\tcexp\label{rule:tcexp}}{
  \tcexpr{\cenv}{\genv}{\expr}{\type}{\effect}{\predmap}{\cenv'}
}{
  \tcstmt{\cenv}{\genv}{\expr}{\effect}{\cenv'}
}
\and
\inferruleright{\tcvardecl\label{rule:tcvardecl}}{
  \tcexpr{\cenv}{\genv}{\expr}{\type}{\effect}{\predmap}{\cenv'}
}{
  \tcstmt{\cenv}{\genv}{\varassign{\evar}{\expr}}{\effectconcat{\effect}{\evar}}{\envupd{\cenv'}{\evar}{\type}}
}
\and
\inferruleright{\tcif\label{rule:tcif}}{
  \tcexpr{\cenv}{\genv}{\expr}{\type}{\effect}{\predmap}{\cenv'}
  \\
  \tcstmt{\refine{\cenv'}{\predmap}}{\genv}{\stmt_1}{\effect_1}{\cenv_1'}
  \\
  \tcstmt{\refine{\cenv'}{\negpredmap{\predmap}}}{\genv}{\stmt_2}{\effect_2}{\cenv_2'}
}{
  \tcstmt{\cenv}{\genv}{\ite{\expr}{\stmt_1}{\stmt_2}}{\effectconcatthree{\effect}{\effect_1}{\effect_2}}
  {\envjoin{\cenv_1'}{\cenv_2'}}
}
\and
\inferruleright{\tcseq\label{rule:tcseq}}{
  \tcstmt{\cenv}{\genv}{\stmt_1}{\effect_1}{\cenv_1}
  \\
  \tcstmt{\cenv_1}{\genv}{\stmt_2}{\effect_2}{\cenv_2}
}{
  \tcstmt{\cenv}{\genv}{\seq{\stmt_1}{\stmt_2}}{\effectconcat{\effect_1}{\effect_2}}{\cenv_2}
}
\end{mathpar}
\caption{Statement Typing in \flowcore}
\label{fig:flow:stmt:typing}
\end{figure}


\section{Runtime Typing}

Stating a progress and preservation theorem requires us to
extend the notion of well-typed expressions and statements to
runtime configurations.

\subsection{Term Typing}

\paragraph{Expressions \& Statements}
First we extend typing to runtime expressions. The 
judgment form is similar to the one for static expressions
with the difference that we have to include locations \loc
in the set of typeable expressions. To do that we equip our
judgment with an additional argument, the \emph{heap typing}
\heapty, defined as:
$$
\heapty \production \heaptyempty \sep \envext{\heapty}{\loc}{\type}
$$
The expression typing judgment becomes:
$$
\tcrtexpr{\cenv}{\genv}{\heapty}{\expr}{\type}{\effect}{\predmap}{\cenv'}
$$
Extending the rules for expression typing in
Figures~\ref{fig:flow:expr:var:fun:typing} and~\ref{fig:flow:expr:rec:typing}
to runtime expressions is
straightforward. An important addition is the rule for
location \loc typing:
\begin{mathpar}
\inferruleright{\tcrtloc\label{rule:rt:tcloc}}{
  \is{\idx{\heapty}{\loc}}{\type}
}{
  \tcrtexpr{\cenv}{\genv}{\heapty}{\loc}{\type}{\effempty}{\predmapempty}{\cenv}
}
\end{mathpar}

Similarly the form of typing runtime statements is extended
to:
$$
\tcrtstmt{\cenv}{\genv}{\heapty}{\stmt}{\effect}{\cenv'}
$$

\paragraph{Evaluation Contexts}
A more interesting situation arises when we try to extend
the judgment to evaluation contexts \ectx. The main issue
here is that the object under judgment contains a ``hole''
where another expression is expected to appear.
To address this we include a ``hole'' in the type structure of the return
type to host the type of the term that is expected to fill
in the hole of the evaluation context. The linked effect
and predicate are handled in a similar fashion:
$$
\tcrtexpr{\cenv}{\genv}{\heapty}{\ectx}
    {\ectxidx{\type'}{\type}}
    {\ectxidx{\effect'}{\effect}}
    {\ectxidx{\predmap'}{\predmap}}{\cenv'}
$$
\figref{fig:flow:ectx:typing} contains a selection of rules
for this judgment.


\begin{figure}[t!]
\judgementHead{Evaluation Context Typing Rules (selected)}
  {\tcrtexpr{\cenv}{\genv}{\heapty}{\ectx}
    {\ectxidx{\type'}{\type}}
    {\ectxidx{\effect'}{\effect}}
    {\ectxidx{\predmap'}{\predmap}}{\cenv'}}
\begin{mathpar}
\inferruleright{\ectxhole\label{rule:ectx:hole}}{
}{
\tcrtexpr{\cenv}{\genv}{\heapty}{\ectxemp}{\ectxemp}{\ectxemp}{\ectxemp}{\cenv}
}
\and
\inferruleright{\ectxcall\label{rule:ectx:call}}{
  \tcrtexpr{\cenv}{\genv}{\heapty}{\ectx}
    {\ectxidx{\type_1'}{\type_1}}
    {\ectxidx{\effect_1'}{\effect_1}}
    {\ectxidx{\predmap_1'}{\predmap_1}}
    {\cenv_1}
  \\
  \tcexpr{\cenv_1}{\genv}{\expr}{\type_2}{\effect_2}{\predmap_2}{\cenv_2}
  \\
  \issubtype{\type_1'}{\tarrow{\type_2}{\effect}{\type}}
  \\
  \is{\cenv'}{\erasevarto{\cenv_2}{\genv}{\effect}}
}{
  \tcrtexpr{\cenv}{\genv}{\heapty}{\efuncall{\ectx}{\expr}}
    {\ectxidx{\type}{\type_1}}
    {\effectconcatthree{\ectxidx{\effect_1'}{\effect_1}}{\effect_2}{\effect}}
    {\ectxidx{\predmapempty}{\predmap_1}}{\cenv'}
}
\end{mathpar}
\caption{Evaluation Context Typing in \flowcore}
\label{fig:flow:ectx:typing}
\end{figure}



When inverting typing relations, we often need to decompose
the typing of filled evaluation contexts \ectxidx{\ectx}{\expr}.
The following lemma deconstructs
the typing of such an expression to the typing of a bare evaluation
context \ectx and a typing of the filling expression \expr.

\begin{lemma}[Decomposing Evaluation Context Typing]
\label{lemma:decomp:ectx}
If
\shiftedalign{
  \tcrtexpr{\cenv}{\genv}{\heapty}{\ectxidx{\ectx}{\expr}}{\type}{\effect}{\predmap}{\cenv''}
}
then there exist $\type'$, $\effect'$, $\predmap'$ and $\cenv'$ \st
\begin{enumerate}[label=(\alph*)]
  \item
    \tcrtexpr{\cenv}{\genv}{\heapty}{\expr}{\type'}{\effect'}{\predmap'}{\cenv'}
  \item \imp{\is{\predmap'}{\predmapempty}}{
    \tcrtexpr{\cenv'}{\genv}{\heapty}{\ectx}
      {\ectxidx{\type}{\type'}}
      {\ectxidx{\effect}{\effect'}}
      {\ectxidx{\predmap}{\predmap'}}
      {\cenv''}
    }
\end{enumerate}
\end{lemma}
\begin{proof}
\emph{
By examining all possible cases of typing 
evaluation contexts \ectx, we will always type the expression \expr in the 
``hole'' first and then the evaluation context \ectx.}
\end{proof}

\subsection{Configuration Typing}

A runtime configuration in \flowcore contains the runtime
state, that itself comprises a heap \heap, a stack \stack
and a store \store, and a program term. Typing configurations
amounts to typing their subparts.
Before we move on to that we define two auxiliary
functions.

\paragraph{Auxiliary Functions}
The first one is the \emph{environment composition}
\compose{M}{N}. This operation works in the
usual way. The range of environment $N$ needs to be
compatible with the domain of $M$, otherwise the result
is undefined:
\begin{align*}
  \idx{\parens{\compose{M}{N}}}{x} & \equals\begin{cases}
    \idx{M}{\idx{N}{x}} & \text{ if \inset{X}{\envdom{N}} and \inset{\idx{N}{x}}{\envdom{M}}}   \\
    \mathit{undefined} & \text{otherwise}
  \end{cases}
\end{align*}

The second operator is the \emph{environment override}
\overrides{M}{N}. This operator produces an
environment whose domain is the union of the domains of the
two arguments. For each one of its arguments the override 
first attempts to return a binding by looking it up in environment $M$;
if this fails it tries $N$; 
and finally returns undefined if it fails there as well.
\begin{align*}
  \idx{\parens{\overrides{M}{N}}}{x} & \equals \begin{cases}
    \idx{M}{x} & \text{ if \inset{x}{\envdom{M}}}   \\
    \idx{N}{x} & \text{ if \inset{x}{\difference{\envdom{N}}{\envdom{M}}}}  \\
    \mathit{undefined} & \text{otherwise}
  \end{cases}
\end{align*}

\paragraph{Stack}
\begin{figure}[t!]
\judgementHead{Stack Typing Rules}{\tcstack{\genv}{\heapty}{\stack}{\ectxidx{\type'}{\type}}}
\begin{mathpar}
\inferruleright{\rtstackemp\label{rule:rtstackemp}}{
}{
  \tcstack{\genv}{\heapty}{\stackemp}{\ectxidx{\type}{\type}}
}
\and
\inferruleright{\rtstackcons\label{rule:rtstackcons}}{
  \jstore{\heapty}{\store}{\cenv}
  \\
  \tcectx{\cenv}{\genv}{\ectx}{\ectxidx{\type'}{\type}}{\cenv'}
  \\
  \isheaptyoverrides{\heapty'}{\cenv'}{\store}{\heapty}
  \\
  \tcstack{\genv}{\heapty'}{\stack}{\ectxidx{\type''}{\type'}}
}{
  \tcstack{\genv}{\heapty}{\stackcons{\stack}{\store}{\ectx}}{\ectxidx{\type''}{\type}}
}
\end{mathpar}
\caption{Runtime Stack Typing in \flowcore}
\label{fig:app:flow:stack}
\end{figure}
The form of the stack is reminiscent of the evaluation
context, so the judgment we use here has the following form:
$$
\tcstack{\genv}{\heapty}{\stack}{\ectxidx{\type}{\type'}}
$$
\figref{fig:app:flow:stack} contains the rules for this judgment.
The interesting rule here is
Rule~\hyperref[rule:rtstackcons]{\rtstackcons}, that types a stack
\stackcons{\stack}{\store}{\ectx}.
Following the flow of
execution the rule first checks the frame \ectx that is on
the top of the stack and then proceeds with the remaining
stack \stack.
What is interesting here is the construction
of the environment used for checking \stack. Assume $\cenv'$
the output environment after checking \ectx. This
environment contains the most recent updates of all the
variables that were assigned to in \ectx.
Our goal here is to construct an accurate heap typing
$\heapty'$ that corresponds to the state of the heap at the
end of \ectx. This heap typing will subsequently be used to
check \stack.
To do that, for every variable \evar such that
\inset{\binding{\evar}{\loc}}{\store}, \ie in scope at the
beginning of \ectx, we require its type to be looked up in
$\cenv'$. This amounts to
\compose{\cenv'}{\envinvert{\store}}. The rest will just be
looked up in the incoming \heapty.

\paragraph{Heap}
\begin{figure}[t!]
\judgementHead{Heap Typing Rules}{\tcheap{\genv}{\heapty}{\heap}}
\begin{mathpar}
\inferruleright{\rtheapemp\label{rule:rtheapemp}}{
}{
  \tcheap{\genv}{\heapty}{\heapemp}
}
\and
\inferruleright{\rtheaploc\label{rule:rtheaploc}}{
  \loc' \in \envdom{\heap}
  \\
  \tcheap{\genv}{\heapty}{\heap}
  \\
  \is{\idx{\heapty}{\loc}}{\idx{\heapty}{\loc'}}
}{
  \tcheap{\genv}{\heapty}{\heapext{\heap}{\loc}{\loc'}}
}
\and
\inferruleright{\rtheapconst\label{rule:rtheapconst}}{
  \tcheap{\genv}{\heapty}{\heap}
  \\
  \idxis{\heapty}{\loc}{\tbase_{\const}}
}{
  \tcheap{\genv}{\heapty}{\heapext{\heap}{\loc}{\const}}
}
\and
\inferruleright{\rtheapfun\label{rule:rtheapfun}}{
  \tcheap{\genv}{\heapty}{\heap}
  \\
  \idxis{\heapty}{\loc}{\type}
  \\
  \jstore{\heapty}{\store}{\cenv}
  \\
  \tcexpr{\cenv}{\genv}{\arrow{\evar}{\stmt}{\expr}}{\type}{\effempty}{\predmapempty}{\cenv}
}{
  \tcheap{\genv}{\heapty}{\heapext{\heap}{\loc}{\storearrow{\store}{\evar}{\stmt}{\expr}}}
}
\and
\inferruleright{\rtheapobj\label{appendix:rule:rtheapobj}}{
  \tcheap{\genv}{\heapty}{\heap}
  \\
  \idxis{\heapty}{\loc}{\tobj{\fieldsym_1}{\type_1}{\fieldsym_n}{\type_n}}
  \\
  \forallindot{i}{[1,n]}
  {\tcrtexpr{\cenv}{\genv}{\heapty}{\val_i}{\type_i'}{\effempty}{\predmapempty}{\cenv}}
  \\
  \forallindot{i}{[1,n]}{\issubtype{\type_i'}{\type_i}}
}{
  \tcheap{\genv}{\heapty}{\heapext{\heap}{\loc}{\objlit{\fieldsym_1}{\val_1}{\fieldsym_n}{\val_n}}}
}
\end{mathpar}
\caption{Heap Typing in \flowcore}
\label{fig:app:flow:heap}
\end{figure}
\figref{fig:app:flow:heap} shows the rules for checking a
heap \heap against a heap typing \heapty. The most
interesting case here is that of record typing by
Rule~\hyperref[appendix:rule:rtheapobj]{\rtheapobj}. This
rule infers a type for each value $\val_i$ stored at some
field of the record and then unifies this type with the type
of each field specified in the store typing \heapty.

\paragraph{Configuration}
Finally, \figref{fig:app:flow:conf} shows the typing rules
for runtime configurations where the terms are either
expressions, function bodies or statements. These largely
follow the same principles as the typing for stacks that we
saw earlier.

\begin{figure}
\judgementHeadTwo{Runtime Configuration Typing}
  {\tcrtconfexpr{\genv}{\heapty}{\rtstate}{\exprbody}{\type}}
  {\tcrtconfstmt{\genv}{\heapty}{\rtstate}{\stmt}}
\begin{mathpar}
%
%
\inferruleright{\rtconfbody\label{rule:rtconfbody}}{
  \tcheap{\genv}{\heapty}{\heap}
  \\
  \jstore{\heapty}{\store}{\cenv}
  \\
  \tcrtbody{\cenv}{\genv}{\heapty}{\exprbody}{\type}{\effect}{\cenv'}
  \\
  \isheaptyoverrides{\heapty'}{\cenv'}{\store}{\heapty}
  \\
  \tcstack{\genv}{\heapty'}{\stack}{\ectxidx{\type'}{\type}}
}{
\tcrtconfexpr{\genv}{\heapty}{\rttriplet{\heap}{\stack}{\store}}{\exprbody}{\type'}
}
\and
\inferruleright{\rtconfstmt\label{rule:rtconfstmt}}{
  \tcheap{\genv}{\heapty}{\heap}
  \\
  \jstore{\heapty}{\store}{\cenv}
  \\
  \tcrtstmt{\cenv}{\genv}{\heapty}{\stmt}{\effect}{\cenv'}
  \\
  \isheaptyoverrides{\heapty'}{\cenv'}{\store}{\heapty}
  \\
  \tcstack{\genv}{\heapty'}{\stack}{\ectxidx{\type'}{\type}}
}{
\tcrtconfstmt{\genv}{\heapty}{\rttriplet{\heap}{\stack}{\store}}{\stmt}
}
\end{mathpar}
\caption{Runtime Configuration Typing in \flowcore}
\label{fig:app:flow:conf}
\end{figure}


\section{Proofs}

This section contains a statement and proof of soundness of the inference
type system of \secref{sec:cons:sys} with respect to the
declarative system of \secref{sec:declarative}, followed
by our type safety result for the declarative system and by
extension the entire type system.

\subsection{Type Inference Soundness}

The following lemma captures the intuition behind the ``havoc'' mechanism,
as the erasure of the part of the widened environment that is affected
by the reaching effect.

\begin{lemma}[Havoc]\label{lemma:flow:type:inf:havoc}
If
\begin{enumerate}[label=(\roman*)]
\item \widen{\env}{\env'}{\cset}
\item \issupseteq{\cset'}{\consconcat{\cset}{\uniflow{\effvar}{\havocuse{\env'}}}}
\item \satisfies{\gsubst}{\cset'}
  \label{hvc:iii}
\end{enumerate}
then
\shiftedalign{
  \is{\cenv'}{\erasevarto{\cenv}{\genv}{\appsubst{\gsubst}{\effvar}}}
}
where \is{\appsubst{\gsubst}{\env}}{\dividerpair{\cenv}{\genv}}
and \is{\appsubst{\gsubst}{\env'}}{\dividerpair{\cenv'}{\genv}}.
\end{lemma}
\begin{proof}
Let \is{\appsubst{\gsubst}{\effvar}}{\effect}.
For every variable \inset{\evar}{\effect},
it also holds that \inset{\flows{\evar}{\havocuse{\env'}}}{\cset'},
since $\cset'$ is saturated.
Let \is{\idx{\env'}{\evar}}{\entry{\type}{\tvar}}.
By Rule~\cphavoc on the binding for \evar,
it holds that \inset{\flows{\tvar}{\type}}{\cset'}.
Due to~\ref{hvc:iii}, 
\issubtype{\appsubst{\gsubst}{\tvar}}{\appsubst{\gsubst}{\type}}.
Which is also written as
\issubtype{\idx{\genv}{\evar}}{\idx{\cenv}{\evar}}.
But by definition of \genv it holds that 
\issubtype{\idx{\cenv}{\evar}}{\idx{\genv}{\evar}},
so it must be that \is{\idx{\cenv}{\evar}}{\idx{\genv}{\evar}}.
Generalizing for all variables in \appsubst{\gsubst}{\effvar} we prove the wanted.
\end{proof}

\begin{lemma}[Type Inference Soundness]\label{lemma:flow:type:inf:soundness}
If 
\begin{enumerate}[label=(\roman*)]
  \item \jexpr{\env}{\expr}{\type}{\effect}{\predmap}{\cset}{\env'} \label{tis:i}
  \item \satisfies{\gsubst}{\cset}  \label{tis:ii}
\end{enumerate}
then
\shiftedalign{
  \tcexprsingle{\appsubst{\gsubst}{\env}}
    {\expr}{\appsubst{\gsubst}{\type}}
    {\appsubst{\gsubst}{\effect}}{\predmap}
    {\appsubst{\gsubst}{\env'}_1}
}
\end{lemma}

\begin{proof}
By induction on the derivation of \ref{tis:i}:
\begin{itemize}
  \item \hyperref[rule:cgcall]{\cgcall}:
    \la{
    \jexpr{\env}{\efuncall{\expr_1}{\expr_2}}{\tvar}{\effect}{\predmapempty}{\cset}{\env_3}
    \label{tis:1}
    }

    By inverting Rule~\hyperref[rule:cgcall]{\cgcall} on~\eqref{tis:1}:
    \la{
      \jexpr{\env}{\expr_1}{\type_1}{\effect_1}{\predmap_1}{\cset_1}{\env_1}
      \label{tis:2}
      \\
      \jexpr{\env_1}{\expr_2}{\type_2}{\effect_2}{\predmap_2}{\cset_2}{\env_2}
      \label{tis:3}
      \\
      \widen{\env_2}{\env_3}{\cset_{\wsym}}
      \label{tis:4}
      \\
      \is{\effectconcatthree{\effect_1}{\effect_2}{\effvar}}{\effect}
      \label{tis:5}
      \\
      \is{\consconcatfour{\cset_1}{\cset_2}{\cset_{\wsym}}
        {\mkset{\concat{\flows{\effvar}{\havocuse{\env_3}}}
          {\flows{\type_1}{\calluse{\type_2}{\effvar}{\tvar}}}}
        }
      \label{tis:6}
      }{\cset}
    }
    where \fresh{\tvar, \effvar}.

    Since by \eqref{tis:6} it is 
    \issupseteq{\cset}{\cset_1} and \issupseteq{\cset}{\cset_2},
    using~\ref{tis:ii} it holds that:
    \la{
      \satisfies{\gsubst}{\cset_1}
      \label{tis:7}
      \\
      \satisfies{\gsubst}{\cset_2}
      \label{tis:8}
    }

    By induction hypothesis using \eqref{tis:2}, \eqref{tis:7}, \eqref{tis:3} and \eqref{tis:8}:
    \la{
      \tcexprsingle{\appsubst{\gsubst}{\env}}
        {\expr_1}{\appsubst{\gsubst}{\type_1}}
        {\appsubst{\gsubst}{\effect_1}}{\predmap_1}
        {\appsubst{\gsubst}{\env_1}_1}
      \label{tis:9}
      \\
      \tcexprsingle{\appsubst{\gsubst}{\env_1}}
        {\expr_2}{\appsubst{\gsubst}{\type_2}}
        {\appsubst{\gsubst}{\effect_2}}{\predmap_2}
        {\appsubst{\gsubst}{\env_2}_1}
      \label{tis:10}
    }

    By \eqref{tis:6}, using Definition~\ref{def:flow:cons:sat}:
    \la{
      \isequiv{\issubtype{\appsubst{\gsubst}{\type_1}}{\appsubst{\gsubst}{\tarrow{\type_2}{\effvar}{\tvar}}}}
      {\tarrow{\appsubst{\gsubst}{\type_2}}
        {\appsubst{\gsubst}{\effvar}}
        {\appsubst{\gsubst}{\tvar}}
      }
      \label{tis:11}
    }

    By Lemma~\ref{lemma:flow:type:inf:havoc} on \eqref{tis:4}, \eqref{tis:6} and \ref{tis:ii}:
    \la{
      \is{\cenv_3}{\erasevarto{\cenv_2}{\genv}{\appsubst{\gsubst}{\effvar}}}
      \label{tis:12}
    }
    where
    \is{\appsubst{\gsubst}{\env_2}}{\dividerpair{\cenv_2}{\genv}}
    and
    \is{\appsubst{\gsubst}{\env_3}}{\dividerpair{\cenv_3}{\genv}}.

    By Rule~\hyperref[rule:tccall]{\tccall} on \eqref{tis:9}, \eqref{tis:10}, \eqref{tis:11} and \eqref{tis:12}
    \begin{align}
      && \tcexprsingle
        {\appsubst{\gsubst}{\env}}
        {\efuncall{\expr_1}{\expr_2}}{\appsubst{\gsubst}{\tvar}}
        {\effectconcatthree{\appsubst{\gsubst}{\effect_1}}{\appsubst{\gsubst}{\effect_2}}{\appsubst{\gsubst}{\effvar}}}
        {\predmapempty}{\appsubst{\gsubst}{\env_3}_1}
      \label{tis:13}
      \\
      & \therefore & 
      \tcexprsingle
        {\appsubst{\gsubst}{\env}}
        {\efuncall{\expr_1}{\expr_2}}{\appsubst{\gsubst}{\tvar}}
        {\appsubst{\gsubst}{\effectconcatthree{\effect_1}{\effect_2}{\effvar}}}
        {\predmapempty}{\appsubst{\gsubst}{\env_3}_1}
      \end{align}

  \item \hyperref[rule:cgassign]{\cgassign}:
    \la{
      \jexpr{\env}{\assign{\evar}{\expr}}{\type}
        {\effectconcat{\effect}{\evar}}
        {\logforget{\predmap}{\evar}}
        {\cset}
        {\envupd{\env'}{\evar}{\entry{\type}{\tvar}}}
      \label{tis:14}
    }
    
    By inverting Rule~\hyperref[rule:cgassign]{\cgassign} on~\eqref{tis:14}:
    \la{
      \jexpr{\env}{\expr}{\type}{\effect}{\predmap}{\cset_0}{\env'}
      \label{tis:15}
      \\
      \idxis{\env'}{\evar}{\entry{\type_0}{\tvar}}
      \label{tis:16}
      \\
      \is{\cset}{\consconcat{\cset_0}{\uniflow{\type}{\tvar}}}
      \label{tis:17}
    }
    
    Since by \eqref{tis:17} it is 
    \issupseteq{\cset}{\cset_0},
    using~\ref{tis:ii} it holds that:
    \la{
      \satisfies{\gsubst}{\cset_0}
      \label{tis:18}
    }

    By induction hypothesis using \eqref{tis:15} and \eqref{tis:18}:
    \la{
      \tcexprsingle{\appsubst{\gsubst}{\env}}
        {\expr}{\appsubst{\gsubst}{\type}}
        {\appsubst{\gsubst}{\effect}}{\predmap}
        {\appsubst{\gsubst}{\env'}_1}
      \label{tis:19}
    }

    By applying Rule~\hyperref[rule:tcassign]{\tcassign} on~\eqref{tis:19}:
    \begin{align}
      & & \tcexprsingle
        {\appsubst{\gsubst}{\env}}
        {\assign{\evar}{\expr}}
        {\appsubst{\gsubst}{\type}}
        {\effectconcat{\appsubst{\gsubst}{\effect}}{\evar}}
        {\logforget{\predmap}{\evar}}
        {\envupd{\appsubst{\gsubst}{\env'}}{\evar}{\appsubst{\gsubst}{\type}}}
      \label{tis:20}
      \\
      & \therefore &
      \tcexprsingle
        {\appsubst{\gsubst}{\env}}
        {\assign{\evar}{\expr}}
        {\appsubst{\gsubst}{\type}}
        {\appsubst{\gsubst}{\effectconcat{\effect}{\evar}}}
        {\logforget{\predmap}{\evar}}
        {\appsubst{\gsubst}{\envupd{\env'}{\evar}{\entry{\type}{\tvar}}}_1}
      \label{tis:21} 
    \end{align}

\end{itemize}

\noindent\emph{The rest of the cases are handled similarly.}

\end{proof}

\subsection{Type Safety}

In this section we present the proofs of our safety
result that connects the declarative type system of \secref{sec:declarative}
with the runtime semantics of \secref{sec:semantics}.
First we set up a number of
auxiliary lemmas and then proceed with a Preservation Theorem
(\ref{theorem:flow:preservation})
and a Progress Theorem (\ref{theorem:flow:progress}) 
that are later combined to produce
a Type Safety Theorem (\ref{theorem:flow:safety}).



\begin{lemma}[Erased Environment Subtyping]
\label{lemma:erase:subtype}
If
\is{\erasevarto{\cenv}{\genv}{\effect}}{\cenv'},
then
\issubtype{\cenv}{\cenv'}.
\end{lemma}
\begin{proof}
\emph{By definition of the \erasesym operator.}
\end{proof}

In the remaining we use the metavariable \exprbody to denote
a term that is either an expression \expr or a function
body \body{\stmt}{\expr}.

\begin{lemma}[Heap Typing Weakening]
\label{lemma:heapty:weakening}
Let \issubtype{\heapty'}{\heapty}. Then:
\begin{enumerate}[label=\Roman*., ref=\ref{lemma:heapty:weakening}.\Roman*]
  \item If
    \tcrtexpr{\cenv}{\genv}{\heapty}{\exprbody}{\type}{\effect}{\predmap}{\cenv_1}, 
    then
    \begin{enumerate}[label=(\alph*)]
      \item \tcrtexpr{\cenv}{\genv}{\heapty'}{\exprbody}{\type'}{\effect'}{\predmap'}{\cenv_1'}
      \item \issubtype{\type'}{\type} and \issubtype{\effect'}{\effect}
    \end{enumerate}
    \label{lemma:heapty:weakening:exprbody}

  \item If 
    \issubtype{\type_1'}{\type_1}
    and
    \tcstack{\genv}{\heapty}{\stack}{\ectxidx{\type}{\type_1}},
    then
    \begin{enumerate}[label=(\alph*)]
      \item \tcstack{\genv}{\heapty'}{\stack}{\ectxidx{\type'}{\type_1'}} 
      \item \issubtype{\type'}{\type}
    \end{enumerate}
    \label{lemma:heapty:weakening:stack}

  \item If
    \tcheap{\genv}{\heapty}{\heap},
    then
    \tcheap{\genv}{\heapty'}{\heap}.
    \label{lemma:heapty:weakening:heap}

\end{enumerate}
\end{lemma}
\begin{proof}
  \emph{By induction on the given derivation.}
\end{proof}

\begin{lemma}[Environment Strengthening]
\label{lemma:environment:weakening}
For the following, let environments \cenv and $\cenv'$
be defined over common domains.
and \tcexpr{\cenv}{\genv}{\expr}{\type}{\effect}{\predmap}{\cenv_1}.
\begin{enumerate}[label=\Roman*., ref=\ref{lemma:environment:weakening}.\Roman*]
\item
If \issubtype{\cenv'}{\cenv}, then 
\begin{enumerate}[label=(\alph*)]
  \item \tcexpr{\cenv'}{\genv}{\expr}{\type'}{\effect'}{\predmap'}{\cenv_1'}
  \item \issubtype{\type'}{\type} and \issubtype{\effect'}{\effect}
  \item \issubtype{\cenv_1'}{\cenv_1},
        \issubtype{\refine{\cenv_1'}{\predmap'}}{\refine{\cenv_1}{\predmap}}
    and \issubtype{\refine{\cenv_1'}{\neg\predmap'}}{\refine{\cenv_1}{\neg\predmap}}
\end{enumerate}
\label{lemma:environment:weakening:expr}
\item 
  \tcexpr{\refine{\cenv}{\logforget{\predmap}{\effect}}}{\genv}{\expr}{\type'}{\effect'}{\predmap'}
    {\refine{\cenv_1}{\logforget{\predmap}{\effect}}}
  \label{lemma:environment:weakening:expr:refine}
\end{enumerate}
\end{lemma}
\begin{proof}
\emph{By induction on the given derivation.}
\end{proof}

\begin{lemma}[NonEffect]\label{lemma:non:effect}
If
\shiftedalign{
\tcexpr{\cenv}{\genv}{\expr}{\type}{\effect}{\predmap}{\cenv'}
}
then
\shiftedalign{
  \issubtype{\restrict{\cenv'}{\ovE{\effect}}}{\restrict{\cenv}{\ovE{\effect}}}
}
where $\ovE{\effect}$ is the set of program variables
that do not belong to the concrete effect \effect.
\end{lemma}

\begin{proof}

By induction on the given derivation:
\begin{itemize}
  
  \item  \hyperref[rule:tcvar]{\tcvar},
    \hyperref[rule:tcconst]{\tcconst},
    \hyperref[rule:tcfun]{\tcfun} and
    \hyperref[rule:tcpred]{\tcpred}:
    It holds that
    \la{
      \isequiv{\cenv'}{\cenv}
    }
    so the wanted result holds trivially.

  \item  \hyperref[rule:tcassign]{\tcassign}:
    \la{
      \tcexpr{\cenv}{\genv}{
        \underbrace{\assign{\evar}{\expr_0}}_{\expr}
      }
      {\type}
      {\ejoin{\effect_0}{\evar}}
      {\logforget{\predmap_0}{\evar}}
      {\underbrace{\envupd{\cenv_0}{\evar_0}{\type}}_{\cenv'}}
      \label{noneffect:20}
    }

    By inverting \hyperref[rule:tcassign]{\tcassign} on~\eqref{noneffect:20}:
    \la{
      \tcexpr{\cenv}{\genv}{\expr_0}{\type}{\effect_0}{\predmap_0}{\cenv_0}
      \label{noneffect:21}
    }

    By~\eqref{noneffect:20} 
    for a variable \evarb \st \notinset{\evarb}{\effect},
    it also holds that:
    \la{
      \notis{\evarb}{\evar}
      \label{noneffect:39}
      \\
      \notinset{\evarb}{\effect_0}
      \label{noneffect:40}
    }

    By induction hypothesis using~\eqref{noneffect:21} and~\eqref{noneffect:40}:
    \la{
      \issubtype{\idx{\cenv_0}{\evarb}}{\idx{\cenv}{\evarb}}
      \label{noneffect:22}
    }

    By~\eqref{noneffect:39} it holds that 
    \is{\idx{\cenv_0}{\evarb}}{\idx{\cenv'}{\evarb}},
    and so by~\eqref{noneffect:22}:
    \la{
      \issubtype{\idx{\cenv'}{\evarb}}{\idx{\cenv}{\evarb}}
      \label{noneffect:23}
    }

    



  \item \hyperref[rule:tccall]{\tccall}:
    \la{
      \tcexpr{\cenv}{\genv}{\efuncall{\expr_1}{\expr_2}}{\type}
        {\underbrace{\effectconcatthree{\effect_1}{\effect_2}{\effect_c}}_{\effect}}
        {\predmapempty}{\cenv'}
      \label{noneffect:1}
    }

    By inverting~\hyperref[rule:tccall]{\tccall} on~\eqref{noneffect:1}
    \la{
      \tcexpr{\cenv}{\genv}{\expr_1}{\type_1}{\effect_1}{\predmap_1}{\cenv_1}
      \label{noneffect:2}
      \\
      \tcexpr{\cenv_1}{\genv}{\expr_2}{\type_2}{\effect_2}{\predmap_2}{\cenv_2}
      \label{noneffect:3}
      \\
      \issubtype{\type_1}{\tarrow{\type_2}{\effect_c}{\type}}
      \\
      \is{\cenv'}{\erasevarto{\cenv_2}{\genv}{\effect_c}}
      \label{noneffect:4}
    }
    
    For a variable \inset{\evar}{\envdom{\cenv'}} \st
    \notinset{\evar}{\effect},
    it also holds that:
    \la{
      \notinset{\evar}{\effect_1}
      \label{noneffect:50}
      \\
      \notinset{\evar}{\effect_2}
      \label{noneffect:51}
      \\
      \notinset{\evar}{\effect_c}
      \label{noneffect:10}
    }
    
    By induction hypothesis on~\eqref{noneffect:2} and \eqref{noneffect:50},
    and~\eqref{noneffect:3} and~\eqref{noneffect:51}:
    \la{
      \issubtype{\idx{\cenv_1}{\evar}}{\idx{\cenv}{\evar}}
      \label{noneffect:5}
      \\
      \issubtype{\idx{\cenv_2}{\evar}}{\idx{\cenv_1}{\evar}}
      \label{noneffect:6}
    }

    By definition of the \erasesym operator
    on \eqref{noneffect:4}
    for \evar \st ~\eqref{noneffect:10}:
    \la{
      \is{\idx{\cenv'}{\evar}}{\idx{\cenv_2}{\evar}}
      \label{noneffect:14}
    }

    By \eqref{noneffect:5}, \eqref{noneffect:6} and~\eqref{noneffect:14}:
    \la{
      \issubtype{\idx{\cenv'}{\evar}}{\idx{\cenv}{\evar}}
    }

  \item  
      \hyperref[rule:tcand]{\tcand},
      \hyperref[rule:tcor]{\tcor}, 
      \hyperref[rule:tcnot]{\tcnot},
      \hyperref[rule:tcpred]{\tcpred},
      \hyperref[rule:tcrecord]{\tcrecord},
      \hyperref[rule:tcfldrd]{\tcfldrd} and
      \hyperref[rule:tcfldwr]{\tcfldwr}:
        \emph{Similar to above.}

\end{itemize}
\end{proof}


\begin{assumption}[Dead Code Checking]
\label{assumption:dead:code}
\mbox{}
\begin{enumerate}[label=\Roman*., ref=\ref{assumption:dead:code}.\Roman*]
  \item \tcstmt{\cenv}{\genv}{\ite{\vtrue}{\stmt_1}{\stmt_2}}{\effect}{\cenv'} iff
        \tcstmt{\cenv}{\genv}{\stmt_1}{\effect}{\cenv'}.
    \label{assumption:dead:code:ite:true}
  \item \tcstmt{\cenv}{\genv}{\ite{\vfalse}{\stmt_1}{\stmt_2}}{\effect}{\cenv'} iff
        \tcstmt{\cenv}{\genv}{\stmt_2}{\effect}{\cenv'}.
    \label{assumption:dead:code:ite:false}
  \item \tcexpr{\cenv}{\genv}{\binand{\vtrue}{\expr}}{\type}{\effect}{\predmap}{\cenv'} iff
        \tcexpr{\cenv}{\genv}{\expr}{\type}{\effect}{\predmap}{\cenv'}.
    \label{assumption:dead:code:and:true}
  \item \tcexpr{\cenv}{\genv}{\binand{\vfalse}{\expr}}{\type}{\effect}{\predmap}{\cenv'} iff
        \tcexpr{\cenv}{\genv}{\vfalse}{\type}{\effect}{\predmap}{\cenv'}.
    \label{assumption:dead:code:and:false}
  \item \tcexpr{\cenv}{\genv}{\binor{\vfalse}{\expr}}{\type}{\effect}{\predmap}{\cenv'} iff
        \tcexpr{\cenv}{\genv}{\expr}{\type}{\effect}{\predmap}{\cenv'}.
    \label{assumption:dead:code:or:false}
  \item \tcexpr{\cenv}{\genv}{\binor{\vtrue}{\expr}}{\type}{\effect}{\predmap}{\cenv'} iff
        \tcexpr{\cenv}{\genv}{\vtrue}{\type}{\effect}{\predmap}{\cenv'}.
    \label{assumption:dead:code:or:true}
\end{enumerate}
\end{assumption}

\begin{lemma}[Preservation of Typing by Expression Reduction]
\label{lemma:expr:preservation}
Typing is preserved over the reduction of an expression that
preserves the state of the stack. That is, for an initial
runtime state
\isdef{\rtstate}{\angletriplet{\heap}{\stack}{\store}},
a target state
\isdef{\rtstate'}{\angletriplet{\heap'}{\stack}{\store'}}
if, under a heap typing \heapty:
\begin{enumerate}[label=(\roman*)]
  \item \tcheap{\genv}{\heapty}{\heap}
    \label{er:i}
  \item \tcrtexpr{\cenv}{\genv}{\heapty}{\expr}{\type}{\effect}{\predmap}{\cenv_1} 
    \label{er:ii}
  \item \stepsconf{\rtstate}{\expr}{\rtstate'}{\expr'}
    \label{er:iii}
\end{enumerate}
where
\isdef{\cenv}{\compose{\heapty}{\store}},
then there exist $\heapty'$ \st: 
\begin{enumerate}[label=(\alph*)]
  \item \tcheap{\genv}{\heapty'}{\heap'}
    \label{er:a}
  \item \tcrtexpr{\cenv'}{\genv}{\heapty'}{\expr'}{\type'}{\effect'}{\predmap'}{\cenv_1'} 
    \label{er:b}
  \item \issubtype{\type'}{\type}
    \label{er:c}
  \item \issubtype{\effect'}{\effect}
    \label{er:d}
  \item
    \begin{itemize}
      \item \imp{\logand{\isval{\expr'}}{\istruthy{\expr'}}}{\issubtype{\refine{\cenv_1'}{\predmap'}}{\refine{\cenv_1}{\predmap}}}
      \item \imp{\logand{\isval{\expr'}}{\isfalsy{\expr'}}}{\issubtype{\refine{\cenv_1'}{\neg\predmap'}}{\refine{\cenv_1}{\neg\predmap}}}
      \item \imp{\neg\isval{\expr'}}{\logand{\issubtype{\refine{\cenv_1'}{\predmap'}}{\refine{\cenv_1}{\predmap}}}{\issubtype{\refine{\cenv_1'}{\neg\predmap'}}{\refine{\cenv_1}{\neg\predmap}}}}
    \end{itemize}
    \label{er:e}
\end{enumerate}
where \isdef{\cenv'}{\compose{\heapty'}{\store'}}.
\end{lemma}

\begin{proof}

By induction on the derivation of~\ref{er:ii}:

\begin{itemize}


  \item  \hyperref[rule:rtpredvar]{\rtpredvar}:
    \la{
      \stepsconf
        {\angletriplet{\heap}{\stack}{\store}}{\underbrace{\predof{\evar}}_{\expr}}
        {\angletriplet{\heap}{\stack}{\store}}{\val}
      \label{er:pv:1}
    }

    By Rule~\hyperref[rule:tcpred]{\tcpred}, \ref{er:ii} is of the form:
    \la{
      \tcrtexpr{\cenv}{\genv}{\heapty}{\predof{\evar}}{\tboolean}{\effempty}{\underbrace{\singletonpredmap{\evar}{\primpred}}_{\predmap}}{\cenv}
      \label{er:pv:4}
    }

    Let \isdef{\loc}{\idx{\store}{\evar}}.

    We examine the case where \istruthy{\val} and pick:
    \la{
      \isdef{\heapty'}{\upd{\heapty}{\loc}{\refine{\idx{\heapty}{\loc}}{\primpred}}}
      \label{er:pv:4a}
    }
    
    The case for \isfalsy{\val} is similar, replacing \primpred with $\neg\primpred$.

    By definition \eqref{er:pv:4a} it holds that:
    \la{
      \issubtype{\heapty'}{\heapty}
      \label{er:pv:4b}
    }

    Store and heap do not evolve, \ie \is{\store'}{\store} and \is{\heap'}{\heap}.

    So, by Lemma~\ref{lemma:heapty:weakening:heap} on~\ref{er:i} and~\eqref{er:pv:4b}:
    \la{
      \tcheap{\genv}{\heapty'}{\heap'}
      \label{er:pv:5}
    }
    which proves \ref{er:a}.

    By definition of $\cenv'$, it holds that:
    \la{
      \is{\cenv'}{\envupd{\cenv}{\evar}{\refine{\idx{\cenv}{\evar}}{\primpred}}}
      \label{er:pv:5a}
    }

    By applying Rule \hyperref[rule:tcconst]{\tcconst} on \val (\vtrue or \vfalse)
    \la{
      \tcrtexpr{\cenv'}{\genv}{\heapty'}{\val}{\tboolean}{\effempty}
        {\underbrace{\predmapempty}_{\predmap'}}
        {\underbrace{\cenv'}_{\cenv_1'}}
      \label{er:pv:32}
    }
    which proves~\ref{er:b},~\ref{er:c} and~\ref{er:d}.

    For~\ref{er:e} we have:
    \[
    \begin{array}{rcl}
      \refine{\cenv_1'}{\predmap'} & \equalwithlabel{\eqref{er:pv:32}} & \refine{\cenv'}{\predmapempty} \\
                                   & \equalwithlabel{\eqref{er:pv:5a}} & \upd{\cenv}{\evar}{\refine{\idx{\cenv}{\evar}}{\primpred}}  \\
                                   & \equiv & \refine{\cenv}{\parens{\singletonpredmap{\evar}{\primpred}}}  \\
                                   & \equalwithlabel{\eqref{er:pv:4}} & \refine{\cenv}{\predmap}  \\
                                   & \equalwithlabel{\eqref{er:pv:4}} & \refine{\cenv_1}{\predmap}
    \end{array}
    \]
    This proves the first case of \ref{er:e}.
    The rest are trivially true.


  \item \hyperref[rule:rtasgn]{\rtasgn} with \is{\val}{\const}:
    \la{
      \stepsconf
        {\angletriplet{\heap}{\stack}{\store}}
        {\underbrace{\assign{\evar}{\const}}_{\expr}}
        {\angletriplet{\heap'}{\stack}{\store}}
        {\const}
      \label{er:a:0}
    }

    By inverting Rule~\hyperref[rule:rtasgn]{\rtasgn} on~\eqref{er:a:0}:
    \la{
      \is{\heap'}{\upd{\heap}{\idx{\store}{\evar}}{\const}}
    }

    By Rule~\hyperref[rule:tcassign]{\tcassign},~\ref{er:ii} is of the form:
    \la{
      \tcrtexpr{\cenv}{\genv}{\heapty}{\assign{\evar}{\const}}{\tbase_{\const}}
        {\effectconcat{\effect}{\evar}}
        {\logforget{\predmap}{\evar}}
        {\underbrace{\upd{\cenv}{\evar}{\tbase_{\const}}}_{\cenv_1}}
      \label{er:a:3}
    }

    By inverting Rule~\hyperref[rule:tcassign]{\tcassign} on~\eqref{er:a:3}:
    \la{
      \tcrtexpr{\cenv}{\genv}{\heapty}{\const}{\tbase_{\const}}{\effempty}{\predmapempty}{\cenv}
      \label{er:a:4}
    }
    So, \is{\effect}{\effempty} and \is{\predmap}{\predmapempty}.

    Let \isdef{\loc}{\idx{\store}{\evar}}.

    We pick \isdef{\heapty'}{\upd{\heapty}{\loc}{\tbase_{\const}}}.
    By~\hyperref[rule:tcconst]{\tcconst} on \const under $\cenv'$:
    \la{
      \tcrtexpr{\cenv'}{\genv}{\heapty'}{\const}{\tbase_{\const}}{\effempty}
      {\underbrace{\predmapempty}_{\predmap'}}
      {\underbrace{\cenv'}_{\cenv_1'}}
      \label{er:a:10}
    }

    Let $\heapty_0$ and $\heap_0$ \st
    \is{\heapty}{\heaptyext{\heapty_0}{\loc}{\type_{\loc}}}
    and
    \is{\heap}{\heapext{\heap_0}{\loc}{\const}}.
    
    It holds that:
    \la{
      \is{\heapty'}{\envext{\heapty_0}{\loc}{\tbase_{\const}}}
      \label{er:a:6}
    }

    By applying Rule~\hyperref[rule:rtheapconst]{\rtheapconst} on~\ref{er:i} (on the part of $\heap_0$)
    and~\eqref{er:a:6}:
    \la{
      \tcheap{\genv}{\heapty'}{\heapext{\heap_0}{\loc}{\const}}
      \label{er:a:7}
    }
    which proves~\ref{er:a}.

    By~\eqref{er:a:4} we prove~\ref{er:b},~\ref{er:c} and~\ref{er:d}.

    Since $\cenv'$ and \cenv agree on all variables with the exception potentially of
    \evar, we limit the scope to \evar. By definition of $\cenv'$ it holds that:
    \[
    \begin{array}{rcl}
      \idx{\parens{\refine{\cenv_1'}{\predmap'}}}{\evar}
      & \equalwithlabel{\eqref{er:a:4}} & \idx{\cenv'}{\evar} \\
      & \equals & \idx{\parens{\compose{\heapty'}{\store}}}{\evar} \\
      & \equals & \idx{\heapty'}{\idx{\store}{\evar}} \\
      & \equals & \idx{\heapty'}{\loc} \\
      & \equalwithlabel{\eqref{er:a:6}} & \tbase_{\const} \\
      & \equalwithlabel{\eqref{er:a:3}} & \idx{\cenv_1}{\evar}  \\
      & \equals & \idx{\parens{\refine{\cenv_1}{\logforget{\predmap}{\evar}}}}{\evar}
    \end{array}
    \]

    The last line above holds since \evar is excluded from the set
    of variables that are refined by predmap.
    This proves the first case of \ref{er:e}.
    The second case for $\neg\predmap$ can be proven similarly.
    The third case is trivially true.


  \item  \hyperref[rule:rtectx]{\rtectx}:
   \la{
     \stepsconf{\angletriplet{\heap}{\stack}{\store}}{\underbrace{\ectxidx{\ectx}{\expr_0}}_{\expr}}
     {\angletriplet{\heap'}{\stack}{\store'}}{\ectxidx{\ectx}{\expr_0'}}
      \label{er:ectx:0}
    }

    By inverting Rule~\hyperref[rule:rtectx]{\rtectx} on~\eqref{er:ectx:0}:
    \la{
      \stepsconf{\angletriplet{\heap}{\stack}{\store}}{\expr_0}{\angletriplet{\heap'}{\stack}{\store'}}{\expr_0'}
      \label{er:ectx:1}
    }


    By induction hypothesis using~\ref{er:i},~\eqref{er:ectx:4} and~\eqref{er:ectx:1}
    there exists $\heapty'$ \st:
    \la{
      \tcheap{\genv}{\heapty'}{\heap'}
      \label{er:ectx:7}
      \\
      \tcrtexpr{\cenv'}{\genv}{\heapty'}{\expr_0'}{\type_0'}{\effect_0'}{\predmap_0'}{\cenv_0'}
      \label{er:ectx:8}
      \\      
      \issubtype{\type_0'}{\type_0}
      \label{er:ectx:9}
      \\
      \issubtype{\effect_0'}{\effect_0}
      \label{er:ectx:10}
      \\
      \imp{\logand{\isval{\expr_0'}}{\istruthy{\expr_0'}}}{\issubtype{\refine{\cenv_0'}{\predmap_0'}}{\refine{\cenv_0}{\predmap_0}}}
      \label{er:ectx:11a}
      \\
      \imp{\logand{\isval{\expr_0'}}{\isfalsy{\expr_0'}}}{\issubtype{\refine{\cenv_0'}{\neg\predmap_0'}}{\refine{\cenv_0}{\neg\predmap_0}}}
      \label{er:ectx:11b}
      \\
      \imp{\neg\isval{\expr_0'}}{\logand{
        \issubtype{\refine{\cenv_0'}{\predmap_0'}}{\refine{\cenv_0}{\predmap_0}}
      }{\issubtype{\refine{\cenv_0'}{\neg\predmap_0'}}{\refine{\cenv_0}{\neg\predmap_0}}}}
      \label{er:ectx:11c}
    }

    By~\eqref{er:ectx:7} we prove~\ref{er:a}.

    We examine cases on the form of \ectx and the value of $\expr_0$:
    \begin{itemize}

      \item \isequiv{\ectx}{\binand{\hole}{\expr_1}} and \isval{\expr_0'}.
        Let's also assume that \istruthy{\expr_0'}.
        (The case for \isfalsy{\expr_0'} is symmetrical.)

        It holds that:
        \la{
          \isequiv{\expr}{\binand{\expr_0}{\expr_1}}
          \label{er:ectx:54}
          \\
          \isequiv{\expr'}{\binand{\expr_0'}{\expr_1}}
        }

        By inverting \hyperref[rule:tcand]{\tcand} on~\ref{er:ii} using \eqref{er:ectx:54}:
        \la{
          \tcexpr{\cenv}{\genv}{\expr_0}{\type_0}{\effect_0}{\predmap_0}{\cenv_0}
          \label{er:ectx:4}
          \\
          \tcexpr{\refine{\cenv_0}{\predmap_0}}{\genv}{\expr_1}{\type_1}{\effect_1}{\predmap_1}{\cenv_{01}}
          \label{er:ectx:5}
          \\
          \is{\type}{\tjoin{\refine{\type_0}{\falsy}}{\type_1}}
          \label{er:ectx:5a}
          \\
          \is{\effect}{\effectconcat{\effect_0}{\effect_1}}
          \label{er:ectx:5b}
          \\
          \is{\predmap}{\logand{\parens{\logforget{\predmap_0}{\effect_1}}}{\predmap_1}}
          \label{er:ectx:5c}
          \\
          \is{\cenv_1}{\envjoin{\parens{\refine{\cenv_0}{\logneg{\predmap_0}}}}{\cenv_{01}}}
          \label{er:ectx:5d}
        }


        By Lemma~\ref{lemma:environment:weakening} on \eqref{er:ectx:11a} and \eqref{er:ectx:5}:
        \la{
          \tcrtexpr{\refine{\cenv_0'}{\predmap_0'}}{\genv}{\heapty'}{\expr_1}{\type_1'}{\effect_1'}{\predmap_1'}{\cenv_{01}'}
          \label{er:ectx:40}
          \\
          \issubtype{\type_1'}{\type_1}
          \label{er:ectx:41}
          \\
          \issubtype{\effect_1'}{\effect_1}
          \label{er:ectx:42}
          \\
          \issubtype{\refine{\cenv_{01}'}{\predmap_1'}}{\refine{\cenv_{01}}{\predmap_1}}
          \label{er:ectx:43}
          \\
          \issubtype{\refine{\cenv_{01}'}{\neg\predmap_1'}}{\refine{\cenv_{01}}{\neg\predmap_1}}
          \label{er:ectx:44}
        }

        By assumption~\ref{assumption:dead:code:and:true} using \eqref{er:ectx:40}:
        \la{
          \tcrtexpr{\underbrace{\refine{\cenv_0'}{\predmap_0'}}_{\cenv'}}
            {\genv}{\heapty'}{\binand{\vtrue}{\expr_1}}{\type_1'}{\effect_1'}
            {\underbrace{\predmap_1'}_{\predmap'}}
            {\underbrace{\cenv_{01}'}_{\cenv_1'}}
          \label{er:ectx:45}
        }

        By \eqref{er:ectx:45} we prove \ref{er:b}.

        By \eqref{er:ectx:41} and \eqref{er:ectx:5a} we prove \ref{er:c}.

        By \eqref{er:ectx:42} and \eqref{er:ectx:5b} we prove \ref{er:d}.

        By Lemma~\ref{lemma:environment:weakening:expr:refine}
        on \eqref{er:ectx:5} refining with \logforget{\predmap_0}{\effect_1}:
        \la{
          \tcexpr{\underbrace{\refine{\refine{\cenv_0}{\predmap_0}}{\logforget{\predmap_0}{\effect_1}}}_{\equiv\refine{\cenv_0}{\predmap_0}}}
            {\genv}{\expr_1}{\_}{\_}{\_}{\refine{\cenv_{01}}{\logforget{\predmap_0}{\effect_1}}}
          \label{er:ectx:46}
        }

        By Lemma~\ref{lemma:environment:weakening} on \eqref{er:ectx:11a}, \eqref{er:ectx:40} and \eqref{er:ectx:46}:
        \begin{align}
          & \quad & \issubtype{\cenv_{01}'}{\refine{\cenv_{01}}{\logforget{\predmap_0}{\effect_1}}}
          \\
          & \therefore &
          \issubtype{\refine{\cenv_{01}'}{\predmap_1'}}{\refine{\cenv_{01}}{\logforget{\predmap_0}{\effect_1}}}
          \label{er:ectx:48}
        \end{align}
        
        By \eqref{er:ectx:43} and \eqref{er:ectx:48}:
        {
        \def\arraystretch{1.3}
        \begin{longtable}{>{$}c<{$} >{$}c<{$} >{$}l<{$}}
          & \quad & \issubtype{\refine{\cenv_{01}'}{\predmap_1'}}{\refine{\cenv_{01}}{\logand{\logforget{\predmap_0}{\effect_1}}{\predmap_1}}}
          \\
          \stackrel{\eqref{er:ectx:45}}{\therefore} 
          &
          & \issubtype{\refine{\underline{\cenv_1'}}{\underline{\predmap'}}}{\refine{\cenv_{01}}{\logand{\logforget{\predmap_0}{\effect_1}}{\predmap_1}}}
          \\
          \stackrel{\eqref{er:ectx:5c}}{\therefore}
          &
          & \issubtype{\refine{\cenv_1'}{\predmap'}}{\refine{\cenv_{01}}{\underline{\predmap}}}
          \\
          \stackrel{\eqref{er:ectx:5d}}{\therefore}
          &
          & \issubtype{\refine{\cenv_1'}{\predmap'}}{\refine{\underline{\cenv_1}}{\predmap}}
        \end{longtable}
        }
        which proves the first part of \ref{er:e}.
        The case for $\neg\predmap$ is similar.

      \item \emph{The remaining cases are treated similarly.}

    \end{itemize}


  \item  \hyperref[rule:rtandtru]{\rtandtru}:
    \la{
      \stepsconf{\angletriplet{\heap}{\stack}{\store}}{\underbrace{\binand{\val_1}{\expr_2}}_{\expr}}
      {\angletriplet{\heap}{\stack}{\store}}{\expr_2}
      \label{er:and:1}
    }

    By inverting Rule~\hyperref[rule:rtandtru]{\rtandtru} on \eqref{er:and:1}:
    \la{
      \istruthy{\val_1}
      \label{er:and:3}
    }

    Due to~\eqref{er:and:1} judgment~\ref{er:ii} is of the form:
    \la{
      \tcrtexpr{\cenv}{\genv}{\heapty}{\binand{\val_1}{\expr_2}}
        {\underbrace{\tjoin{\refine{\type_1}{\falsy}}{\type_2}}_{\type}}
        {\underbrace{\effectconcat{\effect_1}{\effect_2}}_{\effect}}
        {\underbrace{\logand{\parens{\logforget{\predmap_1}{\effect_1}}}{\predmap_2}}_{\predmap}}
        {\underbrace{\envjoin{\refine{\cenv_{01}}{\logneg{\predmap_1}}}{\cenv_2}}_{\cenv_1}}
      \label{er:and:4}
    }

    By inverting Rule~\hyperref[rule:tcand]{\tcand} on \ref{er:ii}
    and simplifying by using Rules~\hyperref[rule:rt:tcloc]{\rttcloc} and \hyperref[rule:tcconst]{\tcconst}:
    \la{
      \tcrtexpr{\cenv}{\genv}{\heapty}{\val_1}{\type_1}
        {\underbrace{\effempty}_{\effect_1}}
        {\underbrace{\predmapempty}_{\predmap_1}}
        {\underbrace{\cenv}_{\cenv_{01}}}
      \label{er:and:5}
      \\
      \tcrtexpr{\cenv}{\genv}{\heapty}{\expr_2}{\type_2}{\effect_2}{\predmap_2}{\cenv_2}
      \label{er:and:7}
    }

    Store and heap do not evolve, \ie \is{\store'}{\store} and \is{\heap'}{\heap}.

    We pick \isdef{\heapty'}{\heapty} and so by \eqref{er:and:7}:
    \la{
      \is{\cenv'}{\cenv}
      \\
      \is{\cenv_1'}{\cenv_2}
      \label{er:and:10}
      \\
      \is{\predmap'}{\predmap_2}
      \label{er:and:11}
    }
    
    By~\ref{er:i}:
    \la{
      \tcheap{\genv}{\heapty'}{\heap'}
    }
    which proves \ref{er:a}.

    By \eqref{er:and:7} we prove \ref{er:b}.

    It holds that 
    \issubtype{\type_2}{
      \isequivwithlabel{\parens{\tjoin{\refine{\type_1}{\falsy}}{\type_2}}}
        {\eqref{er:and:4}}{\type}}, 
    which proves~\ref{er:c}.

    It holds that 
    \issubtype{\effect_2}{
      \isequivwithlabel{\effectconcat{\effect_1}{\effect_2}}{\eqref{er:and:4}}{\effect}},
    which proves~\ref{er:d}.

    By~\eqref{er:and:4} and \eqref{er:and:5}:
    \la{
      \is{\predmap}{\predmap_2}
      \label{er:and:12}
    }

    Finally, it holds that
    \[
    \begin{array}{rcl}
      \refine{\cenv_1'}{\predmap'} 
      & \equalwithlabel{\eqref{er:and:10}} & \refine{\underline{\cenv_2}}{\predmap'}
      \\
      & \equalwithlabel{\eqref{er:and:11}} & \refine{\cenv_2}{\underline{\predmap_2}}
      \\
      & \subtypesym & \refine{\parens{\envjoin{\underline{\refine{\cenv_{01}}{\logneg{\predmap_1}}}}{\cenv_2}}}{\predmap_2}
      \\
      & \equalwithlabel{\eqref{er:and:4}} & \refine{\underline{\cenv_1}}{\predmap_2}
      \\
      & \equalwithlabel{\eqref{er:and:12}} & \refine{\cenv_1}{\underline{\predmap}}
    \end{array}
    \]

    which proves the first case of~\ref{er:e}. The second case is proven similarly.

\end{itemize}

\noindent\emph{The rest of the cases are handled similarly.}
\end{proof}

\begin{lemma}[Preservation of Typing by Statement Reduction]
\label{lemma:stmt:preservation}
Typing is preserved over the reduction of a statement that
preserves the state of the stack. That is, for an initial
runtime state
\isdef{\rtstate}{\angletriplet{\heap}{\stack}{\store}},
a target state
\isdef{\rtstate'}{\angletriplet{\heap'}{\stack}{\store'}}
if, under a heap typing \heapty:
\begin{enumerate}[label=(\roman*)]
  \item \tcheap{\genv}{\heapty}{\heap}
    \label{st:r:i}
  \item \tcrtstmt{\cenv}{\genv}{\heapty}{\stmt}{\effect}{\cenv_1} 
    \label{st:r:ii}
  \item \stepsconf{\rtstate}{\stmt}{\rtstate'}{\stmt'}
    \label{st:r:iii}
\end{enumerate}
where
\isdef{\cenv}{\compose{\heapty}{\store}},
then there exist $\heapty'$ \st:
\begin{enumerate}[label=(\alph*)]
  \item \tcheap{\genv}{\heapty'}{\heap'}
    \label{st:r:a}
  \item \tcrtstmt{\cenv'}{\genv}{\heapty'}{\stmt'}{\effect'}{\cenv_1'} 
    \label{st:r:b}
  \item \issubtype{\effect'}{\effect}
    \label{st:r:c}
  \item \issubtype{\cenv_1'}{\cenv_1}
    \label{st:r:d}
\end{enumerate}
where \isdef{\cenv'}{\compose{\heapty'}{\store'}}.
\end{lemma}

\begin{proof}
  \emph{Similar to proof of Lemma~\ref{lemma:expr:preservation}.}
  \todo{Do case of if-then-else.}






\end{proof}

\begin{theorem}[Subject Reduction]
\label{theorem:flow:preservation}
Typing is preserved over expression reduction. 
Formally, if
\begin{enumerate}[label=(\roman*)]
  \item \tcrtconfexpr{\genv}{\heapty}{\rtstate}{\expr}{\type}
    \label{sr:i}
  \item \stepsconf{\rtstate}{\expr}{\rtstate'}{\exprbody'}
    \label{sr:ii}
\end{enumerate}
then there exists $\heapty'$ \st
\begin{enumerate}[label=(\alph*)]
  \item \tcrtconfexpr{\genv}{\heapty'}{\rtstate'}{\exprbody'}{\type'}
    \label{sr:a}
  \item \issubtype{\type'}{\type}
    \label{sr:b}
\end{enumerate}
\begin{proof}
Let
\la{
  \rtstate \equiv \rttriplet{\heap}{\stack}{\store}
  \label{sr:22}
  \\
  \rtstate' \equiv \rttriplet{\heap'}{\stack'}{\store'}
  \label{sr:23}
}

By inverting Rule~\hyperref[rule:rtconfbody]{\rtconfbody} on~\ref{sr:i}:
\la{
  \tcheap{\genv}{\heapty}{\heap}
  \label{sr:24}
  \\
  \jstore{\heapty}{\store}{\cenv}
  \label{sr:7}
  \\
  \tcrtbody{\cenv}{\genv}{\heapty}{\expr}{\type_{\expr}}{\effect_{\expr}}{\cenv_{\expr}}
  \label{sr:8}
  \\ 
  \isheaptyoverrides{\heapty_{\stack}}{\cenv_{\expr}}{\store}{\heapty}
  \label{sr:213}
  \\
  \tcstack{\genv}{\heapty_{\stack}}{\stack}{\ectxidx{\type}{\type_{\expr}}}
  \label{sr:6}
}

By induction on the derivation of~\ref{sr:ii}:

\begin{itemize}
  \item 
    \hyperref[rule:rtasgn]{\rtasgn},
    \hyperref[rule:rtarrow]{\rtarrow},
    \hyperref[rule:rtpredvar]{\rtpredvar},
    \hyperref[rule:rtandtru]{\rtandtru},
    \hyperref[rule:rtandfls]{\rtandfls},
    \hyperref[rule:rtortru]{\rtortru},
    \hyperref[rule:rtorfls]{\rtorfls},
    \hyperref[rule:rtneg]{\rtneg},
    \hyperref[rule:rtlet]{\rtlet},
    \hyperref[rule:rtiftru]{\rtiftru},
    \hyperref[rule:rtiffls]{\rtiffls},
    and \hyperref[rule:rtskip]{\rtskip}
    do not evolve the stack, so can be proven by use of
    Lemma~\ref{lemma:expr:preservation}.

  \item  \hyperref[rule:rtcall]{\rtcall}:
    \la{
      \stepsconf{\rttriplet{\heap}{\stack}{\store}}
      {\underbrace{\ectxidx{\ectx}{\efuncall{\loc}{\val}}}_{\expr}}
        {\rttriplet
          {\heap'}
          {\stack'}
          {\store'}
        }
        {\underbrace{\body{\stmt_0}{\expr_0}}_{\exprbody'}}
      \label{sr:1}
    }
    where \fresh{\loc'}. By inverting Rule~\hyperref[rule:rtcall]{\rtcall} on \eqref{sr:1}:
    \la{
      \is{\idx{\heap}{\loc}}{\storearrowshort{\store_0}{\evar}{\exprbody'}}
      \label{sr:2}
      \\
      \is{\many{\evar_i}}{\locals{\exprbody'}}
      \label{sr:2a}
      \\
      \is{\heap'}{\heapcons{\heapext{\heap}{\loc'}{\val}}{\many{\heapbinding{\loc_i}{\vundef}}}}
      \label{sr:2b}
      \\
      \is{\stack'}{\stackcons{\stack}{\store}{\ectx}}
      \label{sr:2c}
      \\
      \is{\store'}{\storecons{\storeext{\store_0}{\evar}{\loc'}}{\many{\storebinding{\evar_i}{\loc_i}}}}
      \label{sr:2d}
    }

    Due to \eqref{sr:1}, judgment \eqref{sr:8} is of the form:
    \la{
      \tcrtbody{\cenv}{\genv}{\heapty}{\ectxidx{\ectx}{\efuncall{\loc}{\val}}}{\type_{\expr}}{\effect_{\expr}}{\cenv_{\expr}}
      \label{sr:12}
    }

    By Lemma~\ref{lemma:decomp:ectx} on \eqref{sr:12} and given that the
    predicate mapping associated with a function call is empty:
    \la{
      \tcrtexpr{\cenv}{\genv}{\heapty}
        {\efuncall{\loc}{\val}}{\type_c}
        {\effect_c}{\predmap_c}{\cenv_3}
      \label{sr:13}
      \\
      \tcrtexpr{\cenv_3}{\genv}{\heapty}{\ectx}
        {\ectxidx{\type_{\expr}}{\type_c}}
        {\ectxidx{\effect_{\expr}}{\effect_c}}
        {\ectxidx{\predmap}{\predmap_c}}
        {\cenv_{\expr}}
      \label{sr:65}
    }

    By inverting Rule~\hyperref[rule:tccall]{\tccall} on~\eqref{sr:13}:
    \la{
      \tcrtexpr{\cenv}{\genv}{\heapty}{\loc}{\type_{\loc}}{\effempty}{\predmapempty}{\cenv_1}
      \label{sr:16}
      \\
      \tcrtexpr{\cenv_1}{\genv}{\heapty}{\val}{\type_{\val}}{\effempty}{\predmapempty}{\cenv_2}
      \label{sr:17}
      \\
      \issubtype{\type_{\loc}}{\tarrow{\type_{\val}}{\effect_c}{\type_c}}
      \label{sr:19}
      \\
      \is{\cenv_3}{\erasevarto{\cenv_2}{\genv}{\effect_c}}
      \label{sr:20}
    }

    By Rules~\hyperref[rule:rt:tcloc]{\rttcloc} and \hyperref[rule:tcconst]{\tcconst}
    on \eqref{sr:16} and on \eqref{sr:17}:
    \la{
      \is{\cenv}{\is{\cenv_1}{\cenv_2}}
      \label{sr:66}
    }

    So \eqref{sr:17} becomes:
    \la{
      \tcrtexpr{\cenv}{\genv}{\heapty}{\val}{\type_{\val}}{\effempty}{\predmapempty}{\cenv}
      \label{sr:35}
    }

    By inverting Rule~\hyperref[rule:rtheapfun]{\rtheapfun} on \eqref{sr:24}
    using \eqref{sr:2}:
    \la{
      \tcheap{\genv}{\heapty}{\heap_0}
      \label{sr:61}
      \\
      \idxis{\heapty}{\loc}{\type_{\loc}}
      \\
      \jstore{\heapty}{\store_0}{\cenv_0}
      \label{sr:26}
      \\
      \tcexpr{\cenv_0}{\genv}{\arrow{\evar}{\stmt_0}{\expr_0}}{\type_{\loc}}{\effempty}{\predmapempty}{\cenv_0'}
      \label{sr:36}
    }
    where
    \la{
      \isequiv{\heap}{\heapext{\heap_0}{\loc}{\storearrowshort{\store_0}{\evar}{\exprbody'}}}
    }

    By inverting Rule~\hyperref[rule:tcfun]{\tcfun} on \eqref{sr:36}:
    \la{
      \tcrtbody
      {\underbrace{\envcons{\envext{\eraseto{\genv}{\cenv_0}}{\evar}
      {\type_{\evar}}}{\many{\binding{\evar_i}{\tvoid}}}}_{\cenv_{0.1}}}{\genv}{\heapty}
      {\body{\stmt_0}{\expr_0}}{\type_0}{\effect_0}{\cenv_{0.2}}
      \label{sr:44}
      \\
      \isequiv{\type_{\loc}}{\tarrow{\type_{\evar}}{\logforget{\effect_0}{\evar,\many{\evar_i}}}{\type_0}}
      \label{sr:38}
    }
    where \many{\evar_i} are the local variables defined in $\stmt_0$.

    By~\eqref{sr:19} and~\eqref{sr:38}:
    \la{
    \issubtype{\tarrow{\type_{\evar}}{\logforget{\effect_0}{\evar,\many{\evar_i}}}{\type_0}}
        {\tarrow{\type_{\val}}{\effect_c}{\type_c}}
      \label{sr:183}
    }

    By subtyping decomposition on \eqref{sr:183}:
    \la{
      \issubtype{\type_{\val}}{\type_{\evar}}
      \label{sr:184a}
      \\
      \issubtype{\type_0}{\type_c}
      \\
      \issubtype{\logforget{\effect_0}{\evar,\many{\evar_i}}}{\effect_c}
      \label{sr:184}
    }

   \spara{After the reduction step},
    we pick:
    \la{
      \is{\heapty'}{\heaptyexts{\heaptyext{\heapty}{\loc'}{\type_{\val}}}{\loc_i}{\tvoid}}
      \label{sr:32}
    }

    The body \is{\exprbody'}{\body{\stmt_0}{\expr_0}} is checked under the environment
    produced by store $\store_0$ and heap typing $\heapty'$.
    $\heapty'$ coincides with $\heapty$ on their common domain $\envdom{\store}$,
    so:
    \la{
      \cenv_0 \equals \compose{\heapty'}{\store_0} \equals \compose{\heapty}{\store_0}
      \label{sr:141}
    }

    We extend the store $\store_0$ with a binding from $\evar$ to $\loc'$, and
    from every variable declared in the body $\stmt_0$ to \tvoid, resulting 
    in the following environment:
    \la{
      \cenv' \equals \envexts{\envext{\cenv_0}{\evar}{\type_{\val}}}{\evar_i}{\tvoid}
      \equals \compose{\heapty'}{\underbrace{\envexts{\envext{\store_0}{\evar}{\loc'}}{\evar_i}{\loc_i}}_{\store'}}
      \label{sr:33}
    }

    By Lemma~\ref{lemma:heapty:weakening:exprbody} on \eqref{sr:35}:
    \la{
      \tcrtexpr{\cenv}{\genv}{\heapty'}{\val}{\type_{\val}}{\effempty}{\predmapempty}{\cenv}
      \label{sr:67}
    }

    By applying Rules~\hyperref[rule:rtheaploc]{\rtheaploc} and \hyperref[rule:rtheapconst]{\rtheapconst} using
    \eqref{sr:24}, \eqref{sr:32} and \eqref{sr:67}:
    \la{
      \tcheap{\genv}{\heapty'}{\heap'}
      \label{sr:39}
    }

    By definition of \erasesym and~\eqref{sr:184a}:
    \la{
      \issubtype
        {\underbrace{\envexts{\envext{\cenv_0}{\evar}{\type_{\val}}}{\evar_i}{\tvoid}}_{\cenv'}}
        {\underbrace{\envexts{\envext{\eraseto{\genv}{\cenv_0}}{\evar}{\type_{\evar}}}{\evar_i}{\tvoid}}_{\cenv_{0.1}}}
      \label{sr:121}
    }

    By Lemma~\ref{lemma:environment:weakening:expr} on \eqref{sr:44} and \eqref{sr:121}, 
    and using the extended heap typing $\heapty'$:
    \la{
      \tcrtbody{\cenv'}{\genv}{\heapty'}{\body{\stmt_0}{\expr_0}}{\type_0'}{\effect_0'}{\cenv_2'}
      \label{sr:124}
      \\
      \issubtype{\type_0'}{\type_0}
      \label{sr:125a}
      \\
      \issubtype{\cenv_2'}{\cenv_{0.2}}
      \label{sr:125b}
      \\
      \issubtype{\effect_0'}{\effect_0}
      \label{sr:125c}
    }






    \spara{Stack \is{\stack'}{\stackcons{\stack}{\store}{\ectx}}} 
    is checked under a heap typing:
    \la{
      \isdef{\heapty_{\stackcons{\stack}{\store}{\ectx}}'}
      {\heaptyoverrides{\cenv_2'}{\store_0}{\heapty'}}
      \label{sr:161}
    }
    
    \spara{Evaluation context \ectx} is checked under an environment:
    \la{
      \jstore{\heapty_{\stackcons{\stack}{\store}{\ectx}}'}{\store}{\cenv_3'}
      \label{sr:162}
    }

    Let \evarset and $\evarset_0$ be the domains of \store and $\store_0$:
    \la{
      \evarset   \defeq \envdom{\store}
      \label{sr:163}
      \\
      \evarset_0 \defeq \envdom{\store_0}
      \label{sr:164}
    }

    Since \is{\envdom{\cenv_3}}{\evarset}, we can examine $\cenv_3$ in 
    two parts based on whether an element $\evar$ in $\evarset$, also 
    belongs to $\evarset_0$ or not:
    \la{
      \isequiv{\cenv_3}{\envcons{\restrict{\cenv_3}{\intersection{\evarset_0}{\evarset}}}}
                          {\restrict{\cenv_3}{\difference{\evarset}{\evarset}_0}}
      \label{sr:165}
    }

    We similarly examine $\cenv_3'$ into two parts:
    (i)~the closure environment $\cenv_2'$ at the end of the function body, and
    (ii)~the part of the environment at the call-site that is 
    not part of the closure environment and so retains the typing from before the function call:
    \la{
      \isequiv{\cenv_3'}{\envcons{\restrict{\cenv_2'}{\intersection{\evarset_0}{\evarset}}}}
                           {\restrict{\cenv}{\difference{\evarset}{\evarset}_0}}
      \label{sr:166}
    }

    We examine the two non-overlapping domains separately:
    \begin{itemize}

      \item  \intersection{\evarset_0}{\evarset}.
        By restricting~\eqref{sr:33} to \intersection{\evarset_0}{\evarset}:
        \begin{align}
          & 
          & \restrict{\cenv'}{\intersection{\evarset_0}{\evarset}}
          & \equals \restrict{\envexts{\envext{\cenv_0}{\evar}{\type_{\val}}}{\evar_i}{\tvoid}}{\intersection{\evarset_0}{\evarset}}
          && \equals
          \compose{\heapty'}
          {\bparens{\restrict{\envexts{\envext{\store_0}{\evar}{\loc'}}{\evar_i}{\loc_i}}{\intersection{\evarset_0}{\evarset}}}}
          \label{sr:200}
        \\
          & \therefore 
          & \restrict{\cenv'}{\intersection{\evarset_0}{\evarset}}
          & \equals
            \restrict{\cenv_0}{\intersection{\evarset_0}{\evarset}}
          && \equals
            \compose{\heapty'}{\restrict{\store_0}{\intersection{\evarset_0}{\evarset}}}
          \label{sr:201}
        \end{align}


        Note that due to $\alpha$-renaming every variable is uniquely defined. 
        Therefore, each variable is bound to the same location in a store that
        contains it. In particular, for $\store_0$ and \store it holds that:
        \la{
          \is{\restrict{\store_0}{\intersection{\evarset_0}{\evarset}}}
             {\restrict{\store}{\intersection{\evarset_0}{\evarset}}}
          \label{sr:202}
        }

        By restricting~\eqref{sr:7} to \intersection{\evarset_0}{\evarset}:
        \la{
          \jstore{\heapty}
            {\restrict{\store}{\intersection{\evarset_0}{\evarset}}}
            {\restrict{\cenv}{\intersection{\evarset_0}{\evarset}}}
          \label{sr:203}
        }

        By combining~\eqref{sr:66},~\eqref{sr:201},~\eqref{sr:202} and~\eqref{sr:203}:
        \la{
          \is{\restrict{\cenv'}{\intersection{\evarset_0}{\evarset}}}
             {\restrict{\cenv_2}{\intersection{\evarset_0}{\evarset}}}
          \label{sr:167}
        }


        Effect $\effect_c$ is concrete so it can be interpreted as a set of variables.
        We split the set \intersection{\evarset_0}{\evarset} in the following:
        \la{
          \is{\intersection{\evarset_0}{\evarset}}{\concat
            {\underbrace{\intersection{\intersection{\evarset_0}{\evarset}}{\effect_c}}_{\evarset_{\effect}}}
            {\underbrace{\difference{\bparens{\intersection{\evarset_0}{\evarset}}}{\effect_c}}_{\evarset_{\overline{\effect}}}}
          }
        }

        We examine each part separately.

        \begin{itemize}

          \item  $\evarset_{\effect}$. 
            We first restrict~\eqref{sr:20} to domain $\effect_c$
            (a concrete effect interpreted as a set): 
            \la{
              \is{\restrict{\cenv_3}{\effect_c}}{\restrict{\erasevarto{\cenv_2}{\genv}{\effect_c}}{\effect_c}}
              \label{sr:174}
            }
            
            By definition of \erasesym,~\eqref{sr:174} can be written as:
            \la{
              \is{\restrict{\cenv_3}{\effect_c}}{\restrict{\genv}{\effect_c}}
              \label{sr:175}
            }

            By definition of \genv it holds that:
            \la{
              \issubtype{\cenv_2'}{\genv}
              \label{sr:177a}
            }

            By~\eqref{sr:175} and~\eqref{sr:177a}:
            \la{
              \issubtype{\restrict{\cenv_2'}{\evarset_{\effect}}}{\restrict{\cenv_3}{\evarset_{\effect}}}
              \label{sr:177b}
            }

          \item  $\evarset_{\overline{\effect}}$.
            By definition of \erasesym:
            \la{
              \is{\restrict{\erasevarto{\cenv_2}{\genv}{\effect_c}}{\difference{\bparens{\intersection{\evarset_0}{\evarset}}}{\effect_c}}}
                {\restrict{\cenv_2}{\difference{\bparens{\intersection{\evarset_0}{\evarset}}}{\effect_c}}}
              \label{sr:179a}
            }
            since the binding for variables not in $\effect_c$ will not be affected by the erasure.

            By~\eqref{sr:20} and~\eqref{sr:179a}:
            \la{
              \is
                {\restrict{\cenv_2}{\difference{\bparens{\intersection{\evarset_0}{\evarset}}}{\effect_c}}}
                {\restrict{\cenv_3}{\difference{\bparens{\intersection{\evarset_0}{\evarset}}}{\effect_c}}}
              \label{sr:179b}
            }

            By~\eqref{sr:184} and~\eqref{sr:125c} (interpreting concrete effects as sets):
            \la{
              \issubseteq{\effect_0'}{\effect_c}
              \label{sr:176}
            }


            %
            By Lemma~\ref{lemma:non:effect} on~\eqref{sr:124}:
            \la{
              \issubtype{\restrict{\cenv_2'}{\difference{\bparens{\intersection{\evarset_0}{\evarset}}}{\effect_0'}}}
                        {\restrict{\cenv'}{\difference{\bparens{\intersection{\evarset_0}{\evarset}}}{\effect_0'}}}
              \label{sr:179c}
            }

            By \eqref{sr:176} and \eqref{sr:179c}:
            \la{
              \issubtype{\restrict{\cenv_2'}{\difference{\bparens{\intersection{\evarset_0}{\evarset}}}{\effect_c}}}
                        {\restrict{\cenv'}{\difference{\bparens{\intersection{\evarset_0}{\evarset}}}{\effect_c}}}
              \label{sr:179d}
            }


            By~\eqref{sr:179d} and~\eqref{sr:167}:
            \la{
              \issubtype{\restrict{\cenv_2'}{\difference{\bparens{\intersection{\evarset_0}{\evarset}}}{\effect_c}}}
                        {\restrict{\cenv_2}{\difference{\bparens{\intersection{\evarset_0}{\evarset}}}{\effect_c}}}
              \label{sr:181}
            }

            By definition of $\evarset_{\overline{\effect}}$,~\eqref{sr:181} can be written as
            \la{
              \issubtype{\restrict{\cenv_2'}{\evarset_{\overline{\effect}}}}
                        {\restrict{\cenv_2}{\evarset_{\overline{\effect}}}}
              \label{sr:182}
            }
        \end{itemize}

      \item  \difference{\evarset}{\evarset_0}. We follow a similar reasoning to above
        restricting the difference \difference{\evarset}{\evarset_0} to variables
        contained in $\effect_c$ or not. We examine the cases:
        \begin{itemize} 
          \item Restrict to $\effect_c$. By~\eqref{sr:20}:
            \la{
              \is{\restrict{\cenv_3}{\intersection{\parens{\difference{\evarset}{\evarset_0}}}{\effect_c}}}
                {\restrict{\erasevarto{\cenv_2}{\genv}{\effect_c}}{\intersection{\parens{\difference{\evarset}{\evarset_0}}}{\effect_c}}}
              \label{sr:300}
            }

            By definition of \erasesym the above becomes:
            \la{
              \is{\restrict{\cenv_3}{\intersection{\parens{\difference{\evarset}{\evarset_0}}}{\effect_c}}}
                {\restrict{\cenv_2}{\intersection{\parens{\difference{\evarset}{\evarset_0}}}{\effect_c}}}
              \label{sr:301}
            }

          \item Restrict to $\overline{\effect_c}$. By~\eqref{sr:20}:
            \la{
              \is{\restrict{\cenv_3}{\difference{\parens{\difference{\evarset}{\evarset_0}}}{\effect_c}}}
                {\restrict{\erasevarto{\cenv_2}{\genv}{\effect_c}}{\difference{\parens{\difference{\evarset}{\evarset_0}}}{\effect_c}}}
              \label{sr:302}
            }

            By definition of \erasesym the above becomes:
            \la{
              \is{\restrict{\cenv_3}{\difference{\parens{\difference{\evarset}{\evarset_0}}}{\effect_c}}}
                {\restrict{\genv}{\difference{\parens{\difference{\evarset}{\evarset_0}}}{\effect_c}}}
              \label{sr:303}
            }

        \end{itemize}

        In either case it holds that:
        \la{
          \issubtype
            {\restrict{\cenv}{\difference{\evarset}{\evarset_0}}}
            {\restrict{\cenv_3}{\difference{\evarset}{\evarset_0}}}
          \label{sr:210}
        }

    \end{itemize}

    By~\eqref{sr:165} it holds that:
    \la{
      \isequiv{\bparens{\envconsthree
        {\restrict{\cenv_3}{\evarset_{\effect}}}
        {\restrict{\cenv_3}{\evarset_{\overline{\effect}}}}
        {\restrict{\cenv_3}{\difference{\evarset}{\evarset}_0}}
      }}{\cenv_3}
      \label{sr:250}
    }

    By~\eqref{sr:177b} and~\eqref{sr:250}:
    \la{
      \issubtype{\bparens{\envconsthree
        {{\color{darkblue}{\restrict{\cenv_2'}{\evarset_{\effect}}}}}
        {\restrict{\cenv_3}{\evarset_{\overline{\effect}}}}
        {\restrict{\cenv_3}{\difference{\evarset}{\evarset}_0}}
      }}{\cenv_3}
      \label{sr:251}
    }

    By Lemma~\ref{lemma:erase:subtype} on~\eqref{sr:251} and~\eqref{sr:19}:
    \la{
      \issubtype{\bparens{\envconsthree
        {\restrict{\cenv_2'}{\evarset_{\effect}}}
        {{\color{darkblue}{\restrict{\cenv_2}{\evarset_{\overline{\effect}}}}}}
        {\restrict{\cenv_3}{\difference{\evarset}{\evarset}_0}}
      }}{\cenv_3}
      \label{sr:252}
    }

    By~\eqref{sr:252} and~\eqref{sr:182}:
    \la{
      \issubtype{\bparens{\envconsthree
        {\restrict{\cenv_2'}{\evarset_{\effect}}}
        {{\color{darkblue}{\restrict{\cenv_2'}{\evarset_{\overline{\effect}}}}}}
        {\restrict{\cenv_3}{\difference{\evarset}{\evarset}_0}}
      }}{\cenv_3}
      \label{sr:253}
    }

    By~\eqref{sr:210} and~\eqref{sr:253}:
    \la{
      \issubtype{\underbrace{\bparens{\envconsthree
        {\restrict{\cenv_2'}{\evarset_{\effect}}}
        {\restrict{\cenv_2'}{\evarset_{\overline{\effect}}}}
        {{\color{darkblue}{\restrict{\cenv}{\difference{\evarset}{\evarset}_0}}}}
      }}_{\cenv_3'}}{\cenv_3}
      \label{sr:254}
    }

    By Lemma~\ref{lemma:environment:weakening:expr}
    (extended for evaluation contexts)
    \todo{add the exact lemma?}
    on~\eqref{sr:254} and~\eqref{sr:65}:
    \la{
      \tcrtexpr{\cenv_3'}{\genv}{\heapty'}{\ectx}
        {\ectxidx{\type_{\expr}'}{\type_c}}
        {\ectxidx{\effect_{\expr}'}{\effect_c}}
        {\ectxidx{\predmap'}{\predmap_c}}
        {\cenv_{\expr}'}
      \label{sr:190}
      \\
      \issubtype{\type_{\expr}'}{\type_{\expr}}
      \label{sr:191a}
      \\
      \issubtype{\effect_{\expr}'}{\effect_{\expr}}
      \label{sr:191b}
      \\
      \issubtype{\cenv_{\expr}'}{\cenv_{\expr}}
      \label{sr:193}
    }

    We check the remaining stack \stack under a heap typing (see also~\eqref{sr:161}):
    \la{
      \isdef 
      {\heapty_{\stack}'}
      {\heaptyoverrides{\cenv_{\expr}'}{\store}
      {\underbrace{\heaptyoverrides{\cenv_2'}{\store_0}{\heapty'}}_{\heapty_{\stackcons{\stack}{\store}{\ectx}}'}}}
      \label{sr:194}
    }

    Let \locset and $\locset_0$ the ranges of \store and $\store_0$:
    \la{
      \isdef{\locset}{\envrng{\store}}
      \label{sr:260}
      \\
      \isdef{\locset_0}{\envrng{\store_0}}
      \label{sr:261}
    }

    We examine $\heapty_{\stack}'$ in the following subdomains that correspond to the
    three parts of the definition above:
    \begin{itemize}
      \item  \locset.
        By restricting~\eqref{sr:213} and~\eqref{sr:194}, respectively, to \locset
        (the first part of the override will always be selected):
        \la{
          \is{\restrict{\heapty_{\stack}}{\locset}}{\compose{\cenv_{\expr}}{\envinvert{\store}}}
          \label{sr:216}
          \\
          \is{\restrict{\heapty_{\stack}'}{\locset}}{\compose{\cenv_{\expr}'}{\envinvert{\store}}}
          \label{sr:217}
        }

        By substituting~\eqref{sr:216} and~\eqref{sr:217} in~\eqref{sr:193}:
        \la{
          \issubtype{\restrict{\heapty_{\stack}'}{\locset}}
            {\restrict{\heapty_{\stack}}{\locset}}
          \label{sr:224}
        }

      \item  \difference{\locset_0}{\locset}. By~\eqref{sr:213} and~\eqref{sr:26}:
        \la{
          \isoverset
            {\eqref{sr:213}}
            {\restrict{\heapty_{\stack}}{\difference{\locset_0}{\locset}}}
            {\isoverset
              {\eqref{sr:26}}
              {\restrict{\heapty}{\difference{\locset_0}{\locset}}}
              {\compose{\cenv_0}{\envinvert{\parens{\difference{\store_0}{\store}}}}}
            }
          \label{sr:221}
        }

        By restricting~\eqref{sr:194} to \envdiff{\locset_0}{\locset} (the second part of the 
        override will always be selected):
        \la{
          \is{\restrict{\heapty_{\stack}'}{\difference{\locset_0}{\locset}}}
          {\compose{\cenv_2'}{\envinvert{\parens{\difference{\store_0}{\store}}}}}
          \label{sr:222}
        }

        By~\eqref{sr:181},~\eqref{sr:221} and~\eqref{sr:222}:
        \la{
          \issubtype
            {\restrict{\heapty_{\stack}'}{\difference{\locset_0}{\locset}}}
            {\restrict{\heapty_{\stack}}{\difference{\locset_0}{\locset}}}
          \label{sr:225}
        }

      \item  \setcomplement{\bparens{\setunion{\locset_0}{\locset}}}.
        By~\eqref{sr:213} and~\eqref{sr:161}:
        \la{
          \is{\restrict{\heapty_{\stack}}{\setcomplement{\bparens{\setunion{\locset_0}{\locset}}}}}
            {\restrict{\heapty}{\setcomplement{\bparens{\setunion{\locset_0}{\locset}}}}}
          \label{sr:218}
          \\
          \is{\restrict{\heapty_{\stack}}{\setcomplement{\bparens{\setunion{\locset_0}{\locset}}}}}
          {\restrict{\heapty'}{\setcomplement{\bparens{\setunion{\locset_0}{\locset}}}}}
          \label{sr:219}
        }
        
        By~\eqref{sr:218},~\eqref{sr:219} and~\eqref{sr:32}:
        \la{
          \issubtype
          {\restrict{\heapty_{\stack}'}{\setcomplement{\bparens{\setunion{\locset_0}{\locset}}}}}
          {\restrict{\heapty_{\stack}}{\setcomplement{\bparens{\setunion{\locset_0}{\locset}}}}}
          \label{sr:226}
        }

    \end{itemize}

    By composing~\eqref{sr:224},~\eqref{sr:225} and~\eqref{sr:226}:
    \la{
      \issubtype{\heapty_{\stack}'}{\heapty_{\stack}}
      \label{sr:227}
    }

    By Lemma~\ref{lemma:heapty:weakening} on~\eqref{sr:227},~\eqref{sr:191a} and~\eqref{sr:6}:
    \la{
      \tcstack{\genv}{\heapty_{\stack}'}{\stack}{\ectxidx{\type'}{\type_{\expr}'}}
      \label{sr:144}
      \\
      \issubtype{\type'}{\type}
      \label{sr:232}
    }

    By Rule~\hyperref[rule:rtstackcons]{\rtstackcons} on
    \eqref{sr:162},
    \eqref{sr:190},
    \eqref{sr:194} and
    \eqref{sr:144}
    we get the typing for \is{\stack'}{\stackcons{\stack}{\store}{\ectx}}:
    \la{
      \tcstack{\genv}{\heapty_{\stackcons{\stack}{\store}{\ectx}}'}
        {\stackcons{\stack}{\store}{\ectx}}
        {\ectxidx{\type'}{\type_c}}
      \label{sr:145}
    }



    By applying Rule~\hyperref[rule:rtconfstmt]{\rtconfbody} on
    \eqref{sr:39},
    \eqref{sr:33},
    \eqref{sr:124} and
    \eqref{sr:145}:
    \la{
      \tcrtconfexpr{\genv}{\heapty'}{\rttriplet{\heap'}{\stack'}{\store'}}
        {\body{\stmt_0}{\expr_0}}{\type'}
    }
    which proves \ref{sr:a}.

    By~\eqref{sr:232} we prove~\ref{sr:b}.

      \item \hyperref[rule:rtret]{\rtret}

    \emph{This case is treated similarly.}

    \end{itemize}

%

\end{proof}
\end{theorem}

\begin{theorem}[Progress -- Expressions and Function Bodies]
\label{theorem:flow:progress}
If
\shiftedalign{
  \tcrtconfexpr{\genv}{\heapty}{\rtstate}{\exprbody}{\type}
}
then one of the following holds:
\begin{enumerate}[label=(\alph*)]
  \item \exprbody is a value
    \label{pr:a}
  \item there exist $\rtstate'$ and $\exprbody'$ \st
    \stepsconf{\rtstate}{\exprbody}{\rtstate'}{\exprbody'}.
    \label{pr:b}
\end{enumerate}
\begin{proof}
Let
\la{
  \rtstate \equiv \rttriplet{\heap}{\stack}{\store}
  \label{pr:1}
}

We prove the desired by induction on the given derivation.
\begin{itemize}
  \item \hyperref[rule:rtconfbody]{\rtconfbody}:
    \la{
      \tcrtconfexpr{\genv}{\heapty}{\rtstate}{\expr}{\type}
      \label{pr:2}
    }

    By inverting Rule~\hyperref[rule:rtconfexpr]{\rtconfexpr} on \eqref{pr:2}:
    \la{
      \tcheap{\genv}{\heapty}{\heap}
      \label{pr:4}
      \\
      \is{\cenv}{\compose{\heapty}{\store}}
      \label{pr:6}
      \\
      \tcrtexpr{\cenv}{\genv}{\heapty}{\expr}{\type_{\expr}}{\effect}{\predmap}{\cenv'}
      \label{pr:7}
      \\
      \isheaptyoverrides{\heapty'}{\cenv'}{\store}{\heapty}
      \\
      \tcstack{\genv}{\heapty'}{\stack}{\ectxidx{\type}{\tvar_{\stack}}}
      \label{pr:5}
    }

    By induction on the derivation of~\eqref{pr:7}:
    \begin{itemize}
      \item \hyperref[rule:tcconst]{\tcconst}: This expression is already a
        value so \ref{pr:a} holds.

      \item \hyperref[rule:tccall]{\tccall}:
        \la{
          \tcrtexpr{\cenv}{\genv}{\heapty}{\underbrace{\efuncall{\loc}{\val}}_{\expr}}{\type_{\expr}}{\effect}{\predmap}{\cenv'}
          \label{pr:14}
        }

        By inverting Rule \hyperref[rule:tccall]{\tccall} on \eqref{pr:14}:
        \la{
          \tcrtexpr{\cenv}{\genv}{\heapty}{\loc}{\type_{\loc}}{\effempty}{\predmapempty}{\cenv_1}
          \label{pr:9}
          \\
          \tcrtexpr{\cenv_1}{\genv}{\heapty}{\val}{\type_{\val}}{\effempty}{\predmapempty}{\cenv_2}
          \label{pr:10}
          \\
          \issubtype{\type_{\loc}}{\tarrow{\type_{\val}}{\effect}{\type_{\expr}}}
          \label{pr:11}
          \\
          \is{\cenv'}{\erasevarto{\cenv_2}{\genv}{\effect}}
          \label{pr:12}
        }

        By inverting Rule~\hyperref[rule:rt:tcloc]{\tcloc} on \eqref{pr:9}:
        \la{
          \is{\cenv_1}{\cenv}
          \label{pr:16}
          \\
          \is{\idx{\heapty}{\loc}}{\type_{\loc}}
          \label{pr:28}
        }

        By Rule~\hyperref[rule:tcconst]{\tcconst} (or Rule~\hyperref[rule:rt:tcloc]{\tcloc})
        on \eqref{pr:10}:
        \la{
          \is{\cenv_2}{\cenv}
          \label{pr:18}
        }

        By \eqref{pr:4} for location \loc:
        \la{
          \tcheap{\genv}{\heapty}{\heapext{\heap_0}{\loc}{\hval}}
          \label{pr:19}
        }

        For some heap $\heap_0$ and heap value \hval.

        Next we prove that
        \is{\idx{\heap}{\loc}}{\storearrow{\store_0}{\evar}{\stmt_0}{\expr_{0}}} for some
        $\store_0$, $\stmt_0$ and $\expr_0$ by induction on the derivation of \eqref{pr:19}:
        \begin{itemize}
          \item  \hyperref[rule:rtheaploc]{\rtheaploc}:
            \la{
              \tcheap{\genv}{\heapty}{\heapext{\heap_0}{\loc}{\loc'}}
              \label{pr:20}
            }
            For some location $\loc'$ distinct from \loc.
            
            By inverting \hyperref[rule:rtheaploc]{\rtheaploc} on \eqref{pr:20}:
            \la{
              \tcheap{\genv}{\heapty}{\heap_0}
              \label{pr:21}
            }

            Let \is{\heap_0}{\heapext{\heap_0'}{\loc'}{\hval'}}. \eqref{pr:21}
            becomes:
            \la{
              \tcheap{\genv}{\heapty}{\heapext{\heap_0'}{\loc'}{\hval'}}
              \label{pr:22}
            }

            By induction hypothesis using \eqref{pr:22}:
            \la{
              \is{\idx{\heap_0}{\loc'}}{\storearrow{\store_0'}{\evar}{\stmt_0'}{\expr_0'}}
              \label{pr:23}
            }

          \item \hyperref[rule:rtheapconst]{\rtheapconst}:
            \la{
              \tcheap{\genv}{\heapty}{\heapext{\heap_0}{\loc}{\const}}
              \label{pr:24}
            }

            For some constant $\const$.
            By inverting \hyperref[rule:rtheapconst]{\rtheapconst} on
            \eqref{pr:24}:
            \la{
              \tcheap{\genv}{\heapty}{\heap_0}
              \label{pr:25}
              \\
              \idxis{\heapty}{\loc}{\tbase_{\vconst}}
              \label{pr:26}
            }

            The subtyping constraint~\eqref{pr:11} and \eqref{pr:26} lead to a contradiction.

          \item  \hyperref[rule:rtheapfun]{\rtheapfun}:
            \la{
              \tcheap{\genv}{\heapty}
              {\heapext{\heap_0}{\loc}{\storearrow{\store_0}{\evar}{\stmt_0}{\expr_0}}}
            \label{pr:30}
            }
            which proves the desired result immediately.

        \item  \hyperref[appendix:rule:rtheapobj]{\rtheapobj}:
          \emph{Similar to rule \hyperref[rule:rtheapconst]{\rtheapconst}.}

        \end{itemize}

        So there exist $\store_0$, $\stmt_0$ and $\expr_0$ \st:
        \la{
          \is{\idx{\heap}{\loc}}{\storearrow{\store_0}{\evar}{\stmt_0}{\expr_0}}
          \label{pr:31}
        }

        We pick:
        \la{
          \is{\heap'}{\heapexts{\heapext{\heap}{\loc_{\val}}{\val}}{\loc_i}{\vundef}}
          \label{pr:32}
          \\
          \is{\stack'}{\stackcons{\stack}{\store}{\hole}}
          \label{pr:33}
          \\
          \is{\store'}{\storeexts{\storeext{\store_0}{\evar}{\loc_{\val}}}{\evar_i}{\loc_i}}
          \label{pr:34}
        }
        where $\many{\evar_i}$ are the variables defined in the function body, and
        $\loc_{\val}$ and $\loc_{i}$ are fresh locations.

        By applying Rule~\hyperref[rule:rtcall]{\rtcall} using \eqref{pr:31}, \eqref{pr:32},
        \eqref{pr:33} and \eqref{pr:34}
        \la{
          \stepsconf{\rtstate}{\efuncall{\loc}{\val}}
          {\rttriplet{\heap'}{\stack'}{\store'}}{\body{\stmt_0}{\expr_0}}
          \label{pr:35}
        }
        which proves \ref{pr:b}.

      \item 
        \hyperref[rule:tcvar]{\tcvar},
        \hyperref[rule:tcassign]{\tcassign},
        \hyperref[rule:tcfun]{\tcfun},
        \hyperref[rule:tcand]{\tcand},
        \hyperref[rule:tcor]{\tcor},
        \hyperref[rule:tcnot]{\tcnot} and 
        \hyperref[rule:tcpred]{\tcpred} are straight-forward since they 
        only impose very minimal preconditions for the 
        respective transition to happen.

    \end{itemize}

  \item \hyperref[rule:rtconfstmt]{\rtconfstmt}
    Proved by applying Theorem~\ref{theorem:flow:progress:stmt} on the statement part 
    of the body.

\end{itemize}
\end{proof}
\end{theorem}

\begin{theorem}[Progress -- Statements]\label{theorem:flow:progress:stmt}
If
\begin{enumerate}[label=(\roman*)]
  \item \jrtconfstmt{\heapty}{\rtstate}{\stmt}{\cset}
    \label{pr:stmt:i}
  \item \cset is consistent
\end{enumerate}
then one of the following holds:
\begin{enumerate}[label=(\alph*)]
  \item \stmt is a irreducible form 
  \item there exists $\rtstate'$ and $\stmt'$ \st \stepsconf{\rtstate}{\stmt}{\rtstate'}{\stmt'}
\end{enumerate}
  \begin{proof}

    \emph{The proof is by induction on the derivation
    of~\ref{pr:stmt:i}.}




  \end{proof}
\end{theorem}


\begin{theorem}[Type Safety]
\label{theorem:flow:safety}
A well-typed program is either in normal form or reduces to another well typed
state.
\end{theorem}
\begin{proof}
  The proof follows by subsequent applications of Theorems
  \ref{theorem:flow:progress}, \ref{theorem:flow:progress:stmt} and 
  \ref{theorem:flow:preservation}.
\end{proof}

\end{document}